\documentclass[onecolumn,superscriptaddress,longbibliography]{revtex4-2}

\usepackage[margin=1in]{geometry}
\usepackage[utf8]{inputenc}
\usepackage[colorlinks=true, linkcolor=blue, citecolor=blue, urlcolor=blue]{hyperref}
\usepackage{amsmath,amsthm,amsfonts}
\usepackage{bbold}
\usepackage{parskip}
\usepackage[normalem]{ulem}
\usepackage{hyperref}
\usepackage[dvipsnames]{xcolor}
\usepackage[ruled]{algorithm2e}
\usepackage{qcircuit}
\usepackage{adjustbox}
\usepackage{thmtools}
\usepackage{thm-restate}
\usepackage{enumitem}
\usepackage{tabularx}

\usepackage{soul}

\newtheorem{thm}{Theorem}
\newtheorem{lem}[thm]{Lemma}

\newtheorem{defn}[thm]{Definition}
\newtheorem{prob}[thm]{Problem}

\newcommand{\ket}[1]{|#1\rangle}
\newcommand{\bra}[1]{\langle#1|}
\newcommand{\braket}[2]{\langle#1|#2\rangle}
\newcommand{\ketbra}[2]{|#1\rangle\!\langle#2|}

\newcommand{\norm}[1]{\left\lVert#1\right\rVert}

\newcommand{\lb}{\left(}
\newcommand{\rb}{\right)}
\newcommand{\tn}{\Tilde{N}}
\newcommand{\tx}{\Tilde{x}}

\newcommand{\prep}{\mathtt{PREP}}
\newcommand{\sel}{\mathtt{SEL}}
\newcommand{\swap}{\mathtt{SWAP}}
\newcommand{\Del}{D^{\, \tel}_{n,j}}
\newcommand{\qft}{\mathtt{QFT}}
\newcommand{\cnot}{\mathtt{CNOT}}
\newcommand{\tel}{\text{el}}
\newcommand{\nuc}{\text{nuc}}
\newcommand{\tcla}{\text{class}}
\newcommand{\hax}{\overline{x}}
\newcommand{\hap}{\overline{p}}
\newcommand{\has}{\overline{s}}

\SetKwInput{KwTime}{Runtime}

\newcommand{\UofT}{\affiliation{
Department of Computer Science, University of Toronto, Canada}}

\newcommand{\UofTP}{\affiliation{
Department of Physics, University of Toronto, Canada}}

\newcommand{\PNNL}{\affiliation{Pacific Northwest National Laboratory, Richland WA, USA}}

\newcommand{\CIFAR}{\affiliation{Canadian Institute for Advanced Studies, Toronto, Canada}}

\newcommand{\BIQuantum}{\affiliation{
Quantum Lab, Boehringer Ingelheim, 55218 Ingelheim am Rhein, Germany}}

\begin{document}

\title{Improved precision scaling for simulating coupled quantum-classical dynamics}

\author{Sophia Simon}
\email{sophia.simon@mail.utoronto.ca}
\UofTP

\author{Raffaele Santagati}
\email{raffaele.santagati@boehringer-ingelheim.com}
\BIQuantum

\author{Matthias Degroote}
\BIQuantum

\author{Nikolaj Moll}
\BIQuantum

\author{Michael Streif}
\BIQuantum

\author{Nathan Wiebe}
\email{nathan.wiebe@utoronto.ca}
\UofT
\PNNL
\CIFAR

\begin{abstract}
    We present a super-polynomial improvement in the precision scaling of quantum simulations for coupled classical-quantum systems in this paper. Such systems are found, for example, in molecular dynamics simulations within the Born-Oppenheimer approximation. By employing a framework based on the Koopman-von Neumann formulation of classical mechanics, we express the Liouville equation of motion as unitary dynamics and utilize phase kickback from a dynamical quantum simulation to calculate the quantum forces acting on classical particles. This approach allows us to simulate the dynamics of these classical particles without the overheads associated with measuring gradients and solving the equations of motion on a classical computer, resulting in a super-polynomial advantage at the price of increased space complexity. We demonstrate that these simulations can be performed in both microcanonical and canonical ensembles, enabling the estimation of thermodynamic properties from the prepared probability density. 
\end{abstract}

\maketitle


\section{Introduction} 
We are accustomed to thinking of nature in terms of binaries. Specifically, we often speak of dynamics as if they were either purely quantum or purely classical. In reality, many models of physical interest actually share features with both quantum and classical matter.  
As a particular example, molecular dynamics is often formulated in this way wherein the nuclei are assumed to follow Newton's equations, but the electrons follow the Schr\"odinger equation. In other cases, we may treat an electromagnetic field as a time-dependent classical field and the particles interacting with it as quantum. In both cases, neither a fully quantum model nor a fully classical model can be used to address the problem efficiently.
Quantum computers have long been known to provide, under reasonable complexity-theoretic conjectures, exponential advantages for simulating certain quantum systems~\cite{lloyd1996universal, kempe2006complexity, lanyon2010towards, aaronson2011computational}. Recently, this has been extended to show that quantum computers can offer exponential advantages even for systems that are classical~\cite{Babbush2023}. However, when examining systems that straddle this line, such as the molecular dynamics example considered above, the situation is not as clear because such simulation methods are comparatively underexplored.

Descriptions within the Born-Oppenheimer (BO) approximation, where the wavefunctions of the nuclei can be considered independent of the wavefunctions of the electrons, play an important role in the chemical and pharmaceutical industry. These are often used to compute thermodynamic quantities of the chemical systems under study, such as the entropy or the free energy~\cite{Tuckerman1992, Martyna1996,  Xu2023, Ries2022, Santagati2023}.  In fact, in classical computational chemistry, thermodynamic averages can be obtained by combining molecular dynamics (MD) simulations with the use of thermostats to go beyond microcanonical ensembles~\cite{Tuckerman2000, Nose1984, Bond1999NosePoincare}. MD simulations introduce another approximation on top of the BO approximation by treating the nuclei as classical particles but retaining the quantum description of the electrons.
Recent works have explored the study of molecular dynamics on fault-tolerant quantum computers, for example, via force calculations, to update the coordinates of the classical nuclei~\cite{Kassal2010, Ollitrault2020, Fedorov2021, Kuroiwa2022, Sokolov2021, Obrien2019, Obrien2022}. Some of these works go beyond the Born-Oppenheimer approximation. In contrast to approaches where the full system is treated quantum mechanically, BO models have some advantages, namely that the scale of the quantum and the classical dynamics are naturally separated and that classical noise and external forces can be easily applied to the system without needing an expensive fully quantum description. 

Recent work \cite{Obrien2022,Steudtner2023} analyzed the cost of computing interatomic forces on quantum computers, which can be used to update the classical nuclear positions on a classical computer iteratively. In those works, the cost of an $\epsilon$-approximate gradient evaluation was found to scale with the error tolerance $\epsilon$ as $O(1/\epsilon)$. This approach is even less practical if we consider that, for many practical applications, we need to compute the properties of the chemical systems in the canonical ensemble~\cite{Tuckerman2000, MANATHUNGA2022, Xu2023, Santagati2023}. A key problem in simulating molecular dynamics in the canonical ensemble is to prepare the classical distribution of the nuclei (their probability density function) according to the Boltzmann distribution at a certain temperature.  Current quantum computing methods do not allow for efficient encoding and evolution of the probability density of the classical degrees of freedom on the quantum computer. In the 1930s, Koopman and von Neumann (KvN) proposed an empirical formulation of classical mechanics that incorporated a Hilbert space consisting of complex and square-integrable wavefunctions $\psi_{\text{KvN}}(q,p)$ which depend on the position $q$ and momentum $p$ of the particle~\cite{Koopman1931}. These wavefunctions were understood as probability densities, $\rho(p,q) = \psi_{\text{KvN}}^*\psi_{\text{KvN}}$, of finding the particle in a specific configuration $(p,q)$ of the phase space ~\cite{Mauro2002KvN}. Both, $\rho$ and $\psi_{\text{KvN}}$, evolve according to the Liouville equation
\begin{equation}
    \partial_t \rho = - i L \rho,
\end{equation}
where $L$ is a Hermitian operator called the Liouvillian operator. This gives a natural encoding (and subsequent evolution) of classical probability densities into quantum states that follow Hamiltonian dynamics.

The KvN formalism has been recently exploited to propose new algorithms for simulating classical systems on quantum computers and for solving non-linear partial differential equations~\cite{Joseph2020, Jin2023}. In this work, we do not use quantum linear systems algorithms but explicitly evolve the positions and momenta in time. We provide a specific procedure for the time propagation along with the cost of the block-encodings involved, as well as a method for obtaining ensemble averages from our final state.

We implement the time evolution of a hybrid quantum-classical system, by propagating the discretized phase space density of $N$ classical nuclei that interact with $\tn$ electrons, which are treated quantumly using the first quantized simulation method of~\cite{Babbush2019, Su2021first_quant_sim}. We present this result for the microcanonical ensemble~\cite{Tuckerman2000} where the number of particles $N_{\text{tot}} = N + \tn$, volume $V$, and energy $E$ are conserved. By adding a thermostat to the evolution, which couples the nuclear phase space density to a classical heat bath, we can impose a constant temperature $T$, allowing us to perform the simulation in the canonical ensemble~\cite{Nose1984}. 
The asymptotic gate complexity of the Liouvillian simulation algorithm is in
\begin{equation}
    \widetilde{O} \lb \frac{N_{\text{tot}} \mu^{2+o(1)} t^{1+o(1)}}{\widetilde{\gamma} \, \widetilde{\delta} \, \epsilon^{o(1)}} \log \lb \frac{1}{\xi} \rb \rb
\end{equation}
where $\mu$ is an upper bound on the spectral norm of the Liouvillian operator, $t$ is the evolution time, $\widetilde{\gamma}$ is a lower bound on the spectral gap of the electronic Hamiltonian over all configurations visited during the simulation, 
$\widetilde{\delta}$ is a lower bound on the overlap of the initial electronic state with the target electronic state over all configurations visited during the simulation, $\epsilon$ is the desired simulation precision and $\xi$ is an upper bound on the failure probability. Our result  provides super-polynomially better scaling with $\epsilon$ than the existing approach of~\cite{Obrien2022}, which in turn suggests that the roadblocks previously identified for molecular dynamics therein may not be the obstacles that they previously were believed to be. Additionally, the asymptotic space complexity of the Liouvillian simulation algorithm is moderate, scaling linearly in $N_{\text{tot}}$ and logarithmically in all other simulation parameters including the grid spacing as well as the precision $\epsilon$.

To tackle the problem of computing thermodynamic quantities, we introduce a second algorithm which, given a quantum state encoding the discretized phase space density of the system together with the heat bath (e.g., obtained using the first algorithm), can output an estimate of the free energy of the system. In contrast to classical methods and previous approaches for MD simulations on quantum computers~\cite{Ollitrault2020, Fedorov2021, Kuroiwa2022, Sokolov2021}, our approach exploits a fully coherent state preparation of the classical probability density in the canonical ensemble, avoiding sampling and enabling the direct estimation of thermodynamic properties~\cite{Obrien2022, Santagati2023}.

The manuscript is structured as follows. In Section~\ref{sec:prelim}, we introduce the main concepts for coupled quantum-classical dynamics in the Liouvillian picture. We show that a thermostat can be used to prepare the canonical ensemble and that we can implement this simulation on a quantum computer by discretizing the phase space. Our main results are presented in Section~\ref{sec:main}, which includes precise statements of the computational problems as well as the asymptotic gate complexity of our algorithms to solve these problems. In Section~\ref{sec:algorithm}, we give an overview of our quantum algorithms for simulating Liouvillian dynamics and estimating the free energy of a quantum-classical system. We conclude with a brief discussion in Section~\ref{sec:conclusion}.
The proofs of the theorems and lemmas presented in the main text are given in Appendices~\ref{app:L_class} -- \ref{app:EulerCost}.

\section{Preliminaries}
\label{sec:prelim}

In this section, we provide an overview of the fundamental definitions, from the Liouvillian formalism to the mappings into qubit registers and the microcanonical and canonical ensembles, that are essential for the definition of the computational problems and the implementation of the algorithms.

\subsection{Liouvillian formalism}

The trajectories of $N$ classical particles in 3 dimensions are governed by Newton's equations of motion:
\begin{equation}
    F_{n,j} = m_n \Ddot{x}_{n,j}, \quad  n \in \{1,2, \dots, N \}, \quad j \in \{1,2,3 \},
\end{equation}
where $F_{n,j}$ is the $j$-th component of the force on the $n$-th particle, $\Ddot{x}_{n,j}$ is its acceleration and $m_n$ is its mass. In general, these differential equations are nonlinear and nonunitary, meaning that the time evolution of the positions of the particles cannot be directly implemented on a quantum computer. The Born-Oppenheimer approximation in molecular dynamics (MD) turns the Hamiltonian problem of jointly evolving nuclei and electrons into an example of such a nonlinear and nonunitary time evolution. 
We overcome the issue by working with the Liouvillian formulation of classical mechanics instead, which is centered around the phase space probability density $\rho \lb \{x_n\}, \{p_n\}, t \rb$ of the classical particles. The probability density depends on the positions $x_n \in \mathbb{R}^3$ and momenta $p_n \in \mathbb{R}^3$ of the particles as well as on time $t$. It is normalized according to $\int_{\mathbb{R}^{6N}}\rho \lb \{x_n\}, \{p_n\}, t \rb d\{x_n\}d\{p_n\} = 1$ and satisfies the following equation of motion:
\begin{equation}
    \frac{\partial \rho}{\partial t} = -i L \rho,
\label{eom}
\end{equation}
where $L$ is the Liouvillian operator defined below:
\begin{equation}
    L := -i\sum_{n=1}^{N} \sum_{j=1}^{3} \lb \frac{\partial H}{\partial {p_{n,j}}} \partial_{x_{n,j}} - \frac{\partial H}{\partial {x_{n,j}}} \partial_{p_{n,j}}\rb.
\label{liouvillian}
\end{equation}
Here, $H$ is the classical Hamiltonian of the system and $\partial_{x_{n,j}}$ ($\partial_{p_{n,j}}$) denotes the partial derivative with respect to the $j$-th position (momentum) component of the $n$-th particle.

The formal solution to Eq.~(\ref{eom}) reads
\begin{equation}
    \rho(t) = e^{-iLt}\rho(0).
\label{ev_rho}
\end{equation}
Note that $L$ is Hermitian, which implies that $e^{-iLt}$ is unitary. 
The similarities to quantum mechanics become even more apparent when employing the Koopman-von Neumann formulation of classical mechanics~\cite{Mauro2002KvN}. The idea is to introduce a complex wave function, $\psi_{\text{KvN}}\lb \{x_n\}, \{p_n\}, t \rb$, which evolves according to the Liouville equation just like $\rho \lb \{x_n\}, \{p_n\}, t \rb$:
\begin{equation}
   \frac{\partial \psi_{\text{KvN}}}{\partial t} = -i L \psi_{\text{KvN}}.
\end{equation}
The phase space density can then be recovered via the relation $\rho = \psi_{\text{KvN}}^*\psi_{\text{KvN}}$, which resembles the quantum-mechanical calculation of probabilities from amplitudes. This works out because the Liouvillian contains only first-order derivatives, $\partial_{x_{n,j}}$ and $\partial_{p_{n,j}}$, meaning that the product of two wave functions, each satisfying the Liouville equation, also satisfies the Liouville equation.
In contrast, the Schrödinger equation of quantum mechanics generically contains second-order derivatives, $\partial_{x_{n,j}}^2$, implying that the product of two solutions does not necessarily satisfy the Schrödinger equation.

Considering the Koopman-von Neumann wave function $\psi_{\text{KvN}}$ rather than the phase space density $\rho$ has one significant advantage: $\psi_{\text{KvN}}$ can be easily encoded on a quantum computer since it has the same properties as a ``true'' quantum wave function. For example, while $\rho$ needs to be real-valued and positive, $\psi_{\text{KvN}}$ can take on complex values. Furthermore, $\psi_{\text{KvN}}$ is normalized according to the 2-norm, i.e. $\int_{\mathbb{R}^{6N}} |\psi_{\text{KvN}}|^2 d\{x_n\}d\{p_n\} = 1$, in contrast to $\rho$ which is normalized according to the 1-norm.

While the Liouvillian formalism can be applied to any classical system governed by a Hamiltonian, we will focus on molecular dynamics within the Born-Oppenheimer approximation. In this setting, the nuclei are treated as classical particles, whereas the electrons are treated quantumly. More specifically, solving the electronic Schrödinger equation as a function of the nuclear positions yields potential energy surfaces which determine the dynamics of the classical nuclei. We consider MD simulations in the microcanonical and the canonical ensemble, as discussed in the following subsections.

\subsection{Evolution in the microcanonical \texorpdfstring{($NVE$)}{NVE} ensemble}

The microcanonical ($NVE$) ensemble is a thermodynamic ensemble in which the number of nuclei $N$, the volume $V$ and the system's total energy $E$ are constants of motion. 
In order to prevent potential misunderstandings, let us remind ourselves that the phase space representation of an $NVE$ ensemble is usually considered to be a time-independent phase space density $\rho_{NVE}$ which is constant over all configurations with energy $E$ and zero otherwise. This is often written in terms of a Dirac delta distribution:
\begin{equation}
    \rho_{NVE} \lb \{x_n\}, \{p_n\} \rb \propto \delta \lb H \lb \{x_n\}, \{p_n\} \rb - E \rb.
\end{equation}
In the following, we will not assume that the phase space density is given by $\rho_{NVE} \lb \{x_n\}, \{p_n\} \rb$ when we talk about simulations in the $NVE$ ensemble. Rather, we refer to constant energy dynamics of a generic time-dependent phase space density, which evolves according to Eq.~\eqref{ev_rho}.

The classical Hamiltonian of the nuclei in the $NVE$ ensemble takes the following form:
\begin{equation}
    H_{\nuc}^{(NVE)} := H_{\tcla}^{(NVE)}\lb \{x_n\},\{p_n\}\rb + E_{\tel}\lb \{x_n\}\rb,
\label{H_NVE}
\end{equation}
where
\begin{equation}
    H_{\tcla}^{(NVE)} := \sum_{n=1}^{N} \sum_{j=1}^{3} \frac{p_{n,j}^2}{2 m_n} + \sum_{k=1}^{N-1}\sum_{n>k}^{N} \frac{Z_n Z_k}{\norm{x_n - x_k}}
\label{H_class_NVE}
\end{equation}
is the classical Hamiltonian of the nuclei without any electronic contributions. The mass and the atomic number of the $n$-th nucleus are denoted $m_n$ and $Z_n$, respectively. Unless stated otherwise, we will use $\norm{\cdot}$ to refer to the (induced) 2-norm. $H_{\nuc}^{(NVE)}$ also depends on $E_{\tel}\lb \{x_n\}\rb$, the ground state energy of the following quantum Hamiltonian governing the dynamics of the electrons for fixed nuclear positions:
\begin{equation}
    H_{\tel} := -\sum_{n=1}^{\tn} \sum_{j=1}^{3} \frac{\nabla_{n,j}^2}{2} + \sum_{n>k}^{\tn} \frac{1}{\norm{\tx_n - \tx_k}} - \sum_{k, n = 1}^{\tn, N} \frac{Z_n}{\norm{\tx_k - x_n}},
\label{H_el}
\end{equation}
where $\tn$ is the number of electrons, $\tx_n \in \mathbb{R}^3$ denotes the position of the $n$-th electron and $\nabla_{n,j} := \partial_{\tx_{n,j}}$ is the partial derivative operator with respect to the $j$-th coordinate of the $n$-th electron.
Note that the total number of particles, $N_{\text{tot}} := N + \tn$, is also conserved in the $NVE$ ensemble (the same is true for the $NVT$ ensemble discussed in the next subsection).

So far, we have worked with continuous position and momentum variables that can take on any real value.
However, to simulate the time evolution of the phase space density according to Eq.~\eqref{ev_rho} on a quantum computer, we need to consider a finite, discretized phase space. The idea is to restrict each position and momentum component of the nuclei to a finite set of
\begin{align}
    g_x &:= \frac{x_{\max}}{h_x} \in \mathbb{N} \\
    g_p &:= \frac{p_{\max}}{h_p} \in \mathbb{N}
\end{align}
values, respectively, where $x_{\max}$ ($p_{\max}$) is the maximum attainable value of any $x_{n,j}$ ($p_{n,j}$) and $h_x$ ($h_p$) is the grid spacing. The choice of grid spacing depends on the smoothness of the phase space density. A detailed error analysis regarding the grid spacing for real space simulations can be found in~\cite{Kivlichan2017realspace}. Further bounds on the grid spacing and numerical results are presented in~\cite{Obrien2022}.
Since we consider a finite simulation box, we must also specify the boundary conditions. For simplicity, we choose periodic boundary conditions.

Each grid point of the discretized phase space corresponds to a computational basis state of the form
\begin{equation}
    \ket{\{\hax_{n,j}\}, \{\hap_{n,j}\}} := \bigotimes_{n, j} \Big( \ket{\hax_{n,j}} \otimes \ket{\hap_{n,j}} \Big) \, ,
\label{NVE_basis_states}
\end{equation}
where $\hax_{n,j} \in [g_x]$ and $\hap_{n,j} \in [g_p]$ are integers such that 
$x_{n,j} = \hax_{n,j} h_x - x_{\max}/2$ and $p_{n,j} = \hap_{n,j} h_{p} - p_{\max}/2$. Thus, $\ket{\hax_{n,j}}$ is a computational basis state on $\log g_x$ qubits specifying the value of the $j$-th discrete position coordinate of the $n$-th nucleus and analogously for $\ket{\hap_{n,j}}$. The mapping to qubits to obtain the computational basis is shown in Fig.~\ref{fig:mapping}. 

\begin{figure}[tb]
    \includegraphics[width=1.\textwidth]{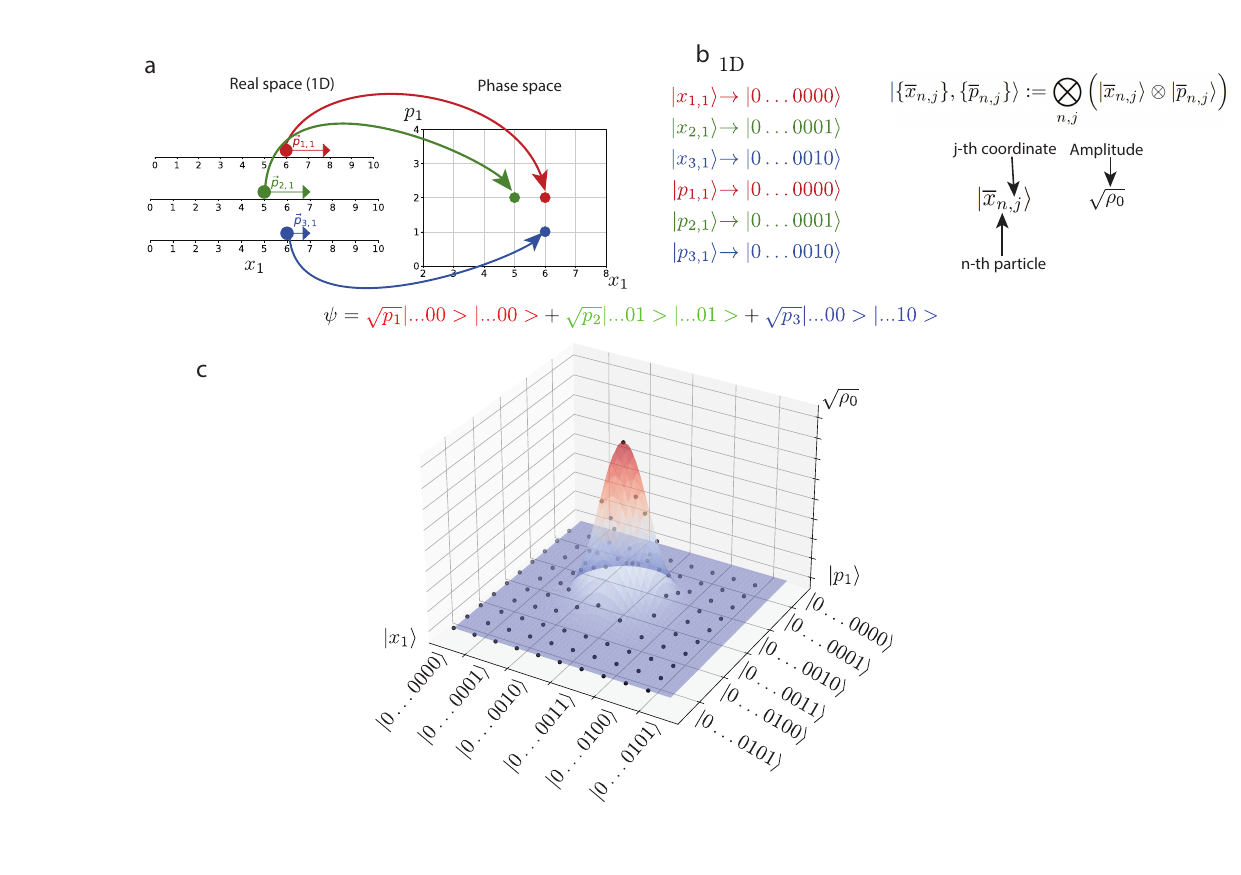}
    \caption{Mapping from real space to phase space and to qubit registers for a one-dimensional problem. \textbf{a} Each configuration is stored as a point in the phase space. \textbf{b} Each point of the discretized phase space is associated with a computational basis state on the quantum computer. The computational basis states can be written as the tensor product of the states encoding the discretized positions and momenta of the individual particles. The amplitude of a computational basis state defines the corresponding probability for that point in the probability density function.
    \textbf{c} Example encoding of the discretized phase density of one particle in one spatial dimension. }
    \label{fig:mapping}
\end{figure}

Given a classical $\lb g_x^{3N} g_p^{3N} \rb$-dimensional probability vector $\Vec{\rho}_0$ encoding the discretized initial phase space density $\rho_0$, the $k$-th amplitude of the initial quantum state representing the associated KvN wave function can simply be chosen to be $\sqrt{\lb\Vec{\rho}_{0}\rb_k}$. In other words, the quantum register is initially prepared in the state $ \ket{\rho_0} := \sum_k \sqrt{\lb\Vec{\rho}_{0}\rb_k} \ket{k}$ where $k \in [g_x^{3N} g_p^{3N}]$ enumerates the points of the discretized phase space. Note that this is just a convenient relabeling of the computational basis states $\{\ket{\{\hax_{n,j}\}, \{\hap_{n,j}\}} \}$ introduced before. 
Depending on the choice of the initial phase space density of the classical particles, one can use a number of different general purpose state preparation methods to prepare a quantum state encoding the initial phase space density. For example, if the initial phase space density is efficiently integrable, one can use the Grover-Rudolph algorithm to prepare the corresponding quantum state~\cite{Grover2002}. For sparse quantum states, this method scales quadratically with the number of qubits~\cite{Ramacciotti2023}. 
Another option is to use the QSVT-based approach developed in~\cite{Mcardle2022state_prep}. Their method provides a qubit-efficient way of encoding functions with a well behaved polynomial or Fourier expansion in the amplitudes of a quantum state.

In order to evolve the discretized quantum state that encodes our system on a quantum computer, we also need to discretize the Liouvillian operator defined in Eq.~\eqref{liouvillian}. This requires us to define discrete versions of the derivative operators appearing in the Liouvillian operator. 
Central finite difference schemes are a popular tool for discretizing derivatives~\cite{Li2005finite_difference}. In the quantum setting, the corresponding discrete operator can be defined as follows:
\begin{defn}[Discrete derivative operator]
    Let $\{ \ket{\overline{y}} \}$ denote a complete set of computational basis states, representing the variable with respect to the derivative operator is applied, e.g., $x_{n,j}$ or $p_{n,j}$. The discrete derivative operator $D_y$ of order $2d$ is defined as follows:
    \begin{equation}
        D_y := \frac{1}{h} \sum_{\overline{y}} \sum_{k = -d}^d c_{d,k} \ket{\overline{y}-k}\!\bra{\overline{y}},
    \end{equation}
    where $h$ is the user-specified grid spacing and the coefficients $\{c_{d,k}\}$ are given by~\cite{Li2005finite_difference}
    \begin{align}
        c_{d,k} := 
        \begin{cases}
            \frac{(-1)^{k+1}(d!)^2}{k(d-k)!(d+k)!} ,& \text{if } k\neq 0\\
            0 ,              & \text{else}.
        \end{cases}
    \end{align}
\label{def:discrete_derivative}
\end{defn}
In general, the higher the order $2d$ of the finite difference scheme, the better the error scaling with respect to the grid spacing, see Lemma \ref{lem:finite_diff} or Ref.~\cite{Obrien2022} for more details.

Now we are ready to define the discretized Liouvillian operator in the $NVE$ ensemble:

\begin{defn}[Discretized Liouvillian operator for BO molecular dynamics in the $NVE$ ensemble]
    Let $H_{\nuc}^{(NVE)}$ be the BO Hamiltonian from Eq.~\eqref{H_NVE}. The discretized Liouvillian for simulations in the $NVE$ ensemble is given by
    \begin{equation}
        L_{NVE} := -i\sum_{n=1}^{N} \sum_{j=1}^{3} \lb D_{x_{n,j}} \otimes \frac{\partial H_{\nuc}^{(NVE)}}{\partial {p_{n,j}}} - \frac{\partial H_{\nuc}^{(NVE)}}{\partial {x_{n,j}}} \otimes D_{p_{n,j}}\rb,
    \label{disc_liouvillian}
    \end{equation}
    where 
    \begin{align}
        \frac{\partial H_{\nuc}^{(NVE)}}{\partial {p_{n,j}}} &= \sum_{\hap_{n,j}} \frac{p_{n,j}}{m_n} \ketbra{{\hap_{n,j}}}{{\hap_{n,j}}}\\
        \frac{\partial H_{\nuc}^{(NVE)}}{\partial {x_{n,j}}} &= \sum_{n' \neq n} \sum_{\hax_{n}} \sum_{\hax_{n'}} \frac{-Z_n Z_{n'}}{\lb \norm{x_n - x_{n'}}^2 + \Delta^2 \rb^{3/2}} (x_{n,j} - x_{n',j}) \ketbra{{\hax_n}}{{\hax_n}} \otimes \ketbra{{\hax_{n'}}}{{\hax_{n'}}} + \Del 
    \end{align}
    are now diagonal matrices of dimension $g_x^{3N}g_p^{3N}$. $\Delta$ is a gap parameter introduced to regularize the Coulomb interaction and to avoid infinities in the simulation. $D_{x_{n,j}}$ is a discrete derivative operator of order $2d_x$ and $D_{p_{n,j}}$ is a discrete derivative operator of order $2d_p$.
    Furthermore,
    \begin{equation}
        \Del := \frac{1}{h_x} \sum_{k = -d_e}^{d_e} \sum_{(n',j') \neq (n,j)} \sum_{\hax_{n',j'}} \sum_{\hax_{n,j}} c_{d_e,k} E_{\tel}\lb \{ x_{n',j'}\}, x_{n,j} + k h_x \rb \ketbra{\hax_{n',j'}}{\hax_{n',j'}} \otimes \ketbra{\hax_{n,j}}{\hax_{n,j}}
    \label{el_derivative}
    \end{equation}
    is a central finite difference approximation of order $2 d_e$ to $\frac{\partial E_{\tel}}{\partial {x_{n,j}}}$.
    Note that we only show the quantum registers that are acted on in a nontrivial manner. For example, 
    \begin{equation}
        \sum_{\hap_{1,1}} \ketbra{{\hap_{1,1}}}{{\hap_{1,1}}} \equiv \lb \mathbb{1}_{x_{1,1}} \otimes \sum_{\hap_{1,1}} \ketbra{{\hap_{1,1}}}{{\hap_{1,1}}} \rb \otimes \lb \mathbb{1}_{x_{1,2}} \otimes \mathbb{1}_{p_{1,2}} \rb \otimes \lb \mathbb{1}_{x_{1,3}} \otimes \mathbb{1}_{p_{1,3}} \rb \bigotimes_{n=2, j=1}^{N,3} \lb \mathbb{1}_{x_{n,j}} \otimes \mathbb{1}_{p_{n,j}} \rb.
    \end{equation}
\label{def:NVE_liouvillian}
\end{defn}

The reason for introducing $\Del$ in the above definition is that we generally do not have an analytic expression for $\frac{\partial E_{\tel}}{\partial {x_{n,j}}}$.

In Appendix \ref{app:overall_Liouvillian} we show that the spectral norm of of $L_{NVE}$ is upper bounded by $\mu_{NVE}$ which is defined as follows:
\begin{equation}
    \mu_{NVE} := 3N \frac{p_{\text{max}}}{m_{\text{min}}} \frac{2 \lb \ln{d_x} + 1 \rb}{h_x} + 6N^2 \frac{2 Z_{\text{max}}^2 x_{\text{max}}}{\Delta^3} \frac{2 \lb \ln{d_p} + 1 \rb}{h_p} + 3N \lambda \frac{2 \ln(d_e+1)}{h_x h_p},
\label{mu_NVE}
\end{equation}
where $\lambda$ is an upper bound on the spectral norm of the discretized electronic Hamiltonian as discussed in Lemma \ref{lem:block-encode_Hel}.

Since the discrete Liouvillian $L_{NVE}$ is still Hermitian, we can use tools from Hamiltonian simulation~\cite{Childs2021trotter_error, Gilyen2019qsvt} to efficiently implement the unitary $U_{NVE} := e^{-iL_{NVE} t}$ on a quantum computer to simulate the time evolution of the discretized phase space density. Our quantum simulation algorithm is discussed in more detail in Section \ref{sec:algorithm}.

\subsection{Evolution in the canonical \texorpdfstring{($NVT$)}{NVT} ensemble}

By default, MD simulations are performed in the $NVE$ ensemble. However, it is often desirable to perform simulations in the canonical ($NVT$) ensemble where the temperature $T$ rather than the energy is held constant. This is especially true when performing conformational searches of molecules, such as those required in drug design~\cite{Yu2017}. 

The Nosé-Hoover thermostat is a common choice in classical MD calculations to simulate dynamics in the $NVT$ ensemble~\cite{Huenenberger2005thermostat, Frenkel20021}. This thermostat is based on non-Hamiltonian equations of motion, meaning there does not exist an underlying Hamiltonian governing the dynamics of the system. Therefore, it cannot be straightforwardly incorporated into the Liouvillian framework.

In contrast, the original Nosé thermostat is compatible with the Liouvillian framework since the equations of motion can be derived from an extended-system Hamiltonian~\cite{Nose1984}. The idea is to introduce additional terms to the classical Hamiltonian that involve an extra degree of freedom, $s$, representing a heat bath. This effectively allows the kinetic energy of the nuclei to be exchanged with the energy of the bath so that the system can be equilibrated to a user-specified temperature $T$. The extended-system Hamiltonian is defined as follows:
\begin{equation}
    H_{\nuc}^{(NVT)} := H_{\tcla}^{(NVT)}\lb \{x_n\},\{p'_n\}, s, p_s \rb + E_{\tel}\lb \{x_n\}\rb,
\label{H_NVT}
\end{equation}
where
\begin{equation}
    H_{\tcla}^{(NVT)} := \sum_{n=1}^{N} \sum_{j=1}^{3} \frac{{p'}_{n,j}^2}{2 m_n s^2} + \sum_{k=1}^{N-1}\sum_{n>k}^{N} \frac{Z_n Z_k}{\norm{x_n - x_k}^2} + {\frac{p_s^2}{2 Q}} + N_f k_B T \ln \lb s \rb \,.
\label{H_class_NVT}
\end{equation}
Here, $p_s$ is the momentum variable conjugate to $s$, and $Q$ is an effective mass of $s$, which controls the coupling of the system to the heat bath. $k_B$ is the Boltzmann constant, and $N_f = 3 N - K$ is equal to the number of degrees of freedom of the system, with $K$ being the number of constraints. 
The heat bath modifies the kinetic energy term of the nuclei while the Coulomb potential term remains unaffected. In particular, $p'_{n,j}$ is the conjugate momentum variable to $x_{n,j}$ in the extended system. It is often called a ``virtual'' momentum variable and is related to the real momentum variable $p_{n,j}$ of the physical system via the following equation:
\begin{equation}
    p_{n,j} = \frac{p'_{n,j}}{s}.
\end{equation}
The third term of $H_{\tcla}^{(NVT)}$ represents the kinetic energy of the heat bath, while the last term represents the potential energy of the heat bath. 
This potential energy term ensures that the partition function $\mathcal{Z}$ associated with a microcanonical ensemble in the extended system gives rise to a canonical partition function when restricted to the real system~\cite{Nose1984partition_function, Huenenberger2005thermostat}:
\begin{equation}
\begin{split}
    \mathcal{Z} &\propto \int d\{x_n\} \int d\{p'_n\} \int ds \int dp_s \, \delta \lb H_{\nuc}^{(NVT)} \lb \{x_n\}, \{p'_n\}, s, p_s \rb - E_{\text{ext}}\rb \\
    &\propto  \int d\{x_n\} \int d\{p_n\} e^{- H_{\nuc}^{(NVE)} \lb \{x_n\}, \{ p_n \} \rb/(k_BT)},
\end{split}
\label{partition_func}
\end{equation}
where $E_{\text{ext}}$ is the conserved energy of the extended system. 

To avoid confusion later on let us briefly mention here that the phase space version of an $NVT$ ensemble is usually considered to be a time-independent phase space density $\rho_{NVT}$ that has the form of a Boltzmann distribution:
\begin{equation}
    \rho_{NVT} \lb \{x_n\}, \{p_n\} \rb \propto e^{-E \lb \{x_n\}, \{p_n\} \rb/(k_b T)},
\end{equation}
where $E$ is the energy of the system (nuclei) for a given configuration. In the following, we will not assume that the phase space density is given by $\rho_{NVT} \lb \{x_n\}, \{p_n\} \rb$ when we talk about simulations in the $NVT$ ensemble. Rather, we refer to constant temperature dynamics of a generic time-dependent phase space density obtained by evolving the joined probability density of the system and heat bath according to Eq.~\eqref{ev_rho} and then integrating out the heat bath. However, if the dynamics of the extended system are ergodic, then we can mimic the behavior of $\rho_{NVT} \lb \{x_n\}, \{p_n\} \rb$ in the sense that we can estimate thermodynamic properties such as the free energy of the system via coherent time averaging as explained in more detail in Section~\ref{sec:free_energy}.

As with the $NVE$ ensemble, we need to discretize the (now $(6N+2)$-dimensional) phase space to simulate the time evolution of the phase space density $\rho \lb  \{x_n\}, \{p'_n\}, s, p_s\rb$ according to Eq.~\eqref{ev_rho} on a quantum computer. We do so by restricting each position and virtual momentum component of the nuclei as well as the bath variables $s$ and $p_s$ to a finite set of
\begin{align}
    g_x &= \frac{x_{\max}}{h_x} \in \mathbb{N} \\
    g_{p'} &:= \frac{p'_{\max}}{h_{p'}} \in \mathbb{N} \\
    g_s &:= \frac{s_{\max}}{h_s} \in \mathbb{N} \\
    g_{p_s} &:= \frac{p_{s, \max}}{h_{p_s}} \in \mathbb{N}
\end{align}
values, respectively, where $x_{\max}$ is the maximum attainable value of $x_{n,j}$ and $h_x$ is the position grid spacing as before. Similarly, $p'_{\max}$, $s_{\max}$ and $p_{s, \max}$ are the maximum attainable values of $p'_{n,j}$, $s$ and $p_s$ and $h_{p'}$, $h_{s}$ and $h_{p_s}$ are the respective grid spacings.
Since we consider a finite simulation box, we also need to specify the boundary conditions. For simplicity, we again choose periodic boundary conditions.

The computational basis states in the $NVT$ ensemble are of the form
\begin{equation}
    \ket{\{\hax_{n,j}\}, \{\hap'_{n,j}\}, s, p_s} := \bigotimes_{n, j} \Big( \ket{\hax_{n,j}} \otimes \ket{\hap'_{n,j}} \Big) \otimes \ket{\has} \otimes \ket{\hap_{s}} \, ,
\label{NVT_basis_states}
\end{equation}
where $\hax_{n,j} \in [g_x]$, $\hap'_{n,j} \in [g_{p'}]$, $\has \in [g_s]$ and $\hap_{s} \in [g_{p_s}]$ are integers such that $x_{n,j} = \hax_{n,j} h_x - x_{\max}/2$, $p'_{n,j} = \hap'_{n,j} h_{p'} - p'_{\max}/2$, $s = \has h_s$ and $p_{s} = \hap_{s} h_{p_s} - p_{s, \max}/2$.

We again employ the Koopman-von Neumann formalism to encode the (discretized) phase space density in a quantum state on a quantum computer.

The discretized Liouvillian in the $NVT$ ensemble is then defined as follows:
\begin{defn}[Discretized Liouvillian operator for BO molecular dynamics in the $NVT$ ensemble]
    Let $H_{\nuc}^{(NVT)}$ be the $NVT$ Hamiltonian as defined in Eq.~\eqref{H_NVT}. The discretized Liouvillian for simulations in the $NVT$ ensemble is given by
    \begin{equation}
        L_{NVT} := -i\sum_{n=1}^{N} \sum_{j=1}^{3} \lb  D_{x_{n,j}} \otimes \frac{\partial H_{\nuc}^{(NVT)}}{\partial {p'_{n,j}}} - \frac{\partial H_{\nuc}^{(NVT)}}{\partial {x_{n,j}}} \otimes D_{p'_{n,j}} \rb -i \lb  D_{s} \otimes \frac{\partial H_{\nuc}^{(NVT)}}{\partial_{p_{s}}} - \frac{\partial H_{\nuc}^{(NVT)}}{\partial_{s}} \otimes D_{p_s} \rb,
    \end{equation}
     where 
    \begin{align}
        \frac{\partial H_{\nuc}^{(NVT)}}{\partial {p'_{n,j}}} &= \sum_{\hap'_{n,j}} \sum_{\has} \frac{p'_{n,j}}{m_n \lb s + s_{\min} \rb^2} \ketbra{{\hap'_{n,j}}}{{\hap'_{n,j}}} \otimes \ketbra{\has}{\has} \\
        \frac{\partial H_{\nuc}^{(NVT)}}{\partial {x_{n,j}}} &= \sum_{n' \neq n} \sum_{\hax_{n}} \sum_{\hax_{n'}} \frac{-Z_n Z_{n'}}{\lb \norm{x_n - x_{n'}}^2 + \Delta^2 \rb^{3/2}} (x_{n,j} - x_{n',j}) \ketbra{{\hax_n}}{{\hax_n}} \otimes \ketbra{{\hax_{n'}}}{{\hax_{n'}}} + \Del \\
        \frac{\partial H_{\nuc}^{(NVT)}}{\partial p_{s}} &= \sum_{\hap_{s}} \frac{p_{s}}{Q} \ketbra{{\hap_{s}}}{\hap_{s}} \\
        \frac{\partial H_{\nuc}^{(NVT)}}{\partial {s}} &= -\sum_{\hap_{n,j}} \sum_{\has} \frac{2 {p'}_{n,j}^2}{m_{n} \lb s + s_{\min} \rb^3} \ketbra{{\hap_{n,j}}}{{\hap_{n,j}}} \otimes \ketbra{{\has}}{\has} + \sum_{\has} \frac{N_f k_B T}{s + s_{\min}} \ketbra{{\has}}{\has}
    \end{align}
    are now diagonal matrices of dimension $\eta := g_x^{3N}g_{p'}^{3N}g_s g_{p_s}$. As with the $NVE$ Liouvillian, $\Delta$ is a gap parameter introduced to regularize the Coulomb interaction and to avoid infinities in the simulation. The bath variable cutoff $s_{\min}$ is introduced for the same reason.
    $D_{x_{n,j}}$ is a discrete derivative operator of order $2d_x$, $D_{p'_{n,j}}$ is a discrete derivative operator of order $2d_{p'}$, $D_{s}$ is a discrete derivative operator of order $2d_s$ and $D_{p_s}$ is a discrete derivative operator of order $2d_{p_s}$. They are constructed according to Definition \ref{def:discrete_derivative}. $\Del$ is again the finite difference approximation to $\frac{\partial E_{\tel}}{\partial {x_{n,j}}}$.
    Note that we only show the quantum registers that are acted on in a nontrivial manner.
\label{def:NVT_liouvillian}
\end{defn}

In Appendix \ref{app:overall_Liouvillian} we show that the spectral norm of of $L_{NVT}$ is upper bounded by $\mu_{NVT}$ which is defined as follows:
\begin{equation}
\begin{split}
    \mu_{NVT} &:= 3N \frac{p_{\text{max}}}{m_{\text{min}} s_{\text{min}}^2} \frac{2 \lb \ln{d_x} + 1 \rb}{h_x} + 6N^2 \frac{2 Z_{\text{max}}^2 x_{\text{max}}}{\Delta^3} \frac{2 \lb \ln{d_p} + 1 \rb}{h_p} + 3N \lambda \frac{2 \ln{(d_e+1)}}{h_x h_p} \\
    & \quad + \frac{p_{s, \text{max}}}{Q} \frac{2 \lb \ln{d_s} + 1 \rb}{h_s} + 3N \frac{2 p_{\text{max}}^2}{m_{\text{min}} s_{\text{min}}^3} \frac{2 \lb \ln{d_{p_s}} + 1 \rb}{h_{p_s}} + \frac{N_f k_B T}{s_{\text{min}}} \frac{2 \lb \ln{d_{p_s}} + 1 \rb}{h_{p_s}}.
\end{split}
\label{mu_NVT}
\end{equation}

\section{Main results}\label{sec:main}

Our main result is an efficient quantum algorithm for simulating the time evolution of the discretized phase space density of $N$ nuclei within the Born-Oppenheimer approximation. 
The formal problem can be stated as follows:
\begin{prob}[Simulating Liouvillian dynamics within the Born-Oppenheimer approximation]
    Let $L \in \{L_{NVE}, L_{NVT} \}$ be the discretized Liouvillian governing the dynamics of the discretized phase space density associated with $N$ classical nuclei in the $NVE$ ensemble (Definition \ref{def:NVE_liouvillian}) or the $NVT$ ensemble (Definition \ref{def:NVT_liouvillian}).
    Given a quantum state $\ket{\rho_0}$ encoding the initial discretized phase space density, output a quantum state that is $\epsilon$-close in $\ell^2$ distance to $\ket{\rho_t} := e^{-iL t}\ket{\rho_0}$.
\label{prob:sim}
\end{prob}

Our algorithm requires access to an initial electronic state preparation oracle $\widetilde{U}_I$ whose precise definition is given later in Definition \ref{def:state_prep}. The main feature of $\widetilde{U}_I$ is that it prepares an initial electronic state $\ket{\phi_0 \{ x_n \}}$ that has nontrivial overlap with the ground state $\ket{\widetilde{\psi}_0 \{ x_n \}}$ of a discretized version of the electronic Hamiltonian from Eq.~\eqref{H_el}. To be more specific, let $0 < \widetilde{\delta} \leq 1$. Then
\begin{equation}
    \widetilde{U}_I \ket{\{ \hax_n \}} \ket{0} = \ket{\{ \hax_n \}} \ket{\phi_0 \{ x_n \}},
\end{equation}
where $\left| \braket{\widetilde{\psi}_0 \{ x_n \}}{\phi_0 \{ x_n \}} \right| \geq \widetilde{\delta}$ for all nuclear configurations visited during the simulation. In other words, $\widetilde{\delta}$ is a lower bound on the overlap of the initial electronic state with the true electronic ground state over all nuclear grid points associated with a nonzero amplitude at some point during the simulation.

Unless stated otherwise, we use ``$\log$'' to refer to the binary logarithm. Furthermore, we write $O \lb z^{o\lb 1 \rb} \rb$ to indicate sub-polynomial scaling with respect to the parameter $z$ and we use the $\widetilde{O}$ notation to hide subdominant logarithmic factors.
With this in mind, we present our first result below.
\begin{restatable}[Complexity of Born-Oppenheimer Liouvillian simulation]{thm}{liouvillian}
    There exists a quantum algorithm that solves Problem \ref{prob:sim} with success probability $\geq 1 -\xi$ using
    \begin{equation*}
        \widetilde{O} \lb \frac{N_{\text{tot}} \, d \, \mu^{2+o(1)} t^{1+o(1)}}{\widetilde{\gamma} \, \widetilde{\delta} \, \epsilon^{o(1)}} \log \lb \frac{1}{\xi} \rb \rb
    \end{equation*}
    Toffoli gates, where $d$ is the maximum order of the finite difference schemes used, $\mu \in \{\mu_{NVE}, \mu_{NVT} \}$ is an upper bound on the spectral norm of the discretized Liouvillian $L \in \{L_{NVE}, L_{NVT} \}$ and $\widetilde{\gamma}$ is a lower bound on the spectral gap of the discretized electronic Hamiltonian over all phase space grid points that are associated with a nonzero amplitude at some point during the simulation.
    Additionally,
    \begin{equation*}
        \widetilde{O} \lb \frac{N d \mu^{1+o(1)} t^{1+o(1)}}{\widetilde{\delta} \, \epsilon^{o(1)}} \log \lb \frac{1}{\xi} \rb \rb
    \end{equation*}
    queries to the initial electronic state preparation oracle $\widetilde{U}_I$ are needed.
\label{thm:complexity_liouvillian}
\end{restatable}

Theorem \ref{thm:complexity_liouvillian} is proved in Appendix \ref{app:overall_Liouvillian}.
In comparison to gradient-based approaches~\cite{Obrien2022}, which, in the worst case scale exponentially with the evolution time, see Appendix~\ref{app:EulerCost}, our approach scale polynomially in time $t$ and sub-polynomially with error $\epsilon$.

We also show how to use our Liouvillian simulation algorithm to estimate the free energy of the nuclei in the $NVT$ ensemble assuming the dynamics of the extended system are ergodic.
Usually, we are interested in the free energy when the system reaches equilibrium. The thermostat allows us to reach thermal equilibrium and then estimate thermodynamic properties such as the free energy of the classical system.
A conceptual challenge that arises in the computation of the free energy is the definition of macrostates in the probability distribution. Specifically, we envision these microstates to be hypercubes in phase space and define the probability of finding the entire system within this hypercube to be $p_i$.  With this in mind, the definition of the discrete free energy is given below.
\begin{defn}[Free energy]
    Let $p_i$ denote the probability of a classical system being in the $i$-th (discrete) microstate and let $E_i$ be the energy associated with the $i$-th microstate. Let
    \begin{equation}
        \mathcal{S}_G := -k_B \sum_j p_j \ln{p_j},
    \end{equation}
    be the Gibbs entropy of the system, where $k_B$ is the Boltzmann constant.
    Furthermore, let
    \begin{equation}
        \mathcal{U} := \sum_j p_j E_j
    \end{equation}
   be the internal energy of the system. The free energy $F$ of the system is then given by
   \begin{equation}
       \mathcal{F} := \mathcal{U} - T \mathcal{S}_G,
   \end{equation}
   where $T$ is the temperature of the system.
\label{def:free_energy}
\end{defn}

In our case, the energies $\{E_i\}$ are the eigenvalues of $H_{\text{nuc}} := H_{\text{kin}} + H_{\text{pot}} + H_{E_{\tel}}$ where
\begin{align}
    H_{\text{kin}} &:= \sum_{n,j} \sum_{\hap'_{n,j}} \sum_{\has} \frac{{p'}_{n,j}^2}{m_n (s + s_{\min})^2} \ketbra{{\hap'_{n,j}}}{{\hap'_{n,j}}} \otimes \ketbra{\has}{\has} \label{H_kin} \\
    H_{\text{pot}} &:= \sum_{n' \neq n} \sum_{\hax_{n}} \sum_{\hax_{n'}} \frac{Z_n Z_{n'}}{\lb \norm{x_n - x_{n'}}^2 + \Delta^2 \rb^{1/2}}\ketbra{{\hax_n}}{{\hax_n}} \otimes \ketbra{{\hax_{n'}}}{{\hax_{n'}}} \label{H_pot} \\
    H_{E_{\tel}} &:= \sum_{\{ \hax_n \}} E_{\tel} \lb \{ x_n \} \rb \ketbra{\{ \hax_n \}}{\{ \hax_n \}}. \label{H_E_el}
\end{align}

In Appendix \ref{app:free} we show that the spectral norm of of $H_{\text{nuc}}$ is upper bounded by $\alpha_{\nuc}$ which is defined as follows:
\begin{equation}
    \alpha_{\nuc} := 3 N \frac{{p'}_{\max}^2}{m_{\min} s_{\min}^2} + N^2 \frac{Z_{\max}^2}{\Delta} + \lambda.
\end{equation}

We now give a formal definition of the free energy estimation problem.

\begin{prob}[Estimating the free energy of a phase space distribution]
    Let $L_{NVT}$ be the discretized Liouvillian operator in the $NVT$ ensemble as in Definition \ref{def:NVT_liouvillian}. Given a quantum state $\ket{\rho_0}$ encoding the initial discretized phase space density of the system and the heat bath, output an $\epsilon$-precise estimate of the free energy $\mathcal{F}$ of the system after time $t$, i.e.~estimate $\mathcal{F}$ associated with 
    \begin{equation}
        \rho_{\text{sys}}(t) := \text{Tr}_{\text{bath}} \lb  e^{-iL_{NVT} t}\ketbra{\rho_0}{\rho_0} e^{iL_{NVT}} \rb.
    \end{equation}
\label{prob:free_energy}
\end{prob}

Note that the above problem description allows the free energy to be a time-dependent quantity. If the dynamics of the extended system (nuclei + heat bath) are ergodic, then we can obtain an estimate of the equilibrium free energy via coherent time averaging; see Section~\ref{sec:free_energy} for more details.

The next theorem provides upper bounds on the complexity of estimating the free energy associated with the nuclear phase space density.

\begin{restatable}[Estimation of the free energy]{thm}{free}
    Let $\eta := g_x^{3N}g_{p'}^{3N}g_s g_{p_s}$ be the number of grid points of the discretized phase space and assume that $\log \lb \eta^2/\epsilon \rb \leq \eta$. Then there exists a quantum algorithm that solves Problem \ref{prob:free_energy} with success probability at least $1 - \xi$ using
    \begin{equation*}
        \widetilde{O} \lb \lb \frac{\eta^{o(1)} N_{\text{tot}} d \mu_{NVT}^{2+o(1)} t^{1+o(1)}}{\widetilde{\gamma} \, \widetilde{\delta} \, \epsilon^{1+o(1)}} \lb \alpha_{\nuc} + \frac{\eta (k_b T)^{1.5 + o(1)}}{\sqrt{\epsilon}} \rb + \frac{N_{\text{tot}} \alpha_{\text{nuc}} \lambda}{\epsilon^2} \rb \log \lb \frac{1}{\xi} \rb  \rb.
    \end{equation*}
    Toffoli gates.
    Additionally,
    \begin{equation*}
        \widetilde{O} \lb \frac{ \eta^{o(1)} N d \mu_{NVT}^{1+o(1)} t^{1+o(1)}}{\widetilde{\delta} \, \epsilon^{1+o(1)}} \log \lb \frac{1}{\xi} \rb \lb \alpha_{\nuc} + \frac{\eta (k_b T)^{1.5 + o(1)}}{\sqrt{\epsilon}} \rb \rb.
    \end{equation*}
    queries to the initial electronic state preparation oracle $\widetilde{U}_I$ are needed.
\label{thm:free_energy}
\end{restatable}

Theorem \ref{thm:free_energy} is proved in Appendix \ref{app:free}.
Although the scaling with $\eta$ may seem challenging at first as it implies, in the worst case, exponential scaling with the number of nuclei, this is actually a reasonable expectation because estimating the free energy is an NP-hard problem~\cite{Sorin2000}. However, for many practical problems, the phase space can be coarse-grained, which reduces the complexity considerably and makes the problem more manageable.

\section{Overview of the algorithm}
\label{sec:algorithm}

As mentioned before, the (discretized) Liouvillian $L$ is Hermitian, meaning that the time evolution operator $e^{-iLt}$ of the (discretized) phase space density is unitary. Hence, we can use Hamiltonian simulation algorithms to implement $e^{-iLt}$ on a quantum computer~\cite{Suzuki1991product_formula, Berry2006trotter, Gilyen2019qsvt, Childs2021trotter_error}. The main idea is to split the overall Liouvillian $L = L_{\tcla} + L_{\tel}$ into a classical, electron-independent part and an electronic part as defined below.

\begin{defn}[Electronic Liouvillian]
     Let $\Del$ be the finite difference approximation to $\frac{\partial E_{\tel}}{\partial {x_{n,j}}}$ as in Definition \ref{def:NVE_liouvillian}.
     In the $NVE$ ensemble, the electronic Liouvillian acting on the nuclei is given by
     \begin{equation}
        L_{\tel}^{(NVE)} := i\sum_{n=1}^{N} \sum_{j=1}^{3} \Del \otimes D_{p_{n,j}}^{1},
    \end{equation}
    where $D_{p_{n,j}}^{1}$ is a 2nd-order discrete derivative approximation to $\partial_{p_{n,j}}$.
    In the $NVT$ ensemble, the electronic Liouvillian acting on the nuclei is given by
    \begin{equation}
        L_{\tel}^{(NVT)} := i\sum_{n=1}^{N} \sum_{j=1}^{3} \Del \otimes D_{p'_{n,j}}^{1},
    \end{equation}
    where $D_{p'_{n,j}}^{1}$ is a 2nd-order discrete derivative approximation to $\partial_{p'_{n,j}}$.
\label{def:el_liouvillian}
\end{defn}
The reason for restricting $D_{p_{n,j}}$ and $D_{p'_{n,j}}$ to be 2nd-order discrete derivatives in the above definition is related to their implementation as explained in more detail in Appendix \ref{app:electronic}. 

\begin{defn}[Classical Liouvillian]
    Let $L \in \{L_{NVE}, L_{NVT} \}$ be the discretized Liouvillian in the $NVE$ ensemble (Definition \ref{def:NVE_liouvillian}) or the $NVT$ ensemble (Definition \ref{def:NVT_liouvillian}). Let $L_{\tel} \in \{L_{\tel}^{(NVE)}, L_{\tel}^{(NVT)} \}$ be the electronic Liouvillian from Definition \ref{def:el_liouvillian}.
    The classical Liouvillian is then given by
    \begin{equation}
        L_{\tcla} := L - L_{\tel}.
    \end{equation}
\label{def:class_liouvillian}
\end{defn}

We simulate $U_{L_{\tcla}} := e^{-iL_{\tcla}t}$ and $U_{L_{\tel}} := e^{-iL_{\tel}t}$ separately and then recombine them using a $2k$-th order Trotter-Suzuki product formula~\cite{Suzuki1991product_formula, Berry2006trotter, Childs2021trotter_error}.

\begin{defn}[$2k$-th order Trotter-Suzuki product formula]
    Let $L = \sum_{\gamma=1}^{\Gamma} L_\gamma$ be an operator consisting of $\Gamma$ Hermitian summands and $t \ge 0$. Then the following recursion defines
    $\mathcal{S}_{2k}(t)$, the Trotter-Suzuki product formula of order $2k$:
     \begin{align}
        \mathcal{S}_{2}(t) &:= e^{L_1 \frac{t}{2}} \cdots e^{L_\Gamma \frac{t}{2}} e^{L_\Gamma \frac{t}{2}} \cdots e^{L_1 \frac{t}{2}}\\
        \mathcal{S}_{2k}(t) &:= \mathcal{S}^2_{2k-2}(u_k t) \mathcal{S}_{2k-2}((1-4u_k)t)\mathcal{S}^2_{2k-2}(u_k t),
    \end{align}
    where
    \begin{equation}
        \frac{1}{3} \le u_k := \frac{1}{\lb 4 - 4^{\frac{1}{2k-1}} \rb} \le \frac{1}{2} \hspace{0.4cm} \forall k \in \mathbb{N}, k \ge 2.
    \end{equation}
\label{def:2ktrotter}
\end{defn}

Using results from~\cite{Childs2021trotter_error}, we show in Appendix \ref{app:overall_Liouvillian} that the cost of our algorithm depends on the spectral norm of the nested commutator of $L_{\tcla}$ and $L_{\tel}$. We derive upper bounds $\mu'_{NVE}$ and $\mu'_{NVT}$ on the nested commutator in the $NVE$ and the $NVT$ ensemble which provide better scaling with respect to the number of nuclei $N$ than $\mu_{NVE}$ (Eq.~\eqref{mu_NVE}) and $\mu_{NVT}$ (Eq.~\eqref{mu_NVT}). 
However, we use the looser bounds $\mu_{NVE}$ and $\mu_{NVT}$ in the main theorems to keep the statements as simple as possible.

Our algorithm requires several different quantum registers. In the $NVE$ ensemble, we use two types of registers. First, we have a nuclear register for encoding the nuclear phase space density. The basis states of this register are given in Eq.~\eqref{NVE_basis_states}. Simulating $U_{L_{\tel}}$ requires a second register which we call the electronic register as it is used to encode the electronic wave function. This register can be treated like an ancilla register in the sense that it is only used to compute $\Del$ during the simulation of $U_{L_{\tel}}$. At the end of the algorithm, the electronic register is uncomputed.
In the $NVT$ ensemble, we have a third register for the bath variables $s$ and $p_s$. The computational basis states of the nuclear register together with the bath register are given in Eq.~\eqref{NVT_basis_states}.

Fig.~\ref{fig:algorithm_scheme} summarizes our quantum algorithm for $NVE$ and $NVT$ Liouvillian simulations. The corresponding pseudocode is presented in Algorithm \ref{alg:liouvillian}. The subroutines for evolving the phase space density under the classical Liouvillian $L_{\tcla}$ and the electronic Liouvillian $L_{\tel}$ are summarized in Algorithms \ref{alg:class_ev} and \ref{alg:el_ev}, respectively.

\begin{figure}[tb]
    \centering
    \resizebox{0.8\textwidth}{!}{
    \includegraphics{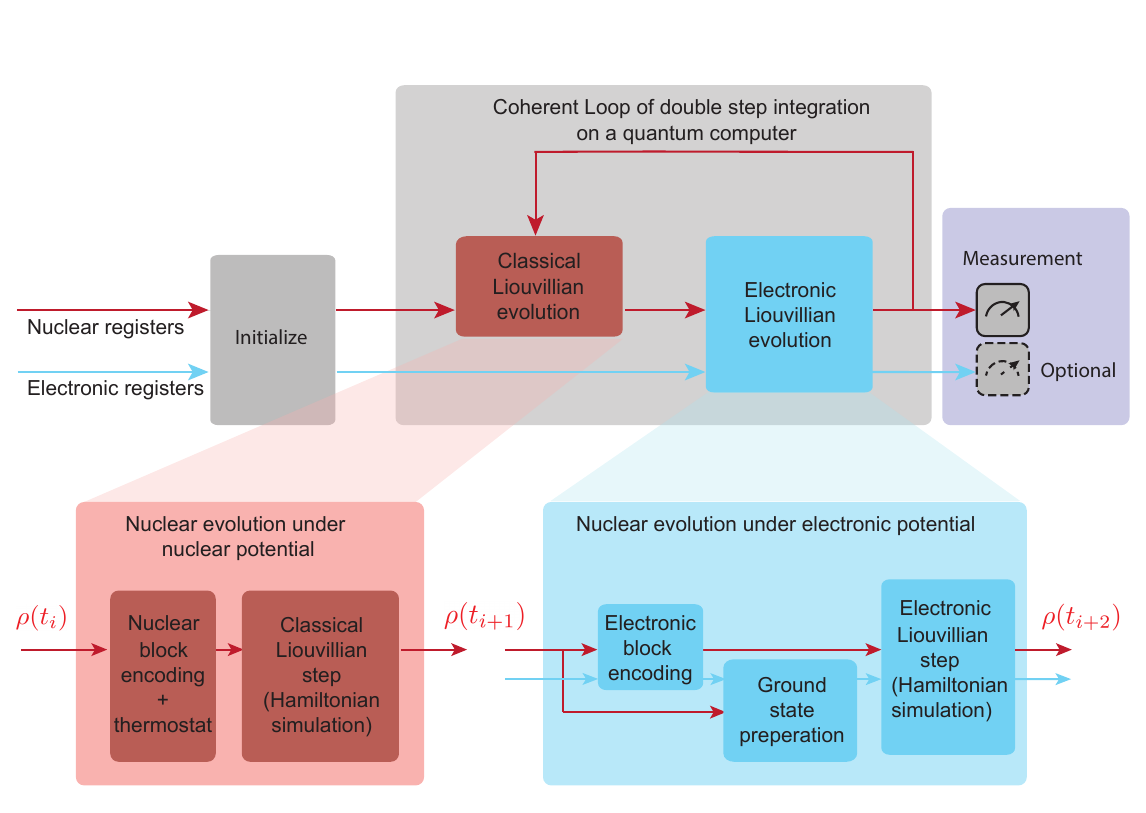} 
    }
    \caption{
    Liouvillian algorithm scheme. Red arrows represent the nuclear registers, while electronic registers are in blue. After initializing the registers into the chosen initial conditions, the time evolution is carried out by iterating a two-step integration. The double-step integration originates from alternating the classical nuclear Liouvillian evolution ($t_i +1$) with the electronic one ($t_i +2$). In the classical nuclear Liouvillian evolution, the classical nuclear block-encoding is used to implement the evolution under the classical NVE or NVT Liouvillians. The electronic Liouvillian evolution step takes care of the implementation of the electronic block-encodings, preparing the electronic ground state given the updated nuclear coordinates, and applying the nuclear evolution due to the electronic wavefunction contribution. After the required number of integration steps is achieved, one can decide whether to measure the output states, use them for other computations, or, for the nuclear register, use it for the estimation of the free energy.}
    \label{fig:algorithm_scheme}
\end{figure}

Let us now explain these subroutines in more detail.
To implement $U_{L_{\tcla}}$ we first construct a block-encoding of $L_{\tcla}$ according to the following definition:

\begin{defn}[Block-encoding (\cite{Gilyen2019qsvt}, Definition 24)]
    Let $A$ be an $s$-qubit operator, $\alpha \in \mathbb{R}$ a normalization constant, $\epsilon \in \mathbb{R}$ the allowable error and $a \in \mathbb{N}$ the number of ancilla qubits. Then we define that the $(s+a)$-qubit unitary $U$ is an $(\alpha, a, \epsilon)$-block-encoding of $A$, if
    \begin{equation}
        \norm{A - \alpha \lb \bra{0}^{\otimes a} \otimes \mathbb{1} \rb U \lb \ket{0}^{\otimes a} \otimes \mathbb{1} \rb} \le \epsilon.
    \end{equation}
\label{def:block-encoding}
\end{defn}
In other words, we embed $L_{\tcla}$ inside a larger unitary matrix. If $\lb \bra{0}^{\otimes a} \otimes \mathbb{1} \rb U \lb \ket{0}^{\otimes a} \otimes \mathbb{1} \rb$ is Hermitian then we call $U$ a Hermitian block-encoding.  

The following two lemmas show that we can block-encode $L_{\tcla}$ efficiently. Both lemmas are proved in Appendix \ref{app:L_class}.
\begin{restatable}[Block-encoding of the discretized classical $NVE$ Liouvillian]{lem}{LclassNVE}
     There exists an $(\alpha_{NVE}, a_{NVE}, \epsilon)$-block-encoding of the discretized classical Liouvillian $L_{\tcla}^{(NVE)}$ with normalization constant
    \begin{equation*}
        \alpha_{NVE} \in O \lb N\frac{p_{\text{max}}}{m_{\text{min}}} \frac{\ln d_x}{h_x} + N^2\frac{ Z_{\text{max}}^2 x_{\text{max}}}{\Delta^3} \frac{\ln d_p}{h_p} \rb
    \end{equation*}
    and a number of ancilla qubits
    \begin{equation*}
        a_{NVE} \in O \lb \log \lb \frac{\alpha_{NVE}}{\epsilon} \rb + \log d \rb
    \end{equation*}
    where $d := \max \{ d_x, d_p \}$.
    This block-encoding can be implemented using
    \begin{equation*}
        \widetilde{O} \lb N \log \lb \frac{g \alpha_{NVE}}{\epsilon} \rb + \log^{\log 3}{\lb \frac{\alpha_{NVE}}{\epsilon} \rb}  + d \log g \rb
    \end{equation*}
    Toffoli gates, where $g := \max \{ g_x, g_p \}$.
\label{lem:bounds_L_class_NVE}
\end{restatable}

\begin{restatable}[Block-encoding of the discretized classical $NVT$ Liouvillian]{lem}{LclassNVT}
     There exists an $(\alpha_{NVT}, a_{NVT}, \epsilon)$-block-encoding of the discretized classical Liouvillian $L_{\tcla}^{(NVT)}$ with normalization constant
    \begin{equation*}
         \alpha_{NVT} \in O \lb N \frac{p'_{\text{max}}}{m_{\text{min}} s_{\text{min}}^2} \frac{\ln d_x}{h_x} + N^2\frac{ Z_{\text{max}}^2 x_{\text{max}}}{\Delta^3} \frac{\ln d_{p'}}{h_{p'}} + \frac{p_{s, \text{max}}}{Q} \frac{\ln d_s}{h_s} + \lb N \frac{{p'}_{\text{max}}^2}{m_{\text{min}} s_{\text{min}}^3} + \frac{N_f k_B T}{s_{\text{min}}} \rb \frac{\ln d_{p_s}}{h_{p_s}} \rb
    \end{equation*}
    and a number of ancilla qubits
    \begin{equation*}
        a_{NVT} \in O \lb \log \lb \frac{\alpha_{NVT}}{\epsilon} \rb + \log d \rb
    \end{equation*}
    where $d := \max \{ d_x, d_{p'}, d_s, d_{p_s} \}$.
    This block-encoding can be implemented using
    \begin{equation*}
        \widetilde{O} \lb N \log \lb \frac{g \alpha_{NVT}}{\epsilon} \rb + \log^{\log 3}{\lb \frac{\alpha_{NVT}}{\epsilon} \rb}  + d \log g \rb
    \end{equation*}
    Toffoli gates, where $g := \max \{ g_x, g_{p'}, g_s, g_{p_s} \}$.
\label{lem:bounds_L_class_NVT}
\end{restatable}

We then apply quantum singular value transformation (QSVT) to the block-encoding of $L_{\tcla}$ to efficiently approximate the exponential $e^{-iL_{\tcla}t}$~\cite{Gilyen2019qsvt}. The idea behind QSVT is to perform polynomial transformations of the singular values of a block-encoded matrix. In our case, we implement polynomial approximations of $\cos (L_{\tcla})$ and $-i\sin (L_{\tcla})$ which can then be added to simulate $e^{-iL_{\tcla}t}$.

Implementing $U_{L_{\tel}}$ is more difficult because we do not generally have an analytic expression for $E_{\tel} \lb\{ x_n \}\rb$ which would be required for constructing a block-encoding of $L_{\tel}$ and subsequently using QSVT.
In principle, one could use quantum phase estimation on the electronic Hamiltonian $H_{\tel}$ to extract the ground state energies at the different nuclear positions and construct a block-encoding of $\Del$ from these (numerical) values. However, the associated computational cost is in $O\lb 1/ \epsilon \rb$ where $\epsilon$ is the desired precision of the simulation. The computational cost of our algorithm on the other hand is only in $O\lb 1/\epsilon^{o(1)} \rb$, which provides a superpolynomial improvement over $O\lb 1/ \epsilon \rb$ scaling.

One important feature of the electronic Liouvillian $L_{\tel}$ is that all summands commute with each other. The evolution operator associated with $L_{\tel}$ can thus be decomposed as follows:

\begin{equation}
\begin{split}
    e^{-iL_{\tel}t} &= e^{-i \lb i\sum_{n=1}^{N} \sum_{j=1}^{3} \Del \otimes D_{p_{n,j}}^1 \rb t} \\
    &= \prod_{n,j} e^{\Del \otimes D_{p_{n,j}}^1 t}.
\end{split}
\label{exp_Lel}
\end{equation}

For simplicity, we use $p_{n,j}$ to refer to either a real or virtual momentum variable.
Let us now explain how to implement a single exponential appearing in Eq.~\eqref{exp_Lel}.
First, we diagonalize the discrete (virtual) momentum derivative operator $D_{p_{n,j}}^1$ by applying the quantum Fourier transform ($\qft$) to the $\ket{p_{n,j}}$ register. Next, we shift the quantum register associated with the nuclear position coordinate $x_{n,j}$ according to the finite difference scheme of $\Del$ and prepare the electronic ground state controlled by the entire nuclear positions register. The electronic ground states are prepared using techniques from~\cite{Lin2020ground_state} together with the following initial state preparation oracle:

\begin{defn}[Initial electronic state preparation oracle]
    Let $U_{H_{\tel} \{ x_n \}}$ be a Hermitian block-encoding of $H_{\tel} \lb\{ x_n \}\rb$ for fixed nuclear positions and let $\ket{\widetilde{\psi}_0 \{ x_n \}}$ be the ground state of $\lb \bra{0} \otimes \mathbb{1} \rb U_{H_{\tel} \{ x_n \}} \lb \ket{0} \otimes \mathbb{1} \rb$. Furthermore, let $0 < \delta \leq 1$. The electronic state preparation oracle $U_I$ is defined via its action on the nuclear positions register $\ket{\{ \hax_n \}}$:
    \begin{equation}
        U_I \ket{\{ \hax_n \}} \ket{0} = \ket{\{ \hax_n \}} \ket{\phi_0 \{ x_n \}},
    \end{equation}
    where $\ket{\phi_0 \{ x_n \}}$ is an initial electronic state that is promised to satisfy $\left| \braket{\widetilde{\psi}_0 \{ x_n \}}{\phi_0 \{ x_n \}} \right| \geq \delta$ for all nuclear configurations. We write $\widetilde{U}_I$ to refer to a variant of the initial electronic state preparation oracle where $\left| \braket{\widetilde{\psi}_0 \{ x_n \}}{\phi_0 \{ x_n \}} \right| \geq \widetilde{\delta}$ with $\widetilde{\delta} \geq \delta$ for all nuclear configurations visited during the simulation. This means that $\widetilde{U}_I$ depends implicitly on the initial nuclear phase space density.
\label{def:state_prep}
\end{defn}
While numerous strategies exist for addressing overlap problems in quantum algorithms for the electronic structure problem, see e.g.~\cite{Reiher2017nitrogenase,Tubman2018,Fomichev2023state_prep}, the overlap issues remain a fundamental problem facing all quantum algorithms within the space and remain an active area of research. Providing an explicit implementation of the initial electronic state preparation oracle is hence beyond the scope of this work.

Controlled by the entire nuclear positions register as well as the Fourier transformed momentum register, we then apply $\exp \lb -iH_{\tel} \{ x_n \} t_{c_k, l_{n,j}} \rb$ to the electronic register holding the corresponding electronic ground state where
\begin{equation}
    t_{c_k, l_{n,j}} := \frac{c_{d_e, k}}{h_x} \frac{\sin{\lb 2\pi l_{n,j}/g_p \rb}}{h_p} t
\end{equation}
is a rescaled time variable depending on the finite difference coefficients $\{c_{d_e, k}\}$ of $\Del$ and the Fourier transform variables $\{l_{n,j}\}$ as explained in more detail in Appendix \ref{app:electronic}. Next, we uncompute the Fourier-transformed momentum register as well as the electronic register and repeat the above procedure for each stencil point of the finite difference formula of $\Del$.

The above method requires access to a discretized electronic Hamiltonian. Instead of utilizing a grid discretization as for the nuclei, we use a finite set of basis functions to discretize the Hilbert space of the electrons. In particular, we choose $B$ plane waves as basis functions, which take the following form in (three-dimensional) position space:
\begin{equation}
    \phi_b (r) := \frac{1}{\sqrt{\Omega}} e^{-i k_b \cdot r} \,.
\end{equation}
$r$ is a vector in position space and $k_b = \frac{2 \pi b}{\Omega^{1/3}}$ is a wave vector in reciprocal space where $b$ is a vector in $\mathbb{Z}^3$ constrained to the cube $G := \big[ -\frac{B^{1/3}-1}{2}, \frac{B^{1/3}-1}{2} \big]^3$. Furthermore, $\Omega \in \Theta \lb B \, h_{\tel}^3 \rb$ is the computational cell volume where $1/h_{\tel}$ is the grid spacing in reciprocal space.

The electronic basis states in first quantization can then be written as $\ket{b_0} \ket{b_1} \cdots \ket{b_{\Tilde{N}-1}}$, where each $\ket{b_j}$ is a qubit register of size $\lceil \log{B} \rceil$ specifying the index $b \in [B]$ of the basis function occupied by electron $j$. The main advantage of using a plane wave expansion of the electronic Hamiltonian is that all terms in the Hamiltonian can be obtained from the nuclear position registers and the plane wave momenta through coherent arithmetic on the quantum computer.
It is shown in Ref.~\cite{Su2021first_quant_sim} that the 1st quantized electronic Hamiltonian in the plane-wave basis takes the following form:
\begin{equation}
\begin{split}
    H_{\tel}^{(\text{pw})} \lb \{x_n\} \rb &:= \sum_{j=1}^{\Tilde{N}}  \sum_{b \in G} \frac{\left \| k_b\right\|^2}{2} \ket{b}\!\bra{b}_{j} \\
     &-\frac{4\pi}{\Omega}\sum_{n=1}^{N}\sum_{j=1}^{\Tilde{N}}\sum_{\substack{b,c\in G\\ b\neq c}}\bigg(Z_n \frac{e^{ik_{c-b}\cdot x_n}}{\norm{k_{b-c}}^2}\bigg)\ket{b}\!\bra{c}_j \\
     &+\frac{2 \pi}{\Omega} \sum_{i\neq j=1}^{\Tilde{N}}\sum_{b,c \in G} \sum_{\substack{\nu\in G_0\\(b+\nu)\in G\\(c-\nu)\in G}}\frac{1}{\left\| k_{\nu}\right\|^2} \ket{b + \nu}\!\bra{b}_i \ket{c-\nu}\!\bra{c}_j,
\end{split}
\label{plane_wave_Hel}
\end{equation}
where $\ket{b}\!\bra{b}_{j}$ acts nontrivially only on the register associated with electron $j$ and similarly for the other terms. Furthermore, $G_0 := \big[ -B^{1/3}, B^{1/3}\big]^3 \subset \mathbb{Z}^3 \backslash \{(0,0,0\}$. Unless stated otherwise, we will write $H_{\tel}$ to refer to $H_{\tel}^{(\text{pw})}$ in the following.

The next Lemma shows that $H_{\tel}$ can be efficiently block-encoded.

\begin{lem}[Block-encoding the electronic Hamiltonian (\cite{Su2021first_quant_sim}, Lemma 1 rephrased)]
    There exists a Hermitian $(\lambda, a_{\tel}, \epsilon)$-block-encoding of the discretized electronic Hamiltonian $H_{\tel}$ with normalization constant
    \begin{equation}
        \lambda \in O \lb \frac{\tn }{h_{\tel}^2} + \frac{N \tn Z_{\text{max}}}{h_{\tel}} + \frac{\tn^2}{h_{\tel}} \rb,
    \end{equation}
    and a number of ancilla qubits
    \begin{equation}
        a_{\tel} \in O \lb \log{\lb \frac{N \tn B}{\epsilon} \rb} \rb.
    \end{equation}
    This block-encoding can be implemented using
    \begin{equation}
        O \lb N + \Tilde{N} + \log{\lb \frac{B}{\epsilon}\rb} \rb
    \end{equation}
    Toffoli gates.
\label{lem:block-encode_Hel}
\end{lem}

Let us briefly discuss the space complexity of the Liouvillian simulation algorithm. On the one hand, we need $O \lb N \log \lb g_x g_p \rb \rb$ qubits to represent the phase space density of the nuclei. We need another $O \lb \tn \log(B) \rb$ qubits to represent the wave function of the electrons. On the other hand, we require a certain number of ancilla qubits for the various block-encodings described above. Since Hamiltonian simulation via QSVT requires only an additional 2 ancilla qubits (see Lemma~\ref{lem:rob_block-Ham_sim}), we find that the overall space complexity is in
\begin{equation}
    O \lb N \log \lb g \rb + \tn \log(B) + \log \lb \frac{\alpha}{\epsilon} \rb + \log (d) \rb,
\end{equation}
where $g = \max \{ g_x, g_{p'}, g_s, g_{p_s} \}$, $\alpha \in \{ \alpha_{NVE}, \alpha_{NVT} \}$ and $d = \max \{ d_x, d_{p'}, d_s, d_{p_s} \}$.
This scaling is very moderate given that it is linear in the total particle number and logarithmic in all other simulation parameters. 

The complexity of our algorithm depends on several user-supplied parameters, which are summarized in Table \ref{tab:NVET_parameters}.

\begin{table}[tb]
\renewcommand{\arraystretch}{1.2}
\setlength{\tabcolsep}{8pt}
    \centering
    \caption{Input parameters that determine the complexity of our quantum algorithm for simulating $NVE$ and $NVT$ Liouvillian dynamics in the Born-Oppenheimer approximation.}
    \begin{tabularx}{\textwidth}{l|X|l}
        \hline \hline
        Description                                                                     & NVE & NVT \\ \hline
        Evolution time                                                                  & $t$ & $t$ \\
        Desired precision                                                      & $\epsilon$ & $\epsilon$\\
        Failure probability                                                             & $\xi$ & $\xi$ \\
        Order of the Trotter product formula                                            & $k$ & $k$\\
        Number of nuclei and electrons                                                  & $N$, $\tn$ & $N$, $\tn$\\
        Mass of the lightest nucleus                                                    & $m_{\text{min}}$ & $m_{\text{min}}$ \\
        Maximum atomic number over all nuclei                                           & $Z_{\text{max}}$ & $Z_{\text{max}}$ \\
        Maximum value of a component of the nuclear position vectors         &$x_{\text{max}}$ & $x_{\text{max}}$  \\
        Maximum value of a component of the (virtual) momentum vectors                  & $p_{\text{max}}$ & $p'_{\text{max}}$  \\
        Grid spacing for a component of the discretized variables                       & $h_x$, $h_{p}$ & $h_x$, $h_{p'}$, $h_s$, $h_{p_s}$ \\
        Order of the finite difference scheme used for approximating derivatives        & $d_x$, $d_{p}$, $d_e$ & $d_x$, $d_{p'}$, $d_s$, $d_{p_s}$, $d_e$\\
        Gap parameter to regularize the Coulomb potential                               &$\Delta$ & $\Delta$ \\
        Number of plane wave basis functions in the electronic Hamiltonian              &$B$ & $B$ \\
        Inverse grid spacing for a component of the electronic wave number              &$h_{\tel}$ & $h_{\tel}$ \\
        Lower bound on the overlap of the initial and true electronic ground state      &$\widetilde{\delta}$ & $\widetilde{\delta}$ \\
        Lower bound on the spectral gap of $H_{\tel}$ during the simulation   & $\widetilde{\gamma}$ &  $\widetilde{\gamma}$ \\
        Upper bound on the higher order derivatives of the electronic energy            & $\chi$ & $\chi$\\
        Number of phase space grid points                             & - & $\eta$ \\
        Number of degrees of freedom of the physical system                             & - & $N_f$\\
        Temperature of the heat bath                                                    & - & $T$\\
        Mass parameter associated with the heat bath                                    & - & $Q$\\
        Minimum value of the bath position variable                          & - & $s_{\text{min}}$ \\
        Maximum value of the bath momentum variable                          & - & $p_{s, \text{max}}$ \\
        \hline \hline
    \end{tabularx}
\label{tab:NVET_parameters}
\end{table}

\begin{algorithm}
    \caption{Liouvillian MD simulation}
    \label{alg:liouvillian}
    \KwIn{Quantum state $\ket{\rho_0}$ encoding initial phase space density of nuclei (+ heat bath for $NVT$ ensemble). \newline
    Input parameters for constructing the Liouvillian operator: $t, \epsilon, \xi, k, N, \tn, \{m_n\}_{n=1}^N, \{Z_n\}_{n=1}^N, x_{\text{max}}, p_{\text{max}}/ p'_{\text{max}}, h_x, h_p/h_{p'}, d_x, d_p/d_{p'}, d_e, \Delta, B, h_{\tel}, \widetilde{\delta}, \widetilde{\gamma}, \chi$. \newline
    Simulations in the $NVT$ ensemble require additional input parameters: $N_f, T, Q, s_{\text{min}}, h_s, h_{p_s}, d_s, d_{p_s}$.}
    \KwOut{An $\epsilon$-precise approximation to $e^{-iL t} \ket{\rho_0}$ with success probability $\geq 1-\xi$.}
    \begin{enumerate}[leftmargin=*]
        \item Compute a set of time steps $\{ t_i\}_{i=1}^{N_{\text{exp}}}$ using the recursive definition of higher-order Trotter product formulas with $N_{\text{exp}} = 2 \times 5^{k-1} + 1$ being the total number of Trotter exponentials~\cite{Berry2006trotter, Childs2021trotter_error}. Additionally, determine the set of indices $I_{\tcla}$ for which $t_i$ corresponds to evolutions under the classical Liouvillian\;
        \item Initialize the nuclear position and momentum registers (and the bath register for $NVT$ simulations), the electronic register, as well as an ancilla register\;
        \item $\epsilon' \gets \epsilon/N_{\text{exp}}$\;
        \item $\xi' \gets \xi/N_{\text{exp}}$\;
    \end{enumerate}
    \For{$1 \leq i \leq N_{\text{exp}}$ }{
        \eIf{$i \in I_{\tcla}$}{
            Apply $\mathtt{ClassLiouvillianEv} \lb t_i, \epsilon', \xi', N, \{m_n\}_{n=1}^N, \{Z_n\}_{n=1}^N, x_{\text{max}}, p_{\text{max}}/p'_{\max}, h_x, h_p/h_{p'}, d_x, d_p/d_{p'}, \Delta \rb$ to the nuclear position and momentum registers (and the bath register in the case of $NVT$ simulations), as well an ancilla register. Simulations in the $NVT$ ensemble also take $N_f, T, Q, s_{\text{min}}, h_s, h_{p_s}, d_s, d_{p_s}$ as input\;
        }{
            Apply $\mathtt{ElectronicLiouvillianEv} \lb t_i, \epsilon', \xi', N, \tn, \{Z_n\}_{n=1}^N,
            x_{\text{max}}, p_{\text{max}}/p'_{\max}, h_x, h_p/h_{p'}, d_e, \Delta, B, h_{\tel}, \widetilde{\delta}, \widetilde{\gamma}, \chi \rb$ to the nuclear position and momentum registers, the electronic register and an ancilla register\;
        }
    }
\end{algorithm}

\begin{algorithm}
    \caption{$\mathtt{ClassLiouvillianEv}$: Evolution under the classical Liouvillian}
    \label{alg:class_ev}
    \KwIn{$t, \epsilon, \xi, N, \{m_n\}_{n=1}^N, \{Z_n\}_{n=1}^N, x_{\text{max}}, p_{\text{max}}/p'_{\max}, h_x, h_p/h_{p'}, d_x, d_p/d_{p'}, \Delta$. \newline
    Simulations in the $NVT$ ensemble require additional input parameters: $N_f, T, Q, s_{\text{min}}, h_s, h_{p_s}, d_s, d_{p_s}$.}
    \KwOut{An $\epsilon$-precise approximation to $e^{-iL_{\tcla} t}$ with success probability $ \geq 1-\xi$.}
    \begin{enumerate}[leftmargin=*]
        \item Construct an $\epsilon/(2t)$-precise block-encoding of $L_{\tcla}$ as shown in Appendix \ref{app:L_class}\;
        \item  Apply the QSVT based Hamiltonian simulation algorithm from~\cite{Gilyen2019qsvt} to the block-encoding of $L_{\tcla}$ to obtain an $\epsilon$-precise block-encoding of $e^{-iL_{\tcla} t}$\;
        \item Use $O \lb \log \lb \frac{1}{\xi} \rb \rb$ rounds of fixed-point amplitude amplification to boost the success probability to at least $1- \xi$\;
    \end{enumerate}
\end{algorithm}

\begin{algorithm}
    \caption{$\mathtt{ElectronicLiouvillianEv}$: Evolution under the electronic Liouvillian}
    \label{alg:el_ev}
    \KwIn{$t, \epsilon, \xi, N, \tn, \{Z_n\}_{n=1}^N, x_{\text{max}}, p_{\text{max}}/p'_{\max}, h_x, h_p/h_{p'}, d_e, \Delta, B, h_{\tel}, \delta, \gamma, \chi$.}
    \KwOut{An $\epsilon$-precise approximation to $e^{-iL_{\tel} t}$ with success probability $\geq 1-\xi$.}
    Initialize the electronic register in the $\ket{0}$ state\;
    \For{$1 \le n \le N$}{
    \For{$1 \le j \le 3$}{
        Fourier transform the quantum register of the $j$-th momentum coordinate of the $n$-th nucleus:
        \begin{equation*}
            \qft \ket{\hap_{n,j}} = \frac{1}{\sqrt{g_p}} \sum_{l_{n,j}} e^{2\pi i \hap_{n,j} l_{n,j}/g_p}\ket{l_{n,j}}
        \end{equation*}
        $\ket{\hax_{n,j}} \gets \ket{\hax_{n,j} - d_e}$\;
        \For{$-d_e \le k \le d_e$}{
            \begin{equation*}
                t_{c_{d_e, k}, l_{n,j}} \gets \frac{c_{d_e, k}}{h_x} \frac{\sin{\lb 2\pi l_{n,j}/g_p \rb}}{h_p} t
            \end{equation*}
            Apply the state preparation oracle $U_I$ from Definition \ref{def:state_prep} to the entire nuclear positions register $\ket{\{x_n\}}$ and the electronic register: 
            \begin{equation*}
                U_I \ket{\{ x_n \}} \ket{0} = \ket{\{ x_n \}} \ket{\phi_0 \{ x_n \}}
            \end{equation*}
            Apply the ground state preparation algorithm from~\cite{Lin2020ground_state} to $\ket{\{ x_n \}} \ket{\phi_0 \{ x_n \}}$ 
            using a block-encoding of $H_{\tel}\lb \{x_n\} \rb$\;
            Apply $\exp \lb -i H_{\tel}\lb \{x_n\} \rb t_{c_{d_e, k}, l_{n,j}} \rb$ to the electronic register\;
            Uncompute the (approximate) electronic ground state\;
            $\ket{\hax_{n,j}} \gets \ket{\hax_{n,j} +1}$\;
        }
            $\ket{\hax_{n,j}} \gets \ket{\hax_{n,j} - d_e}$\;
        Apply $\qft^{-1}$ to the Fourier transformed momentum register to switch back to the $\ket{\hap_{n,j}}$ basis\;
    }
    }
    Use $O \lb \log \lb \frac{1}{\xi} \rb \rb$ rounds of fixed-point amplitude amplification to boost the success probability to at least $1- \xi$\;
\end{algorithm}

\subsection{Estimating the free energy}
\label{sec:free_energy}

Let us now discuss how to estimate the free energy of the nuclei after time $t$ (see Definition~\ref{def:free_energy}) using our Liouvillian simulation algorithm. 
At this stage, we do not assume that the nuclei are in thermal equilibrium.
First, we apply $U_{L_{NVT}}$ to the initial discretized phase space density of the nuclei and the heat bath to evolve them for time $t$. The main idea is to estimate the Gibbs entropy and the internal energy associated with $\ket{\rho_t} = e^{-iL_{NVT} t}\ket{\rho_0}$ separately and then add the results classically to estimate the free energy. This means we require at least two separate simulations.
In Appendix \ref{app:free} we show how to reduce the problem of estimating the Gibbs entropy of the nuclei to the problem of estimating the von Neumann entropy of a density matrix $\rho_{\text{sys}}'$ obtained from $\ket{\rho_t}$ by tracing out the bath register and removing the off-diagonal elements. This allows us to employ Theorem 13 of~\cite{Gilyen2019dist_prop_test}. Their algorithm requires access to a purification of the system's density matrix, which in our case is essentially just $\ket{\rho_t}$. From that purification, we first construct a block-encoding of $\rho_{\text{sys}}'$ and then use QSVT to transform the singular values $\rho_i$ of $\rho_{\text{sys}}'$ via a polynomial approximation to $\ln \lb 1/\rho_i \rb$. The resulting operator is then applied to the purification of $\rho_{\text{sys}}'$. Lastly, we use amplitude estimation to obtain an estimate of 
\begin{equation}
    \mathcal{S}_G = -k_b \, \text{Tr} \lb \rho_{\text{sys}}' \ln \rho_{\text{sys}}' \rb.
\end{equation}

Next, let us discuss how to estimate the internal energy associated with $\ket{\rho_t}$. First, note that a classical system can be described by a density matrix $\rho$ and a Hamiltonian $H$ both of which are diagonal in the computational basis. The internal energy of a classical system can thus be computed as follows:
\begin{equation}
    \mathcal{U} = \text{Tr} \lb  \rho H \rb.
\end{equation}
In our case, we have that $\rho \equiv \rho_{\text{sys}}'$ and $H \equiv H_{\text{nuc}} = H_{\text{kin}} + H_{\text{pot}} + H_{E_{\tel}}$.
In Appendix \ref{app:free} we show how to efficiently block-encode each of the three terms as given in Eqs.~\eqref{H_kin}, \eqref{H_pot} and \eqref{H_E_el}.
The idea then is to use the Hadamard test to estimate $\text{Tr} \lb \rho_{\text{sys}}'  H_{\text{kin}} \rb$, $\text{Tr} \lb \rho_{\text{sys}}'  H_{\text{pot}} \rb$ and $\text{Tr} \lb \rho_{\text{sys}}' H_{E_{\tel}} \rb$ individually and combine the results classically.

Let us now assume that the dynamics of the extended system (nuclei + heat bath) are ergodic, i.e.~that the extended system samples all phase space points associated with energy $E_{\text{ext}}$. This allows us to estimate the equilibrium free energy via coherent time averaging. More specifically, we first prepare the following time-averaged density matrix:
\begin{equation}
    \overline{\rho} := \frac{1}{t} \int_0^t e^{-iL_{NVT}\tau} \ketbra{\rho_0}{\rho_0} e^{iL_{NVT}\tau} d\tau,
\end{equation}
where $\ket{\rho_0}$ is an initial phase space density of the extended system that has support only on configurations with energy $E_{\text{ext}}$.
Operationally speaking, the above density matrix can be prepared by sampling $t' \in [0,t]$ uniformly at random and applying $e^{-iL_{NVT}t'}$ to $\ket{\rho_0}$.
If $t$ is sufficiently large then expectation values of observables estimated with $\overline{\rho}$ are approximately equal to expectation values computed with $\rho_{NVT}$ \cite{Short2012equilibration}.

Note that having access to an initial phase space density describing a microcanonical ensemble in the extended system, i.e. $\rho_0 \propto \delta \lb H_{\nuc}^{(NVT)} - E_{\text{ext}}\rb$, allows us to directly prepare the corresponding Boltzmann distribution over the nuclear variables by tracing out the bath variables, see Eq.~\eqref{partition_func}. This implies that we could estimate the free energy as an ensemble average without having to perform any time evolution.

\section{Conclusion}
\label{sec:conclusion}

Our main achievement is a new approach for efficiently simulating coupled quantum-classical dynamics on fault-tolerant quantum computers which provides a super-polynomial improvement in the precision scaling over previous work. 
The upper bounds on the computational costs of our algorithm for the evolution of a classical phase space density scales polynomially with the 1-norm of the Liouvillian and with the simulation time $t$. This is in stark contrast to earlier gradient-based approaches~\cite{Obrien2022}, which we show in Appendix~\ref{app:EulerCost} can scale under worst-case assumptions exponentially with the evolution time. The presented Liouvillian simulation algorithm illustrates the value of incorporating classical dynamics into quantum simulations coherently on fault-tolerant quantum computers and paves the way for simulating coupled quantum-classical systems. We apply the approach to the simulation of molecular systems in both the microcanonical and canonical ensembles and to the estimation of thermodynamic quantities, such as the free energy. 

To make our algorithms applicable to practical problems, challenges and limitations remain. 
For example, preparing the classical system in the canonical ensemble requires it to thermalize. In classical simulations, it is possible to have clear indicators of thermalization \cite{baras2011molecular,lamoureux2003modeling,bussi2007canonical}, while it is unclear how to estimate those indicators within our approach without sampling. Also, the computational cost of the free energy estimation scales exponentially with the number of particles because of the growth of the phase space \cite{Sorin2000}. 
Another challenge is the preparation of a quantum state encoding the initial phase space density of the classical particles since preparing an arbitrary quantum state can take time that scales exponentially with the number of qubits in the worst case.
Additionally, artifacts affecting classical simulation~\cite{harvey1998} will likely influence our simulation methodology. 

Compared to state-of-the-art classical MD, where many refinements have been developed over the years, our approach is still rudimentary. Several solutions could be explored for solving the above issues and adapted to complement our approach.
For example, exploiting adaptive time steps to improve the computational costs or using multiple coupled thermostats to allow correct thermalization \cite{marx2000ab}.

Future work must optimize our resource scalings, which are based on loose bounds. In terms of quantum circuit design, numerous improvements are possible. One example is the extensive use of products of block-encodings to simulate the parts of the Liouvillian. These costs could likely be brought down by designing a specific block-encoding in a single step. Similarly, it is an open question whether combining the classical and electronic Liouvillians in a higher-order Trotter formula is the most efficient choice. Exploiting the fractional query model~\cite{Berry2014alter_sign_trick} or multi-product formulas may lead to polylogarithmic scaling with the error tolerance rather than the sub-polynomial scalings obtained using high-order Trotter formulas. Another important task is the identification of a first potential application, along with assessing its run time and qubit count. This will aid in pinpointing bottlenecks within the proposed algorithm, as well as enable a comparison with alternative approaches. It would also be interesting to move beyond the Born-Oppenheimer approximation by adapting our algorithm to include excited electronic states.

Looking forward, a larger question emerges about the role that quantum computers may play in the simulation of classical or quantum-classical dynamics.  While our research strengthens our understanding of the advantages that quantum may provide for simulating such hybrid dynamics, it is still necessary to fully explore what are the limitations and opportunities that quantum computers face or provide when simulating both types of dynamical systems~\cite{Joseph2020, Jin2022, Babbush2023, Jin2023}.  Our belief is that further study of such applications will unveil a host of new use cases for quantum computers that lie outside of purely quantum simulations and, in turn will reveal that quantum computation is much more broadly applicable to simulation than previously thought.

\section*{Acknowledgements}
SS acknowledges support from a Research Award from Google Inc., NSERC Discovery Grants, as well as support from Boehringer Ingelheim.  NW's work on this project was supported by the, “Embedding Quantum Computing into Many-body
Frameworks for Strongly Correlated Molecular and Materials Systems” project, which is funded by the U.S. Department of Energy (DOE), Office of Science, Office of Basic
Energy Sciences, the Division of Chemical Sciences, Geosciences, and Biosciences. We thank Benjamin Ries, Aniket Magarkar, Thomas Fox and Rodrigo Ochoa for their insights and stimulating discussions on molecular dynamics applications. The authors thank Christofer Tautermann and Clemens Utschig-Utschig for useful discussions and feedback. Additionally, we thank Ryan Babbush and Tom O'Brien for their comments on the manuscript.

\bibliography{references}

\appendix
\label{appendix}

\section{Evolution under the classical Liouvillian}
\label{app:L_class}

We simulate $e^{-iL_{\tcla}t}$ using qubitization/QSVT~\cite{Low2019hamiltonian, Gilyen2019qsvt}, which requires us to prepare a block-encoding of $L_{\tcla}$ according to Definition \ref{def:block-encoding}. These block encodings are implemented using the linear combination of unitaries (LCU) framework~\cite{Childs2012lcu}. First, for an arbitrary matrix $A$ we decompose it into a linear combination of unitaries, $A = \sum_{j=0}^{2^a-1} \alpha_j U_j$ where $\alpha_j \ge 0 \,\, \forall j$. This linear combination can then be implemented using the following two unitary operations defined via their action on an ancilla register initialized to $\ket{0}^{\otimes a}$ and some quantum state $\ket{\psi}$:
\begin{align}
    \prep \ket{0}^{\otimes a}\ket{\psi} := \sum_j \sqrt{\frac{\alpha_j}{\alpha}}\ket{j}\ket{\psi} \label{prep}\\
    \sel \ket{j}\ket{\psi} := \ket{j}U_j\ket{\psi},
\end{align}
where $\alpha:= \sum_j \alpha_j$ is a normalization constant.
This allows us to implement $A$ probabilistically in the sense that
\begin{equation}
    \frac{A}{\alpha} = \lb\bra{0}^a\otimes \mathbb{1}\rb  \prep^\dagger \cdot \sel \cdot \prep \lb\ket{0}^a\otimes \mathbb{1}\rb.
\end{equation}

Once we have such a block-encoding of the classical Liouvillian $L_{\tcla}$, we can use QSVT to construct a polynomial approximation of $e^{-iL_{\tcla}t}$~\cite{Gilyen2019qsvt, Low2019hamiltonian}. The corresponding query complexity of block-Hamiltonian, or in our case block-Liouvillian, simulation based on qubitization/QSVT is stated below. 

\begin{lem}[Robust block-Hamiltonian simulation~\cite{Gilyen2019qsvt, Martyn2023coherent_sim}]
    Let $t \in \mathbb{R}_{\geq 0}$,  $\epsilon \in (0,1)$ and let $U$ be an $(\alpha, a, \epsilon/2t)$-block-encoding of the Hamiltonian $H$. Then we can implement an $\epsilon$-precise Hamiltonian simulation unitary $V$ which is an $(1, a+2, \epsilon)$-block-encoding of $e^{iHt}$, with probability of success at least $1-\xi$, with
    \begin{equation}
        O \lb \log{\lb \frac{1}{\xi} \rb} \lb \alpha t + \log \lb \frac{1}{\epsilon} \rb \rb \rb
    \end{equation}
    uses of $U$ or its inverse, 3 uses of controlled-$U$ or its inverse, using $O \lb \log{\lb \frac{1}{\xi} \rb} \lb a \alpha t + a \log \lb \frac{1}{\epsilon} \rb \rb \rb$ two-qubit gates and using $O(1)$ ancilla qubits.
\label{lem:rob_block-Ham_sim}
\end{lem}
A proof of Lemma \ref{lem:rob_block-Ham_sim} with constant success probability is given in~\cite{Gilyen2019qsvt}. 
Using (fixed-point) amplitude amplification, we can boost the success probability to $1-\xi$ at the expense of a multiplicative factor of $\log{\lb \frac{1}{\xi} \rb}$ \cite{Martyn2023coherent_sim}.

The following Lemma bounds the error propagation of the $\prep$ subroutine which will be useful for our discussion of the overall block-encoding error of the classical Liouvillian.
\begin{lem}[Error propagation of $\prep$]
    Let $A = \sum_{j=0}^{2^a-1} \alpha_j U_j$ be a a linear combination of unitaries with $\alpha_j \ge 0 \,\, \forall j$. Let $\ket{\prep} := \sum_j \sqrt{\frac{\alpha_j}{\alpha}}\ket{j}$ 
    with $\alpha := \sum_j \alpha_j$ be the quantum state prepared by the $\prep$ subroutine as defined in Eq.~\eqref{prep}. Let $\ket{\widetilde{\prep}} := \sum_j c_j \ket{j}$ be an $\epsilon/(2 \alpha \sqrt{2^a})$-precise approximation to $\ket{\prep}$ prepared by $\widetilde{\prep}$ such that $\norm{\ket{\widetilde{\prep}} - \ket{\prep}} \leq \epsilon/(2 \alpha \sqrt{2^a})$. Given access to the unitary $\sel := \sum_j \ketbra{j}{j} \otimes U_j$,
    we can implement an $\epsilon$-precise block-encoding of $A$.
\label{lem:prep_error}
\end{lem}

\begin{proof}
    First, recall that for any $v \in \mathbb{C}^{2^a}$ it holds that $\norm{v}_1 = \sum_{j=0}^{2^a-1} |v_j| \leq \sqrt{2^a} \norm{v}_2 \equiv \sqrt{2^a} \norm{v}$. Using this inequality and the triangle inequality, one finds that
    \begin{equation}
    \begin{split}
        &\norm{\alpha \lb\bra{0}^a\otimes \mathbb{1}\rb  \widetilde{\prep}^\dagger \cdot \sel \cdot \widetilde{\prep} \lb\ket{0}^a\otimes \mathbb{1}\rb - A} \\ 
        &= \norm{\alpha \lb\bra{0}^a\otimes \mathbb{1}\rb  \widetilde{\prep}^\dagger \cdot \sel \cdot \widetilde{\prep} \lb\ket{0}^a\otimes \mathbb{1}\rb - \alpha \lb\bra{0}^a\otimes \mathbb{1}\rb  \prep^\dagger \cdot \sel \cdot \prep \lb\ket{0}^a\otimes \mathbb{1}\rb} \\
        &= \alpha \norm{\sum_j c_j c_{j}^* U_j - \sum_j \frac{\alpha_j}{\alpha} U_j} 
        \leq \alpha \sum_j \left|c_j c_{j}^* - c_j \sqrt{\frac{\alpha_j}{\alpha}} + c_j \sqrt{\frac{\alpha_j}{\alpha}}  - \frac{\alpha_j}{\alpha}\right|\norm{U_j} \\
        &\leq \alpha \sum_j \left|c_j\right| \left|c_j^* - \sqrt{\frac{\alpha_j}{\alpha}}\right| + \alpha \sum_j \left|\sqrt{\frac{\alpha_j}{\alpha}}\right| \left|c_j - \sqrt{\frac{\alpha_j}{\alpha}}\right|
        \leq 2 \alpha \sum_j \left|c_j - \sqrt{\frac{\alpha_j}{\alpha}}\right| \\
        &\leq 2 \alpha \sqrt{2^a} \sqrt{\sum_j \left|c_j - \sqrt{\frac{\alpha_j}{\alpha}}\right|^2} =  2 \alpha \sqrt{2^a} \norm{\ket{\widetilde{\prep}} - \ket{\prep}} \leq \epsilon.
    \end{split}
    \end{equation}
\end{proof}

Lemmas \ref{lem:bounds_L_class_NVE} and \ref{lem:bounds_L_class_NVT} provide upper bounds on the complexity of block-encoding the classical Liouvillian in the $NVE$ and $NVT$ ensemble, respectively.
We prove these lemmas by explicitly constructing a block-encoding of the relevant $L_{\tcla}$. The general idea is to block-encode each term of $L_{\tcla}$ separately and then combine all the smaller block-encodings to obtain a block-encoding of the overall classical Liouvillian. To do so, we need to know how to multiply two block-encodings.

\begin{lem}[Product of block-encoded matrices (\cite{Gilyen2019qsvt}, Lemma 30)]
    If $U$ is an $(\alpha, a, \delta)$-block-encoding of an s-qubit operator $A$ and $V$ is a $(\beta, b, \epsilon)$-block-encoding of an s-qubit operator $B$, then $(I_a \otimes U)(I_b \otimes V)$ is an $(\alpha \beta, a+b, \alpha \epsilon + \beta \delta)$-block-encoding of $AB$.
\label{lem:product_block_encoding}
\end{lem}

A linear combination of block-encodings can be constructed using the concept of a state preparation pair.

\begin{defn}[State preparation pair (\cite{Gilyen2019qsvt}, Definition 28)]
    Let $y \in \mathbb{C}^m$ and $\norm{y}_1 \leq \beta$. The pair of unitaries $\lb P_L, P_R \rb$ is called a $\lb \beta, b, \epsilon \rb$-state-preparation-pair if $P_L \ket{0}^{\otimes b} = \sum_{j=0}^{2^b-1} c_j \ket{j}$ and $P_R \ket{0}^{\otimes b} = \sum_{j=0}^{2^b-1} d_j \ket{j}$ such that $\sum_{j=0}^{m-1} \left| \beta c_j^* d_j - y_j \right| \leq \epsilon$ and for all $j \in \{ m, \dots, 2^b-1 \}$ we have $c_j^* d_j =0$.
\end{defn}

Note that $b$ in the above definition is chosen such that $2^b \geq m$. This is necessary to accommodate all $m$ entries of $y$. In general, $m$ does not need to be a power of 2. The condition $c_j^* d_j =0$ for all $j \in \{ m, \dots, 2^b-1 \}$ ensures that we are limited to an $m$-dimensional subspace of the $2^b$-dimensional space of the $b$ register.

For our purposes, it will always be true that $P_L = P_R$ in which case we call $P_L$ a state preparation unitary.

\begin{lem}[Linear combination of block-encoded matrices (improved version of Lemma 29, appeared in~\cite{Gilyen2019qsvt}).]
    Let $A = \sum_{j=0}^{m-1} y_j A_j$ be an $s$-qubit operator and $\epsilon \in \mathbb{R}_{> 0}$. Suppose that $(P_L, P_R)$ is a $(\beta, b, \epsilon_1)$-state-preparation pair for $y$, $W = \sum_{j=0}^{m-1} \ket{j}\!\bra{j} \otimes U_j + ((I - \sum_{j=0}^{m-1} \ket{j}\!\bra{j}) \otimes I_a \otimes I_s$ is an $(s + a + b)$-qubit unitary such that $\forall j \in \{0, 1, \dots, m-1\}$ we have that $U_j$ is an $(\alpha, a, \epsilon_2)$-block-encoding of $A_j$. Then we can implement an $(\alpha \beta, a + b, \alpha \epsilon_1 + \beta \epsilon_2)$-block-encoding of $A$, with a single use of $W$, $P_R$ and $P_L^\dagger$.
\label{lem:sum_block_encoding}
\end{lem}

\begin{proof}
    We adapt the proof from~\cite{Gilyen2019qsvt} by showing that $\widetilde{W}:= \lb P_L^{\dagger} \otimes I_a \otimes I_s \rb W \lb P_R \otimes I_a \otimes I_s \rb$ is an $\lb \alpha \beta , a + b, \alpha \epsilon_1 + \beta \epsilon_2 \rb$-block-encoding of $A$:
    \begin{equation*}
    \begin{split}
        \norm{A - \alpha \beta \lb \bra{0}^{\otimes b} \otimes \bra{0}^{\otimes a} \otimes I \rb \widetilde{W} \lb \ket{0}^{\otimes b} \otimes \ket{0}^{\otimes a} \otimes I \rb} &= \norm{A - \alpha \sum_{j=0}^{m-1} \beta \lb c_j^* d_j \rb \lb \bra{0}^{\otimes a} \otimes I \rb U_j \lb \ket{0}^{\otimes a} \otimes I \rb} \\
        &\leq \alpha \epsilon_1 + \norm{A - \alpha \sum_{j=0}^{m-1} y_j \lb \bra{0}^{\otimes a} \otimes I \rb U_j \lb \ket{0}^{\otimes a} \otimes I \rb} \\
        &\leq \alpha \epsilon_1 + \sum_{j=0}^{m-1} y_j \norm{A_j - \alpha \lb \bra{0}^{\otimes a} \otimes I \rb U_j \lb \ket{0}^{\otimes a} \otimes I \rb} \\
        &\leq \alpha \epsilon_1 + \sum_{j=0}^{m-1} y_j \epsilon_2 \\
        &\leq \alpha \epsilon_1 + \beta \epsilon_2.
     \end{split}
    \end{equation*}
\end{proof}

Additionally, we require a bound on the coefficients of higher-order central finite difference approximations arising from the discretized derivative operators $D_x$ and $D_p$ for the $NVE$ ensemble and $D_x$, $D_{p'}$, $D_s$ and $D_{p_s}$ for the $NVT$ ensemble. The following result provides a higher order bound on centered difference formulas that can be used for our purposes.

\begin{lem}[Higher-order central finite difference approximation~\cite{Li2005finite_difference}]
    Let $d$ be an integer and $\{x_k\}$ with $k \in \{-d, -d+1, \dots, d-1, d\}$ a set of equally spaced points on the interval $[a,b]$, i.e.
    \begin{equation}
        x_k = x_0 + h
    \end{equation}
    for some $h > 0$. Furthermore, let $f \in C^{2d+1}[a,b]$ be a function on the interval $[a,b]$ which is $2d+1$ times continuously differentiable. Then one can use a linear combination of $f(x_k)$ to construct a central finite difference formula of order $2d$ to approximate $f^{(1)}(x_0) = f'(x_0)$, i.e.
    \begin{equation}
        f'(x_0) = \frac{1}{h} \sum_{k = -d}^{d} c_{d,k} f(x_k) + R_d(x_0)
    \end{equation}
    where
    \begin{align}
        c_{d,k} := 
        \begin{cases}
            \frac{(-1)^{k+1}(d!)^2}{k(d-k)!(d+k)!} ,& \text{if } k\neq 0\\
            0 ,              & \text{else}.
        \end{cases}
    \label{fin_diff_coeff}
    \end{align}
    The remainder term can be bounded as follows:
    \begin{equation}
        R_d(x_0) \in O \left( \max_{x \in [a,b]} \left| f^{(2d+1)}(x) \right| \left(  \frac{e h}{2} \right)^{2d} \right).
    \label{FDerror}
    \end{equation}
\label{lem:finite_diff}
\end{lem}

In order to provide the cost of block-encoding the result, we need to identify the sum of the coefficients for the finite difference formula.  Such a bound is given in the following lemma.

\begin{lem}
    The coefficients $c_{d,k}$ of the central finite difference formula as defined in Lemma \ref{lem:finite_diff} satisfy
    \begin{equation}
        \sum_{k = -d}^{d} |c_{d,k}| \le 2 \left( \ln{d} + 1 \right).
    \end{equation}
\label{lem:bound_FD_coeff}
\end{lem}

\begin{proof}
    First, note that for $k \neq 0$ we have
    \begin{equation}
        \frac{(d!)^2}{(d-k)!(d+k)!} = \frac{d \cdot (d-1) \cdots (d-k+1)}{(d+k) \cdot (d+k-1) \cdots (d+1)} < 1.
    \end{equation}
    Thus,
    \begin{equation}
        |c_{d,k}| = \left| \frac{(-1)^{k+1}}{k}\frac{(d!)^2}{(d-k)!(d+k)!} \right| \le \left| \frac{1}{k} \right|.
    \end{equation}
    This implies that
    \begin{equation}
        \sum_{k = -d}^{d} |c_{d,k}| \le \sum_{\substack{k = -d \\ i \neq 0}}^{d} \Big|\frac{1}{k} \Big| = 2 \sum_{k=1}^{d} \frac{1}{k} \le 2 \left( \ln{d} + 1 \right).
    \end{equation}
\end{proof}

Finally, let us show how to construct a general inequality testing circuit, which will be used repeatedly as a subroutine for block-encoding $L_{\tcla}$.

\begin{lem}[Inequality testing]
    Let $a$ and $b$ be arbitrary bitstrings of length $n$. Using $n+2$ additional qubits and $5n-2$ Toffoli gates, one can construct a quantum circuit that outputs `0' iff $a \le b$ and `1' otherwise.
\label{lem:inequal}
\end{lem}
\begin{proof}
    Consider the circuit shown in Fig.~\ref{fig:inequality_circuit} with bitstrings $a$ and $b$ as inputs. The general strategy is to perform bitwise comparisons starting with the most significant bits and to store the result in an additional qubit $\ket{r}$. To avoid overwriting the result of the previous bit comparison, we use an additional $n$ qubits, $\ket{c_0}, \dots \ket{c_{n-1}}$ as a clock register. Furthermore, we need one ancilla qubit to implement triple controlled-NOT gates from Toffoli gates. Note that this ancilla qubit is not shown in the circuit diagram. We need one additional qubit to store the result of the inequality test. The state of that qubit is denoted $\ket{r}$ in Fig.~\ref{fig:inequality_circuit}.
    The circuit first compares the most significant bits, $a_0$ and $b_0$ using a CNOT gate. The second (triple controlled-NOT) gate fires only if initially $a_0 = 1$ and $b_0 = 0$, i.e.~if $a_0 > b_0$. In this case, the last qubit, $\ket{r}$, which stores the result of the inequality test, gets flipped to $\ket{1}$. None of the remaining triple controlled-NOT gates fire since the indicator state of the clock register, $\ket{c_0} = \ket{1}$, does not get transferred to the next clock qubit $\ket{c_1}$.
    The same is true for the case $a_0 < b_0$, i.e.~if $a_0 = 0$ and $b_0 = 1$. However, in this case, not even the first triple controlled-NOT gate fires. The indicator state of the clock register gets swapped to $\ket{c_1}$ iff $a_0 = b_0$. This is repeated until $a_j \neq b_j$ for some $j \in [n]$. In the worst case, $j=n$. At the end of the inequality test, we uncompute intermediate results by applying all gates except for the triple-controlled NOT-gates in reverse. Using the fact that a single triple-controlled CNOT gate can be implemented with 3 Toffoli gates we then find that the overall Toffoli count is equal to $5n-2$.
    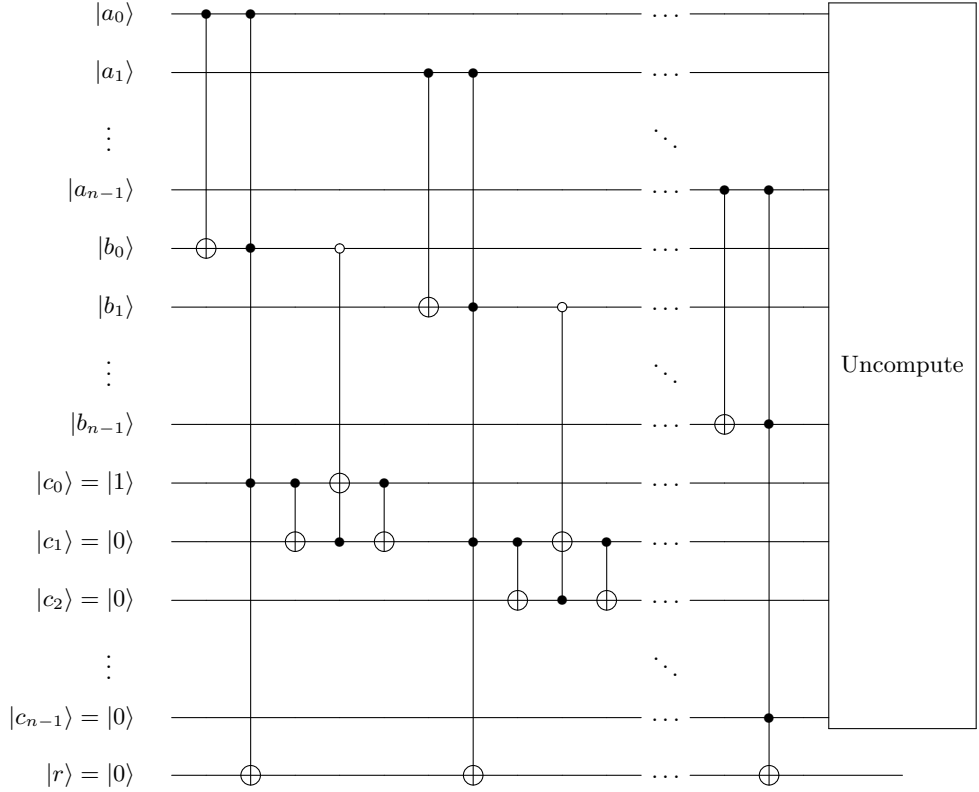
\begin{figure}
        \centering
        \mbox{
        \Qcircuit @C=1em @R=1.5em {
            \lstick{\ket{a_0}} && \ctrl{4} & \ctrl{4} & \qw & \qw & \qw & \qw & \qw & \qw & \qw & \qw  & \qw & \dots & & \qw & \qw & \qw & \multigate{12}{\text{Uncompute}} \\
            \lstick{\ket{a_1}} && \qw & \qw & \qw & \qw & \qw & \ctrl{4} & \ctrl{4} & \qw & \qw & \qw  & \qw & \dots & & \qw & \qw & \qw & \ghost{\text{Uncompute}}\\
            \lstick{\vdots \hspace{0.3cm}}&&  &  &  &  &  & &  & & &  &  & \ddots & & & & &  \pureghost{Uncompute} \\
            \lstick{\ket{a_{n-1}}}&& \qw & \qw & \qw & \qw & \qw & \qw & \qw & \qw & \qw & \qw  & \qw & \dots & & \ctrl{4} & \ctrl{4} & \qw & \ghost{\text{Uncompute}} \\
            \lstick{\ket{b_0}} && \targ & \ctrl{4} & \qw & \ctrlo{4} & \qw & \qw & \qw & \qw & \qw & \qw & \qw & \dots & & \qw & \qw & \qw & \ghost{\text{Uncompute}} \\
            \lstick{\ket{b_1}} && \qw & \qw & \qw & \qw & \qw & \targ & \ctrl{4} & \qw & \ctrlo{4} & \qw & \qw & \dots & & \qw & \qw & \qw & \ghost{\text{Uncompute}}\\
            \lstick{\vdots \hspace{0.3cm}}&&  &  &  &  &  &  &  &  &  &   &  & \ddots & & & &  & \pureghost{Uncompute}\\
            \lstick{\ket{b_{n-1}}}&& \qw & \qw & \qw & \qw & \qw & \qw & \qw & \qw & \qw & \qw & \qw & \dots & & \targ &  \ctrl{5} & \qw & \ghost{\text{Uncompute}}\\
            \lstick{\ket{c_0} = \ket{1}} &&  \qw & \ctrl{5} & \ctrl{1} & \targ & \ctrl{1} & \qw & \qw & \qw & \qw & \qw & \qw & \dots & & \qw & \qw & \qw & \ghost{\text{Uncompute}}\\
            \lstick{\ket{c_1} = \ket{0}} && \qw & \qw & \targ & \ctrl{-1} & \targ & \qw & \ctrl{4} & \ctrl{1} & \targ & \ctrl{1} & \qw & \dots & & \qw & \qw & \qw & \ghost{\text{Uncompute}}\\
            \lstick{\ket{c_2} = \ket{0}} && \qw & \qw & \qw & \qw & \qw & \qw & \qw & \targ & \ctrl{-1} & \targ & \qw & \dots & & \qw & \qw & \qw & \ghost{\text{Uncompute}}\\
            \lstick{\vdots \hspace{0.3cm}}&&  &  &  &  &  &  &  &  &  &   &  & \ddots & & & & & \pureghost{Uncompute} \\
            \lstick{\ket{c_{n-1}} = \ket{0}}&&  \qw & \qw & \qw & \qw & \qw & \qw & \qw & \qw & \qw & \qw & \qw & \dots & & \qw & \ctrl{1} & \qw & \ghost{\text{Uncompute}}\\
            \lstick{\ket{r} = \ket{0}}&& \qw & \targ & \qw & \qw & \qw & \qw & \targ & \qw & \qw & \qw & \qw & \dots & & \qw & \targ & \qw & \qw \\
        }
        }
        \caption{Inequality testing circuit $U_{\leq}$. The basic idea behind this circuit is to perform bit by bit comparison and use an ancillary qubit, $\ket{r}$, in which the outcome of the comparison is stored. The additional n-qubit register $\ket{c}$ is used to avoid overwriting the results of the other bit comparison results.  The ``Uncompute'' part consists of all previous gates applied in reverse except for the triple-controlled NOT-gates which are not uncomputed.
        If $a \leq b$ then the final output is $\ket{r} = \ket{0}$. Otherwise, $\ket{r} = \ket{1}$.}
        \label{fig:inequality_circuit}
    \end{figure}
\end{proof}

Now we are ready to prove Lemmas \ref{lem:bounds_L_class_NVE} and \ref{lem:bounds_L_class_NVT}.

\subsection{Proof of Lemma \ref{lem:bounds_L_class_NVE}}
\label{sec:proof_lem_NVE}

For convenience, let us restate Lemma \ref{lem:bounds_L_class_NVE} here.

\LclassNVE*

\begin{proof}
    We block-encode $L_{\tcla}^{(NVE)}$ via several layers of ``smaller'' block-encodings which can be thought of as a block-encoding hierarchy. More precisely, we apply different $\prep$ and $\sel$ operations in a nested fashion as summarized in Fig.~\ref{fig:hierarchy}.
    Let us first give an overview of all the levels of the hierarchy before discussing the gate and space complexity.

    At the lowest level, we have four types of block-encodings:
    \begin{enumerate}
        \item $U_{p_{n,j}}$, an $\lb \alpha_p, a_p, \epsilon_p \rb$-block-encoding of $\sum_{\hap_{n,j}=0}^{g_p-1} p_{n,j} \ketbra{\hap_{n,j}}{\hap_{n,j}}$,
        \item $U_{D_{x_{n,j}}}$, an $\lb \alpha_{D_x}, a_{D_x}, \epsilon_{D_x} \rb$-block-encoding of $D_{x_{n,j}}$,
        \item $U_{V_{n,n',j}}$, an $\lb \alpha_V, a_V, \epsilon_V \rb$-block-encoding of $\sum_{\hax_{n,j}=0}^{g_x-1} \frac{(x_{n,j} - x_{n',j})}{\lb \norm{x_n - x_{n'}}^2 + \Delta^2 \rb^{3/2}} \ketbra{\hax_{n,j}}{\hax_{n,j}}$,
        \item $U_{D_{p_{n,j}}}$, an $\lb \alpha_{D_p}, a_{D_p}, \epsilon_{D_p} \rb$-block-encoding of $D_{p_{n,j}}$.
    \end{enumerate}

    Using Lemma \ref{lem:product_block_encoding} we combine $U_{p_{n,j}}$ and $U_{D_{x_{n,j}}}$ as well as $U_{V_{n,n',j}}$ and $U_{D_{p_{n,j}}}$ to construct 
    \begin{itemize}
        \item $U_{(pD_x)_{n,j}}$, an $ \lb \alpha_p \alpha_{D_x}, a_p + a_{D_x}, \alpha_p \epsilon_{D_x} + \alpha_{D_x} \epsilon_p \rb$-block-encoding of $\sum_{\hap_{n,j}=0}^{g_p-1} p_{n,j} \ketbra{\hap_{n,j}}{\hap_{n,j}} \otimes D_{x_{n,j}}$, 
        \item $U_{(VD_p)_{n,n',j}}$, an $\lb \alpha_V \alpha_{D_p}, a_V + a_{D_p}, \alpha_V \epsilon_{D_p} + \alpha_{D_p} \epsilon_V \rb$-block-encoding of \\
        $\sum_{\hax_{n,j}=0}^{g_x-1} \frac{(x_{n,j} - x_{n',j})}{\lb \norm{x_n - x_{n'}}^2 + \Delta^2 \rb^{3/2}} \ketbra{\hax_{n,j}}{\hax_{n,j}} \otimes D_{p_{n,j}}$.
    \end{itemize}

    The next level of the hierarchy involves two different state preparation unitaries:
    \begin{itemize}
        \item $\prep_m$, an $\lb \alpha_m, a_m, \epsilon_m \rb$ state preparation unitary which encodes the nuclear masses,
        \item $\prep_Z$, an $\lb \alpha_Z, a_Z, \epsilon_Z \rb$ state preparation unitary which encodes the atomic numbers of the nuclei. 
    \end{itemize}
   
    Applying Lemma \ref{lem:sum_block_encoding} to $\prep_m$ and $\left\{ U_{(pD_x)_{n,j}} \right\}$ 
    as well as to $\prep_Z$ and $\left\{ U_{(VD_p)_{n,n',j}} \right\}$ yields
    \begin{itemize}
        \item $U_{L_{\text{kin}}^{(NVE)}}$, an $\lb \alpha_{\text{kin}}^{(NVE)}, a_{\text{kin}}^{(NVE)}, \epsilon_{\text{kin}}^{(NVE)} \rb$-block-encoding of \\
        $-i\sum_{n=1}^{N} \sum_{j=1}^{3} \sum_{\hap_{n,j}=0}^{g_p-1} \frac{p_{n,j}}{m_n} \ketbra{\hap_{n,j}}{\hap_{n,j}} \otimes D_{x_{n,j}}$, where \\
        \begin{align}
            \alpha_{\text{kin}}^{(NVE)} &= \alpha_m \alpha_p \alpha_{D_x} \\
            a_{\text{kin}}^{(NVE)} &= a_m + a_p + a_{D_x} \\
            \epsilon_{\text{kin}}^{(NVE)} &= \alpha_m (\alpha_p \epsilon_{D_x} + \alpha_{D_x} \epsilon_p) + \alpha_p \alpha_{D_x} \epsilon_m  \label{epsilon_kin}
        \end{align}
        \item $U_{L_{\text{pot}}^{(NVE)}}$, an $\lb \alpha_{\text{pot}}^{(NVE)}, a_{\text{pot}}^{(NVE)}, \epsilon_{\text{pot}}^{(NVE)} \rb$-block-encoding of \\
        $i\sum_{n=1}^{N} \sum_{n'>n} \sum_{j=1}^{3} \sum_{\hax_{n,j}=0}^{g_x-1} \frac{Z_n Z_{n'}(x_{n,j} - x_{n',j})}{\lb \norm{x_n - x_{n'}}^2 + \Delta^2 \rb^{3/2}} \ketbra{\hax_{n,j}}{\hax_{n,j}} \otimes D_{p_{n,j}}$, where
        \begin{align}
            \alpha_{\text{pot}}^{(NVE)} &= \alpha_Z \alpha_V \alpha_{D_p} \\
            a_{\text{pot}}^{(NVE)} &= a_Z + a_V + a_{D_p} \\
            \epsilon_{\text{pot}}^{(NVE)} &= \alpha_Z (\alpha_V \epsilon_{D_p} + \alpha_{D_p} \epsilon_V) + \alpha_V \alpha_{D_p} \epsilon_Z.  \label{epsilon_pot}
        \end{align}
    \end{itemize}
    Note that \ref{lem:sum_block_encoding} requires all block-encodings of the linear combination to have the same block-encoding normalization constant, which is true in both cases discussed above. More specifically, all $\left\{ U_{(pD_x)_{n,j}} \right\}$ terms have the same normalization constant. Similarly, all $\left\{ U_{(VD_p)_{n,n',j}} \right\}$ terms have the same normalization constant.

    Lastly, we use Lemma \ref{lem:sum_block_encoding} once more to combine $U_{L_{\text{kin}}^{(NVE)}}$ and $U_{L_{\text{pot}}^{(NVE)}}$. Since they have different normalization constants, we first need to renormalize both block-encodings. The idea is as follows: if $U$ is an $\lb \alpha, a, \epsilon \rb$-block-encoding of some matrix $A$, then $U$ is also an $\lb \alpha \beta, a, \beta \epsilon \rb$-block-encoding of the scaled matrix $\beta A$. This follows straight from Definition \ref{def:block-encoding}.
    Thus, $U_{L_{\text{kin}}^{(NVE)}}$ can also be viewed as an
    \begin{equation}
        \lb \alpha_{\text{kin}}^{(NVE)} + \alpha_{\text{pot}}^{(NVE)}, a_{\text{kin}}^{(NVE)} + a_{\text{pot}}^{(NVE)}, \frac{\alpha_{\text{kin}}^{(NVE)} + \alpha_{\text{pot}}^{(NVE)}}{\alpha_{\text{kin}}^{(NVE)}} \epsilon_{\text{kin}}^{(NVE)} + \frac{\alpha_{\text{kin}}^{(NVE)} + \alpha_{\text{pot}}^{(NVE)}}{\alpha_{\text{pot}}^{(NVE)}} \epsilon_{\text{pot}}^{(NVE)} \rb
    \end{equation}
    block-encoding of
    \begin{equation}
        -i \frac{\alpha_{\text{kin}}^{(NVE)} + \alpha_{\text{pot}}^{(NVE)}}{\alpha_{\text{kin}}^{(NVE)}} \sum_{n=1}^{N} \sum_{j=1}^{3} \sum_{\hap_{n,j}=0}^{g_p-1} \frac{p_{n,j}}{m_n} \ketbra{\hap_{n,j}}{\hap_{n,j}} \otimes D_{x_{n,j}}
    \end{equation}
    and $U_{L_{\text{pot}}^{(NVE)}}$ can be viewed as an 
    \begin{equation}
        \lb \alpha_{\text{kin}}^{(NVE)} + \alpha_{\text{pot}}^{(NVE)}, a_{\text{kin}}^{(NVE)} + a_{\text{pot}}^{(NVE)}, \frac{\alpha_{\text{kin}}^{(NVE)} + \alpha_{\text{pot}}^{(NVE)}}{\alpha_{\text{kin}}^{(NVE)}} \epsilon_{\text{kin}}^{(NVE)} + \frac{\alpha_{\text{kin}}^{(NVE)} + \alpha_{\text{pot}}^{(NVE)}}{\alpha_{\text{pot}}^{(NVE)}} \epsilon_{\text{pot}}^{(NVE)} \rb
    \end{equation}
    block-encoding of
    \begin{equation}
        i \frac{\alpha_{\text{kin}}^{(NVE)} + \alpha_{\text{pot}}^{(NVE)}}{\alpha_{\text{kin}}^{(NVE)}} \sum_{n=1}^{N} \sum_{n'>n} \sum_{j=1}^{3} \sum_{\hax_{n,j}=0}^{g_x-1} \frac{Z_n Z_{n'}(x_{n,j} - x_{n',j})}{\lb \norm{x_n - x_{n'}}^2 + \Delta^2 \rb^{3/2}} \ketbra{\hax_{n,j}}{\hax_{n,j}} \otimes D_{p_{n,j}}.
    \end{equation}
    The following state preparation unitary is used to recover the appropriate weighting of the two block-encodings:
    \begin{itemize}
        \item $\prep_{\text{out}}$, an $\lb \alpha_{\text{out}}, a_{\text{out}}, \epsilon_{\text{out}} \rb$ state preparation unitary, which prepares the state \\
        $\sqrt{\frac{\alpha_{\text{kin}}^{(NVE)}}{\alpha_{\text{kin}}^{(NVE)} + \alpha_{\text{pot}}^{(NVE)}}} \ket{0} + \sqrt{\frac{\alpha_{\text{pot}}^{(NVE)}}{\alpha_{\text{kin}}^{(NVE)} + \alpha_{\text{pot}}^{(NVE)}}} \ket{1}$. Note that $\alpha_{\text{out}} = 1$ and $a_{\text{out}} = 1$.
    \end{itemize}
    
    Using $\prep_{\text{out}}$ together with $U_{L_{\text{kin}}^{(NVE)}}$ and $U_{L_{\text{pot}}^{(NVE)}}$ we can construct $U_{L_{\tcla}^{(NVE)}}$, an $\lb \alpha_{NVE}, a_{NVE}, \epsilon_{NVE} \rb$ block-encoding of $L_{\tcla}^{(NVE)}$ where 
    \begin{align}
        \alpha_{NVE} &= \alpha_{\text{kin}}^{(NVE)} + \alpha_{\text{pot}}^{(NVE)}\\
        a_{NVE} &= a_{\text{kin}}^{(NVE)} +  a_{\text{pot}}^{(NVE)} + 1 \\
        \epsilon_{NVE} &= \frac{\alpha_{\text{kin}}^{(NVE)} + \alpha_{\text{pot}}^{(NVE)}}{\alpha_{\text{kin}}^{(NVE)}} \epsilon_{\text{kin}}^{(NVE)} + \frac{\alpha_{\text{kin}}^{(NVE)} + \alpha_{\text{pot}}^{(NVE)}}{\alpha_{\text{pot}}^{(NVE)}} \epsilon_{\text{pot}}^{(NVE)} + \lb \alpha_{\text{kin}}^{(NVE)} + \alpha_{\text{pot}}^{(NVE)} \rb \epsilon_{\text{out}}.
    \end{align}

    We have $\epsilon_{NVE} \leq \epsilon$ if
    \begin{align}
        \epsilon_{\text{kin}}^{(NVE)} &\leq \frac{\alpha_{\text{kin}}^{(NVE)}}{\alpha_{\text{kin}}^{(NVE)} + \alpha_{\text{pot}}^{(NVE)}}  \frac{\epsilon}{3} \\
        \epsilon_{\text{pot}}^{(NVE)} &\leq \frac{\alpha_{\text{pot}}^{(NVE)}}{\alpha_{\text{kin}}^{(NVE)} + \alpha_{\text{pot}}^{(NVE)}}  \frac{\epsilon}{3} \\
        \epsilon_{\text{out}} &\leq \frac{1}{\alpha_{\text{kin}}^{(NVE)} + \alpha_{\text{pot}}^{(NVE)}} \frac{\epsilon}{3}. \label{error_out}
    \end{align}

    It follows from Eqs.~\eqref{epsilon_kin} and \eqref{epsilon_pot} that the above error bounds can be achieved by ensuring that 
    \begin{align}
        \epsilon_{p} &\leq \frac{1}{\alpha_m \alpha_{D_x}} \frac{\alpha_{\text{kin}}^{(NVE)}}{\alpha_{\text{kin}}^{(NVE)} + \alpha_{\text{pot}}^{(NVE)}}  \frac{\epsilon}{9} \label{error_p} \\
        \epsilon_{D_x} &\leq \frac{1}{\alpha_m \alpha_p} \frac{\alpha_{\text{kin}}^{(NVE)}}{\alpha_{\text{kin}}^{(NVE)} + \alpha_{\text{pot}}^{(NVE)}}  \frac{\epsilon}{9} \label{error_D_x} \\
        \epsilon_{m} &\leq \frac{1}{\alpha_p \alpha_{D_x}} \frac{\alpha_{\text{kin}}^{(NVE)}}{\alpha_{\text{kin}}^{(NVE)} + \alpha_{\text{pot}}^{(NVE)}}  \frac{\epsilon}{9} \label{error_m} \\
        \epsilon_{V} &\leq \frac{1}{\alpha_Z \alpha_{D_p}} \frac{\alpha_{\text{pot}}^{(NVE)}}{\alpha_{\text{kin}}^{(NVE)} + \alpha_{\text{pot}}^{(NVE)}}  \frac{\epsilon}{9} \label{error_V} \\
        \epsilon_{D_p} &\leq \frac{1}{\alpha_Z \alpha_V} \frac{\alpha_{\text{pot}}^{(NVE)}}{\alpha_{\text{kin}}^{(NVE)} + \alpha_{\text{pot}}^{(NVE)}}  \frac{\epsilon}{9} \label{error_D_p} \\
        \epsilon_{Z} &\leq \frac{1}{\alpha_V \alpha_{D_p}} \frac{\alpha_{\text{pot}}^{(NVE)}}{\alpha_{\text{kin}}^{(NVE)} + \alpha_{\text{pot}}^{(NVE)}}  \frac{\epsilon}{9} . \label{error_Z}
    \end{align}

    \begin{figure}
        \centering
        \resizebox{1\textwidth}{!}{
        \includegraphics{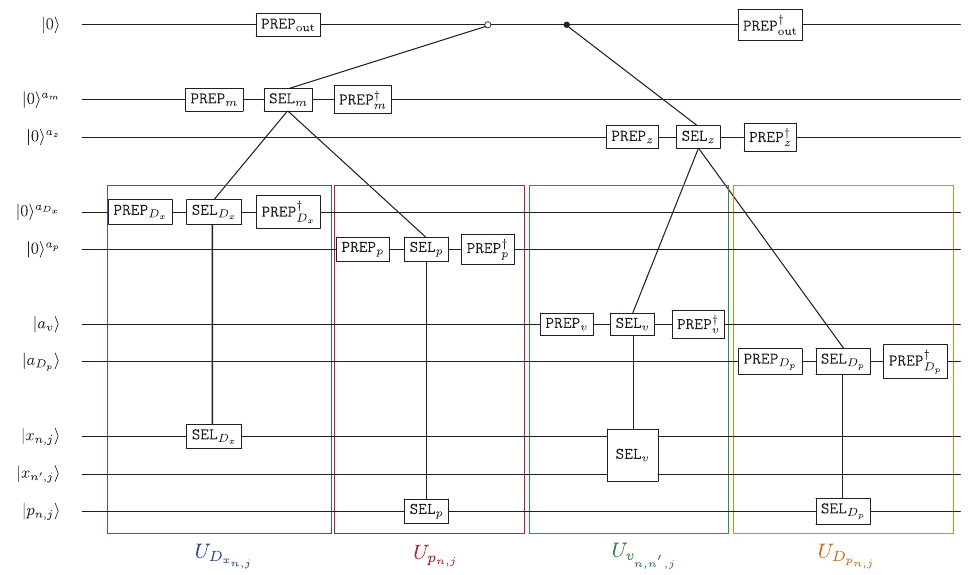} 
        }
        \caption{Circuit diagram of the block-encoding hierarchy used for implementing a block-encoding of $L_{\tcla}^{(NVE)}$.}
    \label{fig:hierarchy}
    \end{figure}

    Let us now show how to implement the four basic block-encodings starting with $U_{p_{n,j}}$. Since this is a block-encoding of a diagonal matrix, we can use the simplest form of the alternating sign trick~\cite{Berry2014alter_sign_trick}. 
    To explain this trick in more detail, let us consider a single computational basis state associated with the momentum variable $p_{n,j}$. We require two additional ancilla registers to implement this trick. The overall input state can then be written as $\ket{0}^{\otimes a_p} \ket{\hap_{n,j}} \ket{0}$ with all remaining quantum registers being suppressed for ease of notation. 
    Let $U_{\leq}$ be the inequality testing unitary from Lemma \ref{lem:inequal}, see Fig.~\ref{fig:inequality_circuit}.
    The circuit in Fig.~\ref{fig:alternating_sign_trick} then evolves the initial state as follows:
    \begin{align*}
        \ket{0}\ket{\hap_{n,j}}\ket{0} &\xrightarrow{\prep_{p}} \frac{1}{\sqrt{2^{a_p}}} \sum_{j=0}^{2^{a_p}-1} \ket{j}\ket{\hap_{n,j}}\ket{0} \\ &\xrightarrow{U_{\leq}} \frac{1}{\sqrt{2^{a_p}}} \lb \sum_{j \leq \hap_{n,j}} \ket{j}\ket{\hap_{n,j}}\ket{0} + \sum_{j > \hap_{n,j}} \ket{j}\ket{\hap_{n,j}}\ket{1} \rb \\
        &\xrightarrow{CZ} \frac{1}{\sqrt{2^{a_p}}} \lb \sum_{j \leq \hap_{n,j}} \ket{j}\ket{\hap_{n,j}}\ket{0} + \sum_{j > \hap_{n,j}} (-1)^j \ket{j}\ket{\hap_{n,j}}\ket{1} \rb \\
        &\xrightarrow{U_{\leq}^\dagger} \frac{1}{\sqrt{2^{a_p}}} \lb \sum_{j \leq \hap_{n,j}} \ket{j} + \sum_{j > \hap_{n,j}} (-1)^j \ket{j} \rb \ket{\hap_{n,j}} \ket{0} \\
         &\xrightarrow{\prep_{p}^\dagger} \frac{1}{2^{a_p}} \sum_k \lb \sum_{j \leq \hap_{n,j}} (-1)^{k \cdot j} + \sum_{j > \hap_{n,j}} (-1)^j (-1)^{k \cdot j} \rb \ket{k} \ket{\hap_{n,j}} \ket{0}
    \end{align*}
    As usual with block-encodings, we postselect on $\ket{k} = \ket{0}$. If $\hap_{n,j}$ is even, this yields the following (unnormalized) state:
    \begin{equation}
        \frac{\hap_{n,j}}{2^{a_p}} \ket{0}\ket{\hap_{n,j}}\ket{0},
    \end{equation}
    as desired. 
    If $\hap_{n,j}$ is odd, then we obtain
    \begin{equation}
        \frac{\hap_{n,j}+1}{2^{a_p}} \ket{0}\ket{\hap_{n,j}}\ket{0}.
    \end{equation}
    The number of $\prep$ ancilla qubits $a_p$ determines the precision $\epsilon_p$ of $U_{p_{n,j}}$. In particular, 
    \begin{equation}
        \norm{U_{p_{n,j}} - \frac{1}{\alpha_p}\sum_{\hap_{n,j}} \hap_{n,j} \ketbra{\hap_{n,j}}{\hap_{n,j}}} \leq \frac{1}{2^{a_p}}.
    \end{equation}
    Note that $a_p$ should be equal to the number of qubits used to represent a single momentum variable (otherwise the inequality test does not work). Increasing $a_p$ therefore requires us to (temporarily) blow up the momentum values as well.
    The above calculation shows that $U_{p_{n,j}} := \prep_{p}^\dagger \cdot \sel_{p} \cdot \prep_{p}$ where
    \begin{align}
        \prep_{p} \ket{0} & := \sum_{j=0}^{2^{a_p}-1} \frac{1}{\sqrt{2^{a_p}}} \ket{j}  \\
        \sel_{p} & := U_{\leq}^\dagger \cdot CZ \cdot U_{\leq},
    \end{align}
    provides an $\lb \alpha_p, a_p, \epsilon_p \rb$-block-encoding of $\sum_{\hap_{n,j}=0}^{g_p-1} p_{n,j} \ketbra{\hap_{n,j}}{\hap_{n,j}}$ where $\alpha_p \in O \lb p_{\max} \rb$ and \\
    $a_p \in O \lb \log{\lb \frac{p_{\max}}{\epsilon_p} \rb} \rb$. The precision $\epsilon_p$ is determined by the overall error tolerance $\epsilon$ as shown in Eq.~\eqref{error_p} which implies that
    \begin{equation}
        a_p \in O \lb \log{\lb \alpha_m \alpha_{D_x} \frac{\alpha_{\text{kin}}^{(NVE)} + \alpha_{\text{pot}}^{(NVE)}}{\alpha_{\text{kin}}^{(NVE)}} \frac{p_{\max}}{\epsilon} \rb} \rb.
    \end{equation}
   Since $\prep_{p}$ only requires Hadamard gates but no Toffoli gates, the Toffoli complexity of $U_{p_{n,j}}$ is equal to the Toffoli complexity of $\sel_{p}$ which is in $O \lb a_p \rb$ due to the inequality testing.

    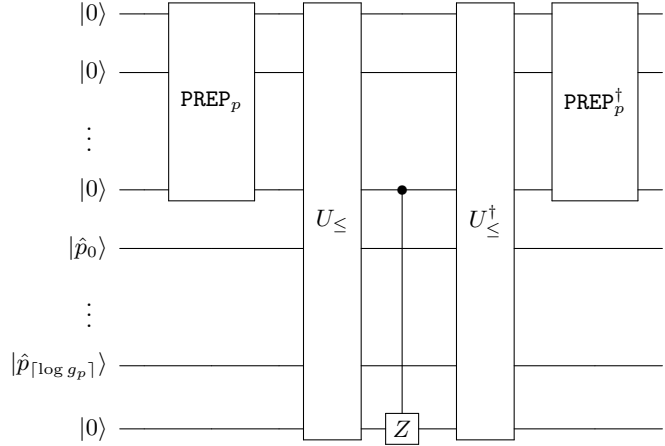
\begin{figure}
    \centering
    \mbox{
        \Qcircuit @C=1em @R=1.5em {
            \lstick{\ket{0}} & \qw & \multigate{3}{\prep_{p}} & \qw & \multigate{7}{U_{\leq}} & \qw & \multigate{7}{U_{\leq}^\dagger} & \multigate{3}{\prep_{p}^\dagger} & \qw \\
            \lstick{\ket{0}} & \qw & \ghost{\prep_{p}} & \qw & \ghost{U_{\leq}} & \qw & \ghost{U_{\leq}^\dagger} & \ghost{\prep_{p}^\dagger} & \qw \\
            \lstick{\vdots \hspace{0.2cm}}&  & \pureghost{\prep_{p}}& & \pureghost{U_{\leq}} &  & \pureghost{U_{\leq}^\dagger} & \pureghost{\prep_{p}^\dagger} \\
            \lstick{\ket{0}} & \qw & \ghost{\prep_{p}} & \qw & \ghost{U_{\leq}} & \ctrl{4} & \ghost{U_{\leq}^\dagger} & \ghost{\prep_{p}^\dagger} & \qw \\
            \lstick{\ket{\hap_{0}}} & \qw & \qw & \qw & \ghost{U_{\leq}} & \qw  & \ghost{U_{\leq}^\dagger} & \qw & \qw \\
            \lstick{\vdots \hspace{0.2cm}}& & & & \pureghost{U_{\leq}} & & \pureghost{(\le ?)^\dagger}  \\
            \lstick{\ket{\hap_{\lceil \log{g_p} \rceil}}}& \qw & \qw & \qw & \ghost{U_{\leq}} & \qw & \ghost{U_{\leq}^\dagger} & \qw & \qw \\
            \lstick{\ket{0}}& \qw & \qw & \qw & \ghost{U_{\leq}} & \gate{Z} & \ghost{U_{\leq}^\dagger} & \qw & \qw \\
    } 
    }
    \caption{Circuit for implementing the alternating sign trick. $\prep_{p}$ prepares a uniform superposition over all computational basis states $\ket{j}$ of the ancilla register where $j \in [2^{a_p}]$. The inequality test $U_{\leq}$ from Fig.~\ref{fig:inequality_circuit} is applied to the ancilla register and the input variable $\ket{p}$ expressed in binary. The result of the inequality test is stored in the bottom qubit (`1' if $j > p$ and `0' otherwise). Next, we apply a $Z$-gate controlled by the least significant qubit of the ancilla register to the output qubit of the inequality test. As long as $j \leq p$, the controlled $Z$-gate acts as the identity gate. However, when $j > p$, the controlled $Z$-gate introduces a minus sign for every second computational basis state of the ancilla register. This creates an alternating sequence of $\pm 1$ such that the contributions of all $j > p$ cancel each other out. Finally, we uncompute the ancilla qubits. 
    }
    \label{fig:alternating_sign_trick}
    \end{figure}

    The implementation of $U_{V_{n,n',j}}$, an $\lb \alpha_{V}, a_{V}, \epsilon_{V} \rb$-block-encoding of
    \begin{equation}
        \sum_{\hax_{n,j}=0}^{g_x-1} \frac{(x_{n,j} - x_{n',j})}{\lb \norm{x_n - x_{n'}}^2 + \Delta^2 \rb^{3/2}} \ketbra{\hax_{n,j}}{\hax_{n,j}}
    \end{equation}
    is also based on the alternating sign trick.
    Using $\prep_V \ket{0} := \sum_{l=0}^{2^{a_V}-1} \frac{1}{\sqrt{2^{a_V}}} \ket{l}$, we test the following inequality:
    \begin{equation}
        l^2 \lb \norm{\hax_n - \hax_{n'}}^2 + \overline{\Delta}^2 \rb^3 \le (\hax_{n,j} - \hax_{n',j})^2,
    \label{coulomb_ineq}
    \end{equation}
    where $\overline{\Delta} \in \mathbb{N}$ such that $\Delta = \overline{\Delta} h_x$.
    Note that $\alpha_V \in O \lb \frac{x_{\text{max}}}{\Delta^3} \rb$. The number of ancilla qubits $a_V$ is again determined by the allowable error. More specifically, $a_V \in O \lb \log{\lb \frac{x_{\text{max}}}{\Delta^3 \epsilon_V}\rb}\rb$.
    The precision $\epsilon_V$ is determined by the overall error tolerance $\epsilon$ as shown in Eq.~\eqref{error_V} which implies that
    \begin{equation}
        a_V \in O \lb \log{\lb \alpha_Z \alpha_{D_p} \frac{\alpha_{\text{kin}}^{(NVE)} + \alpha_{\text{pot}}^{(NVE)}}{\alpha_{\text{pot}}^{(NVE)}} \frac{x_{\text{max}}}{\Delta^3 \epsilon} \rb} \rb.\label{num_qubits_av}
    \end{equation}
    
    To determine the correct sign we also need to test $x_{n,j} \le x_{n',j}$ which has Toffoli complexity in $O \lb \log{(g_x)} \rb$.
    The advantage of testing Eq.~({\ref{coulomb_ineq}}) rather than directly
    \begin{equation}
        l \le \frac{(\hax_{n,j} - \hax_{n',j})}{\lb \norm{\hax_n - \hax_{n'}}^2 + \overline{\Delta}^2 \rb^{3/2}}
    \end{equation}
    is that we do not have to calculate fractions containing square roots.
    However, the inequality test in Eq.~({\ref{coulomb_ineq}}) does require us first to compute the left- and right-hand sides of the inequality using $O(1)$ quantum Karatsuba multiplications. This can be done using $O\lb \lb a_V \rb^{\log 3} \rb$ Toffoli gates, whereas the inequality test itself requires only $O\lb a_V \rb$ Toffolis~\cite{Gidney2019}. 

    Next, let us explain how to construct a block-encoding $U_{D_{x_{n,j}}}$ of the discrete derivative operator $D_{x_{n,j}}$ of order $2 d_x$. The idea is to apply a linear combination of $2d_x$ unitary adders to the $\ket{\hax_{n,j}}$ register as shown in Fig.~\ref{fig:finite_difference}.
    \begin{figure}
        \centering
        \mbox{
            \Qcircuit @C=1em @R=1.5em {
                \lstick{\ket{0}} & \qw & \multigate{4}{\prep_d} & \qw & \ctrlo{1} & \ctrl{1} & \qw & \dots &  &  \ctrl{1}  & \multigate{4}{\prep_d^\dagger} & \qw \\
                \lstick{\ket{0}} & \qw & \ghost{\prep_d} & \qw & \ctrlo{1} & \ctrlo{1} & \qw & \dots &  &  \ctrl{1}  & \ghost{\prep_d^\dagger} & \qw \\
                & & \pureghost{\prep_d} &&&&&&&& \pureghost{\prep_d^\dagger}\\
                \lstick{\vdots \hspace{0.2cm}}& &  \pureghost{\prep_d} & & \vdots & \vdots &  & \ddots & & \vdots & \pureghost{\prep_d^\dagger} & \\
                \lstick{\ket{0}} &  \qw & \ghost{\prep_d} & \qw & \ctrlo{1} & \ctrlo{1} & \qw & \dots &  &  \ctrl{1}  & \ghost{\prep_d^\dagger} & \qw \\
                \lstick{\ket{\hax_{n,j,0}}} & \qw & \qw & \qw & \multigate{3}{\mathtt{Add} - d_x} & \multigate{3}{\mathtt{Add} - d_x + 1} & \qw & \dots & & \multigate{3}{\mathtt{Add} + d_x} \qw & \qw & \qw \\
                \lstick{\ket{\hax_{n,j,1}}}  & \qw & \qw & \qw  & \ghost{\mathtt{Add} - d_x} & \ghost{\mathtt{Add} - d_x + 1} & \qw & \dots & & \ghost{\mathtt{Add} + d_x} & \qw & \qw \\
                \lstick{\vdots \hspace{0.2cm}}& & & & \pureghost{\mathtt{Add} - d_x} & \pureghost{\mathtt{Add}- d_x + 1} & & \ddots & & \pureghost{\mathtt{Add} + d_x} &  \\
                \lstick{\ket{\hax_{n,j,\lceil \log{g_x} \rceil}}}& \qw & \qw & \qw & \ghost{\mathtt{Add} - d_x} & \ghost{\mathtt{Add} - d_x + 1} & \qw & \dots &  & \ghost{\mathtt{Add} + d_x} & \qw & \qw 
        } 
        }
        \caption{Circuit for implementing a central finite difference operator of order $2d_x$. $\mathtt{Add} + j$ with $j$ being an integer is a unitary adder of the form $\sum_{\hax} \ketbra{\hax - j}{\hax}$.}
        \label{fig:finite_difference}
    \end{figure}
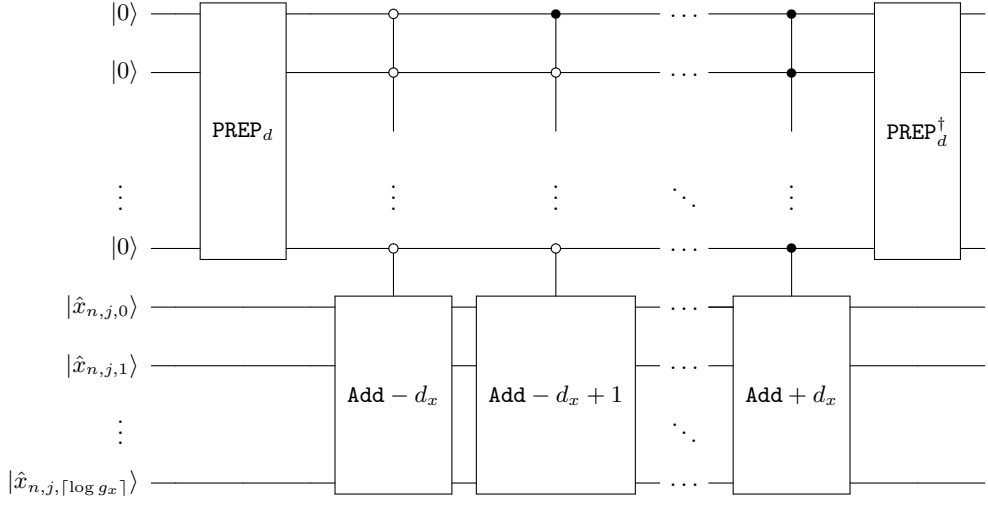

    Let
    \begin{align}
        \prep_{D_x} \ket{0} & := \sum_{k= -d_x}^{d_x} \sqrt{\frac{c_{d_x, k}}{c_{d_x}}} \ket{k}  \\
        \sel_{D_x} & := \sum_{k=-d_x}^{d_x} \ketbra{k}{k} \otimes \sum_{\hax} \ketbra{\hax-k}{\hax},
    \end{align}
    where $\{c_{d_x, k}\}$ are the finite difference coefficients as given in Lemma \ref{lem:finite_diff} and $c_{d_x}:= \sum_{k= -d_x}^{d_x} |c_{d_x, k}|$. Then $U_{D_{x_{n,j}}} := \prep_{D_x}^{\dagger} \sel_{D_x} \prep_{D_x}$ is an $\lb \alpha_{D_x}, a_{D_x}, \epsilon_{D_x} \rb$-block-encoding of $D_{x_{n,j}}$ where
    \begin{equation}
        \alpha_{D_x} = \frac{c_{d_x}}{h_x} \leq \frac{2 \lb \ln{d_x} + 1 \rb}{h_x},
    \end{equation}
    see Lemma~\ref{lem:bound_FD_coeff}, and $a_{D_x} \in O \lb \log{\lb d_x \rb} \rb$. The error $\epsilon_{D_x}$ stems solely from the state preparation error associated with $\prep_{D_x}$. Lemma \ref{lem:prep_error} implies that we need to prepare the state $\prep_{D_x} \ket{0}$ within error
    \begin{equation}
        \frac{\epsilon_{D_x}}{\alpha_{D_x} \sqrt{d_x}} \in O \lb \frac{h_x \epsilon_{D_x}}{\ln{(d_x)} \sqrt{d_x} } \rb.
    \end{equation}
    Such a general quantum state preparation has Toffoli cost in $O \lb d_x \log{\lb \frac{\ln{(d_x)} \sqrt{d_x} }{h_x \epsilon_{D_x}} \rb} \rb$. By Eq.~(\ref{error_D_x}) this is in 
    \begin{equation}
        O \lb d_x \log{\lb \alpha_m \alpha_p \frac{\alpha_{\text{kin}}^{(NVE)} + \alpha_{\text{pot}}^{(NVE)}}{\alpha_{\text{kin}}^{(NVE)}} \frac{\ln{(d_x)} \sqrt{d_x} }{h_x} \frac{1}{\epsilon} \rb}\rb.
    \end{equation}
    A single unitary adder requires $O \lb \log{(g_x)} \rb$ Toffolis where $g_x$ again is the number of grid points for a single position coordinate. Additionally, we need $O \lb \log{d_x} \rb$ Toffolis to implement a controlled version of the adder controlled by the $\prep_d$ register.
    
    In total, we therefore require 
    \begin{equation}
        O \lb d_x \lb \log{d_x} + \log{g_x} + \log{\lb \alpha_m \alpha_p \frac{\alpha_{\text{kin}}^{(NVE)} + \alpha_{\text{pot}}^{(NVE)}}{\alpha_{\text{kin}}^{(NVE)}} \frac{\ln{(d_x)} \sqrt{d_x} }{h_x} \frac{1}{\epsilon} \rb} \rb \rb
    \end{equation}
    Toffolis to implement $U_{D_{x_{n,j}}}$.

    We implement $U_{D_{p_{n,j}}}$, an $\lb \alpha_{D_p}, a_{D_p}, \epsilon_{D_p} \rb$-block-encoding of the discrete derivative operator $D_{p_{n,j}}$ of order $2 d_p$, using exactly the same strategy as for $U_{D_{x_{n,j}}}$. This implies that 
    \begin{equation}
        \alpha_{D_p} \leq \frac{2 \lb \ln{d_p} + 1 \rb}{h_p}
    \end{equation}
    and $a_{D_p} \in O \lb \log{\lb d_p \rb} \rb$. As before, the error $\epsilon_{D_p}$ stems solely from the state preparation error of $\prep_{D_p}$. The Toffoli cost associated with preparing the state $\prep_{D_p} \ket{0}$ within sufficiently small error is then in 
    \begin{equation}
        O \lb d_p \log{\lb \alpha_Z \alpha_V \frac{\alpha_{\text{kin}}^{(NVE)} + \alpha_{\text{pot}}^{(NVE)}}{\alpha_{\text{kin}}^{(NVE)}} \frac{\ln{(d_p)} \sqrt{d_p} }{h_p} \frac{1}{\epsilon} \rb}\rb.
    \end{equation}
    In total, we then require 
    \begin{equation}
        O \lb d_p \lb \log{d_p} + \log{g_p} + \log{\lb \alpha_Z \alpha_V \frac{\alpha_{\text{kin}}^{(NVE)} + \alpha_{\text{pot}}^{(NVE)}}{\alpha_{\text{pot}}^{(NVE)}} \frac{\ln{(d_p)} \sqrt{d_p} }{h_p} \frac{1}{\epsilon} \rb} \rb \rb
    \end{equation}
    Toffolis to implement $U_{D_{p_{n,j}}}$.

    Now that we have shown how to implement the four basic block-encodings, $U_{p_{n,j}}$, $U_{D_{x_{n,j}}}$, $U_{V_{n,n',j}}$ and $U_{D_{p_{n,j}}}$, we can combine them. More specifically, we multiply $U_{p_{n,j}}$ and $U_{D_{x_{n,j}}}$ to obtain $U_{(pD_x)_{n,j}}$. This can be done at no extra Toffoli cost by simply keeping the ancilla qubits separate and applying the two block-encodings consecutively. The same is true when multiplying $U_{V_{n,n',j}}$ and $U_{D_{p_{n,j}}}$ to obtain $U_{(VD_p)_{n,n',j}}$. 

    Next, let us explain how to implement $\prep_m$, an $\lb \alpha_m, a_m, \epsilon_m \rb$ state preparation unitary which we define as follows:
    \begin{equation}
        \prep_m \ket{0} := \sum_{n=1}^{N} \sqrt{\frac{1/m_n}{\alpha_m}} \ket{n} \otimes \frac{1}{\sqrt{3}} \sum_{j=1}^3 \ket{j}, 
    \end{equation}
    where 
    \begin{equation}
        \alpha_m = \sum_{n,j} \frac{1}{m_n} \leq \frac{3N}{m_{\min}}.
    \end{equation}
    The above definition implies that $a_m = \lceil \log{N} \rceil + \lceil \log{3} \rceil$. It follows from Lemma \ref{lem:prep_error} that we need to prepare $\prep_m \ket{0}$ within error
    \begin{equation}
        \frac{\epsilon_{m}}{\alpha_{m} \sqrt{N}} \in O \lb \frac{m_{\min} \epsilon_{m}}{N \sqrt{N}} \rb.
    \end{equation} 
    Such a general quantum state preparation has Toffoli cost in~\cite{NC2010}
    \begin{equation}
        O \lb N \log{\lb \frac{N}{m_{\min} \epsilon_{m}} \rb} \rb.
    \end{equation}
    By Eq.~(\ref{error_m}) this is in
    \begin{equation}
        O \lb N \log{\lb \alpha_p \alpha_{D_x} \frac{\alpha_{\text{kin}}^{(NVE)} + \alpha_{\text{pot}}^{(NVE)}}{\alpha_{\text{kin}}^{(NVE)}} \frac{N}{m_{\min}} \frac{1}{\epsilon} \rb}\rb.
    \end{equation}
    
    We use $\prep_m$ together with $U_{(pD_x)_{n,j}}$ to implement $U_{L_{\text{kin}}^{(NVE)}}$. This can be done efficiently with the help of two additional ancilla registers which we call ``$\swap$ registers''~\cite{Su2021first_quant_sim}. Controlled by the $\prep_m$ register we swap the appropriate position and momentum variables into the two $\swap$ registers. This allows us to apply the block-encodings $U_{p_{n,j}}$ and $U_{D_{x_{n,j}}}$ only once (to the $\swap$ registers holding the appropriate position and momentum variables) rather than $3N$ times (to each individual position and momentum variable). However, we do require a total of $O \lb N \log{\lb g \rb} \rb$ $\swap$ operations, where $g = \max \{ g_x, g_p \}$, implying  $O \lb N \log{\lb g \rb} \rb$ Toffolis.

    $U_{L_{\text{pot}}^{(NVE)}}$ can be implemented following the same strategy. We have
    \begin{equation}
        \prep_Z \ket{0} := \sum_{n=1}^{N} \sqrt{\frac{Z_n}{Z}} \ket{n} \otimes \sum_{n'=1}^{N} \sqrt{\frac{Z_n'}{Z}} \ket{n'} \otimes \frac{1}{\sqrt{3}} \sum_{j=1}^3 \ket{j},
    \end{equation}
    an $\lb \alpha_Z, a_Z, \epsilon_Z \rb$ state preparation unitary with 
    \begin{equation}
        \alpha_Z = \sum_{n,n',j} Z_n Z_{n'} \leq 3 N^2 Z_{\max}^2.
    \end{equation}
    The above definition implies that $a_Z = 2\lceil\log{N}\rceil + \lceil \log{3} \rceil$. Importantly, it is a product state, meaning that $\sum_{n=1}^{N} \sqrt{\frac{Z_n}{Z}} \ket{n}$, $\sum_{n'=1}^{N} \sqrt{\frac{Z_n'}{Z}} \ket{n'}$ and $\sum_{j=1}^3 \ket{j}$ can be prepared individually.
    It follows from Lemma \ref{lem:prep_error} that we need to prepare $\prep_Z \ket{0}$ within error
    \begin{equation}
        \frac{\epsilon_{Z}}{\alpha_{Z} N} \in O \lb \frac{\epsilon_{Z}}{N^3 Z_{\max}^2} \rb.
    \end{equation}
    Preparing such a product state has Toffoli cost in $O \lb N \log{\lb \frac{N Z_{\max}}{\epsilon_{Z}} \rb} \rb$. By Eq.~(\ref{error_Z}) this is in
    \begin{equation}
        O \lb N \log{\lb \alpha_V \alpha_{D_p} \frac{\alpha_{\text{kin}}^{(NVE)} + \alpha_{\text{pot}}^{(NVE)}}{\alpha_{\text{pot}}^{(NVE)}} \frac{N Z_{\max}}{\epsilon} \rb}\rb.
    \end{equation}

    We use $\prep_Z$ together with $U_{(VD_p)_{n,j}}$ to implement $U_{L_{\text{pot}}^{(NVE)}}$. This can be done efficiently with the help of 7 $\swap$ registers, 6 for the 6 nuclear position variables appearing in $\frac{(x_{n,j} - x_{n',j})}{\lb \norm{x_n - x_{n'}}^2 + \Delta^2 \rb^{3/2}}$ and one for the nuclear momentum variable of $D_{p_{n,j}}$. As before, this allows us to apply the block-encodings $U_{V_{n,n',j}}$ and $U_{D_{p_{n,j}}}$ only once (to the $\swap$ registers holding the appropriate position and momentum variables) rather than $3N^2$ times and $3N$ times, respectively. However, we again require a total of $O \lb N \log{\lb g \rb} \rb$ $\swap$ operations, implying  $O \lb N \log{\lb g \rb} \rb$ Toffolis.
    To ensure that we exclude terms where the nuclei are the same, i.e. $n = n'$, and also avoid double counting, we perform an inequality test on $\ket{n}$ and $\ket{n'}$ and store the result in an ancilla qubit. The corresponding Toffoli cost is in $O\lb \log{N} \rb$.

    Lastly, we use the following $\lb \alpha_{\text{out}}, a_{\text{out}}, \epsilon_{\text{out}} \rb$ state preparation unitary
    \begin{equation}
        \prep_{\text{out}}\ket{0} := \sqrt{\frac{\alpha_{\text{kin}}^{(NVE)}}{\alpha_{\text{kin}}^{(NVE)} + \alpha_{\text{pot}}^{(NVE)}}} \ket{0} + \sqrt{\frac{\alpha_{\text{pot}}^{(NVE)}}{\alpha_{\text{kin}}^{(NVE)} + \alpha_{\text{pot}}^{(NVE)}}} \ket{1}
    \end{equation}
    together with $U_{L_{\text{kin}}^{(NVE)}}$ and $U_{L_{\text{pot}}^{(NVE)}}$ to construct $U_{L_{\tcla}^{(NVE)}}$, an $\lb \alpha_{NVE}, a_{NVE}, \epsilon_{NVE} \rb$ block-encoding of $L_{\tcla}^{(NVE)}$. As mentioned before, $\alpha_{\text{out}} = 1$ and $a_{\text{out}} = 1$. It follows from Lemma \ref{lem:prep_error} that we need to prepare $\prep_{\text{out}} \ket{0}$ within error $\epsilon_{\text{out}}$. 
    Such a general quantum state preparation on one qubit has Toffoli cost in $O \lb \log{\lb \frac{1}{\epsilon_{\text{out}}} \rb} \rb$. By Eq.~(\ref{error_out}) this is in
    \begin{equation}
        O \lb \log{\lb \frac{\alpha_{\text{kin}}^{(NVE)} + \alpha_{\text{pot}}^{(NVE)}}{\epsilon} \rb} \rb.
    \end{equation}
    
    Combining all previous results, we find that
    \begin{equation}
         \alpha_{NVE} \in O \lb  N\frac{p_{\text{max}}}{m_{\text{min}}} \frac{\ln{d_x}}{h_x} + N^2\frac{ Z_{\text{max}}^2 x_{\text{max}}}{\Delta^3} \frac{\ln{d_p}}{h_p}  \rb.
    \end{equation}
    Furthermore, ensuring $\epsilon_{NVE} \leq \epsilon$ requires
    \begin{equation}
    \begin{split}
         a_{NVE} &= a_p + a_V + a_{D_x} + a_{D_p} + a_m + a_Z + 1 \\
         &\in O \lb \log{\lb \frac{\alpha_{NVE}}{\epsilon} \rb} + \log{\lb d_x \rb} + \log{\lb d_p \rb} + \log{(N)} \rb \\
         &\in O \lb \log{\lb N\frac{p_{\text{max}}}{m_{\text{min}}} \frac{\ln{d_x}}{h_x} + N^2\frac{ Z_{\text{max}}^2 x_{\text{max}}}{\Delta^3} \frac{\ln{d_p}}{h_p} \rb} + \log{\lb \frac{d}{\epsilon} \rb} \rb
    \end{split}
    \end{equation}
    block-encoding ancilla qubits, where $d := \max \{ d_x , d_p \}$, and
    \begin{equation}
        \widetilde{O} \lb N \log{\lb \frac{g \alpha_{NVE}}{\epsilon} \rb} + \log^{\log{3}}{\lb \frac{\alpha_{NVE}}{\epsilon} \rb}  + d \log{(g)} \rb
    \end{equation}
    Toffoli gates.
\end{proof}

\subsection{Proof of Lemma \ref{lem:bounds_L_class_NVT}}
\label{sec:proof_lem_NVT}

For convenience, let us restate Lemma \ref{lem:bounds_L_class_NVT} here.

\LclassNVT*

\begin{proof}
    Lemma \ref{lem:bounds_L_class_NVT} can be proved analogously to Lemma \ref{lem:bounds_L_class_NVE} via a modified block-encoding hierarchy. Here we only give a brief summary of the construction.
    At the lowest level, we now have nine types of block-encodings which we can express as a function of $\alpha_{NVT}$. Note that the resulting upper bounds on the individual block-encodings are somewhat looser than the corresponding bounds used in the proof of Lemma \ref{lem:bounds_L_class_NVE}. However, this does not affect the overall Toffoli or ancilla complexity. With this in mind, we list the lowest level block-encodings below.
    \begin{enumerate}
        \item $U_{p'_{n,j}}$, an $\lb \alpha_{p'}, a_{p'}, \epsilon_{p'} \rb$ block-encoding of $\sum_{\hap_{n,j}=0}^{g_{p'}-1} \sum_{\has=0}^{g_s-1} \frac{p'_{n,j}}{(s+s_{\text{min}})^2} \ketbra{\hap'_{n,j}}{\hap'_{n,j}} \otimes \ketbra{\has}{\has}$ where 
        \begin{align}
            \alpha_{p'} &\in O \lb \frac{p'_{\text{max}}}{m_{\text{min}} s_{\text{min}}^2} \rb \\
            a_{p'} &\in O \lb \log{\lb \frac{\alpha_{NVT}}{\epsilon} \rb} \rb \\
            \epsilon_{p'} &\in O \lb \frac{\epsilon}{\alpha_{NVT}} \rb.
        \end{align}
        It can be efficiently implemented using the alternating sign trick together with quantum Karatsuba multiplication. The resulting Toffoli cost is in $O \lb \lb a_{p'} \rb^{\log{3}} \rb$, which is dominated by the cost of implementing quantum Karatsuba multiplication with $a_{p'}$ qubits~\cite{Gidney2019}.
        
        \item $U_{D_{x_{n,j}}}$, an $\lb \alpha_{D_x}, a_{D_x}, \epsilon_{D_x} \rb$-block-encoding of $D_{x_{n,j}}$ where
        \begin{align}
            \alpha_{D_x} &\in O \lb \frac{\ln{d_x}}{h_x} \rb \\
            a_{D_x} &\in O \lb \log d_x \rb \\
            \epsilon_{D_x} & \in O \lb \frac{\epsilon}{\alpha_{NVT}} \rb.
        \end{align}
        It can be efficiently implemented via a linear combination of unitary adders.
        The associated Toffoli cost is in
        \begin{equation}
            \widetilde{O} \lb d_x \lb \log{d_x} + \log{g_x} + \log{\lb \frac{\alpha_{NVT}}{\epsilon} \rb} \rb \rb,
        \end{equation}
        which includes the cost of implementing a controlled unitary adder on $\log{g_x}$ qubits and the cost of preparing a state encoding the $2d_x$ coefficients of the central finite difference formula of order $2d_x$.
        
        \item $U_{V_{n,n',j}}$, an $\lb \alpha_V, a_V, \epsilon_V \rb$-block-encoding of $\sum_{\hax_{n,j}=0}^{g_x-1} \frac{(x_{n,j} - x_{n',j})}{\lb \norm{x_n - x_{n'}}^2 + \Delta^2 \rb^{3/2}} \ketbra{\hax_{n,j}}{\hax_{n,j}}$ where
        \begin{align}
            \alpha_{V} &\in O \lb \frac{x_{\text{max}}}{\Delta^3} \rb \\
            a_{V} &\in O \lb \log{\lb \frac{\alpha_{NVT}}{\epsilon} \rb} \rb \\
            \epsilon_{V} &\in O \lb \frac{\epsilon}{\alpha_{NVT}} \rb.
        \end{align}
        It can be efficiently implemented using the alternating sign trick together with quantum Karatsuba multiplication.
        The associated Toffoli cost is in $O \lb \lb a_{V} \rb^{\log{3}} \rb$, which is dominated by the cost of implementing quantum Karatsuba multiplication with $a_{V}$ qubits.
        
        \item $U_{D_{p'_{n,j}}}$, an $\lb \alpha_{D_{p'}}, a_{D_{p'}}, \epsilon_{D_{p'}} \rb$-block-encoding of $D_{p'_{n,j}}$ where
        \begin{align}
            \alpha_{D_{p'}} &\in O \lb \frac{\ln{d_{p'}}}{h_{p'}} \rb \\
            a_{D_{p'}} &\in O \lb \log d_{p'} \rb \\
            \epsilon_{D_{p'}} &\in O \lb \frac{\epsilon}{\alpha_{NVT}} \rb.
        \end{align}
        It can be efficiently implemented via a linear combination of unitary adders.
        The associated Toffoli cost is in
        \begin{equation}
             \widetilde{O} \lb d_{p'} \lb \log{d_{p'}} + \log{g_{p'}} + \log{\lb \frac{\alpha_{NVT}}{\epsilon} \rb} \rb \rb,
        \end{equation}
        which includes the cost of implementing a controlled unitary adder on $\log{g_{p'}}$ qubits and the cost of preparing a state encoding the $2d_{p'}$ coefficients of the central finite difference formula of order $2d_{p'}$.
        
        \item $U_{p_s}$, an $\lb \alpha_{p_s}, a_{p_s}, \epsilon_{p_s} \rb$-block-encoding of $\sum_{\hap_{s}} p_s \ketbra{{\hap_{s}}}{\hap_{s}}$ where
        \begin{align}
            \alpha_{p_s} &\in O \lb p_{s, \text{max}} \rb \\
            a_{p_s} &\in O \lb \log{\lb \frac{\alpha_{NVT}}{\epsilon} \rb} \rb \\
            \epsilon_{p_s} &\in O \lb \frac{\epsilon}{\alpha_{NVT}} \rb.
        \end{align}
        It can be efficiently implemented using the alternating sign trick.
        The associated Toffoli cost is in $O \lb a_{p_s} \rb$.
        \item $U_{D_{s}}$, an $\lb \alpha_{D_{s}}, a_{D_{s}}, \epsilon_{D_{s}} \rb$-block-encoding of $D_{s}$ where
        \begin{align}
            \alpha_{D_{s}} &\in O \lb \frac{\ln{d_{s}}}{h_{s}} \rb \\
            a_{D_{s}} &\in O \lb \log d_s \rb \\
            \epsilon_{D_s} &\in O \lb \frac{\epsilon}{\alpha_{NVT}} \rb.
        \end{align}
        It can be efficiently implemented via a linear combination of unitary adders.
        The associated Toffoli cost is in
        \begin{equation}
             \widetilde{O} \lb d_{s} \lb \log{d_{s}} + \log{g_{s}} + \log{\lb \frac{\alpha_{NVT}}{\epsilon} \rb} \rb \rb,
        \end{equation}
        which includes the cost of implementing a controlled unitary adder on $\log{g_{s}}$ qubits and the cost of preparing a state encoding the $2d_{s}$ coefficients of the central finite difference formula of order $2d_{s}$.
        
        \item $U_{s,n,j}$, an $\lb \alpha_{s}, a_{s}, \epsilon_{s} \rb$ block-encoding of $-\sum_{\hap_{n,j}=0}^{g_p-1} \sum_{\has=0}^{g_s-1} \frac{2p_{n,j}^2}{(s+s_{\text{min}})^3} \ketbra{\hap_{n,j}}{\hap_{n,j}} \otimes \ketbra{\has}{\has}$ where
        \begin{align}
            \alpha_{s} &\in O \lb \frac{p_{\text{max}}^2}{m_{\text{min}} s_{\text{min}}^3} \rb \\
            a_{s} &\in O \lb \log{\lb \frac{\alpha_{NVT}}{\epsilon} \rb} \rb \\
            \epsilon_{s} &\in O \lb \frac{\epsilon}{\alpha_{NVT}} \rb.
        \end{align}
        It can be efficiently implemented using the alternating sign trick together with quantum Karatsuba multiplication.
        The associated Toffoli cost is in $O \lb \lb a_{s} \rb^{\log{3}} \rb$ which is dominated by the cost of implementing quantum Karatsuba multiplication with $a_{s}$ qubits.
        
        \item $U_{1/s}$, an $\lb \alpha_{1/s}, a_{1/s}, \epsilon_{1/s} \rb$-block-encoding of $\sum_{\has=0}^{g_s-1} \frac{1}{s + s_{\text{min}}} \ketbra{{\has}}{\has}$ where
        \begin{align}
            \alpha_{1/s} &\in O \lb \frac{1}{s_{\text{min}}} \rb \\
            a_{1/s} &\in O \lb \log{\lb \frac{\alpha_{NVT}}{\epsilon} \rb} \rb \\
            \epsilon_{1/s} &\in O \lb \frac{\epsilon}{\alpha_{NVT}} \rb.
        \end{align}
        It can be efficiently implemented using the alternating sign trick together with quantum Karatsuba multiplication.
        The associated Toffoli cost is in $O \lb \lb a_{1/s} \rb^{\log{3}} \rb$ which is dominated by the cost of implementing quantum Karatsuba multiplication with $a_{1/s}$ qubits.
        
        \item $U_{D_{p_s}}$, an $\lb \alpha_{D_{p_s}}, a_{D_{p_s}}, \epsilon_{D_{p_s}} \rb$-block-encoding of $D_{p_{s}}$ where
        \begin{align}
            \alpha_{D_{p_s}} &\in O \lb \frac{\ln{d_{p_s}}}{h_{p_s}} \rb \\
            a_{D_{p_s}} &\in O \lb \log d_{p_s} \rb \\
            \epsilon_{D_{p_s}} &\in O \lb \frac{\epsilon}{\alpha_{NVT}} \rb.
        \end{align}
        It can be efficiently implemented via a linear combination of unitary adders.
        The associated Toffoli cost is in
        \begin{equation}
             \widetilde{O} \lb d_{p_s} \lb \log{d_{p_s}} + \log{g_{p_s}} + \log{\lb \frac{\alpha_{NVT}}{\epsilon} \rb} \rb \rb,
        \end{equation}
        which includes the cost of implementing a controlled unitary adder on $\log{g_{p_s}}$ qubits and the cost of preparing a state encoding the $2d_{p_s}$ coefficients of the central finite difference formula of order $2d_{p_s}$.
    \end{enumerate}
    We then use Lemmas \ref{lem:product_block_encoding} and \ref{lem:sum_block_encoding} to combine the above block-encodings. This requires the following state preparation unitaries:
    \begin{itemize}
        \item $\prep_m$, an $\lb \alpha_m, a_m, \epsilon_m \rb$ state preparation unitary which encodes the nuclear masses, where
        \begin{align}
            \alpha_m &\in O \lb \frac{N}{m_{\text{min}}} \rb \\
            a_m &\in O \lb \log{N} \rb \\
            \epsilon_{m} &\in O \lb \frac{\epsilon}{\alpha_{NVT}} \rb.
        \end{align}
        The Toffoli cost of this state preparation unitary is in
        \begin{equation}
            O \lb N \log{\lb \frac{\alpha_{NVT}}{\epsilon} \rb}\rb.
        \end{equation}
        \item $\prep_Z$, an $\lb \alpha_Z, a_Z, \epsilon_Z \rb$ state preparation unitary which encodes the atomic numbers of the nuclei, where
        \begin{align}
            \alpha_Z &\in O \lb N^2 Z_{\text{max}}^2 \rb \\
            a_Z &\in O \lb \log{N} \rb \\
            \epsilon_{Z} &\in O \lb \frac{\epsilon}{\alpha_{NVT}} \rb.
        \end{align}
        The Toffoli cost of this state preparation unitary is in
        \begin{equation}
            O \lb N \log{\lb \frac{\alpha_{NVT}}{\epsilon} \rb}\rb.
        \end{equation}
        \item $\prep_{\text{out}}^{NVT}$, an $\lb \alpha_{\text{out}}, a_{\text{out}}, \epsilon_{\text{out}} \rb$ state preparation unitary which is used to combine all terms of the $NVT$ Liouvillian, where
        \begin{align}
            \alpha_{\text{out}} &\in O \lb 1 \rb \\
            a_{\text{out}} & \in O \lb 1 \rb \\
            \epsilon_{\text{out}} &\in O \lb \frac{\epsilon}{\alpha_{NVT}} \rb.
        \end{align}
        The Toffoli cost of this state preparation unitary is in
        \begin{equation}
            O \lb \log{\lb \frac{\alpha_{NVT}}{\epsilon} \rb} \rb.
        \end{equation}
    \end{itemize}
    
    As before, we utilize $O(1)$ $\swap$ registers to combine the individual block-encodings efficiently. Controlled by the $\prep_m$ or $\prep_Z$ register we swap the appropriate nuclear position and momentum variables into the $\swap$ registers. This allows us to apply the 9 basic block-encodings only once to the $\swap$ registers holding the appropriate position and momentum variables rather than $O \lb N \rb$ or $O \lb N^2 \rb$ times to each individual position or momentum variable. However, we do require a total of $O \lb N \log{\lb g' \rb} \rb$ $\swap$ operations, where $g' = \max \{ g_x, g_{p'} \} \leq g = \max \{ g_x, g_{p'}, g_s, g_{p_s} \}$, implying  $O \lb N \log g \rb$ Toffolis.
    
    Going through the same analysis as for the classical $NVE$ Liouvillian yields the desired complexity bounds. 
\end{proof}

\section{Evolution under the electronic Liouvillian}
\label{app:electronic}

As an analytic expression for the electronic ground state energy $E_{\tel}$ is unavailable, we cannot follow the same strategy as described in Appendix~\ref{app:L_class} for the nuclear part to implement $e^{-iL_{\tel}t}$. In particular, we need to approximate somehow the derivative $\frac{\partial E_{\tel}}{\partial {x_{n,j}}}$ which we do via a central finite difference formula of order $2d_e$. Recall from Eq.~\eqref{el_derivative} that the resulting approximate operator is given by
 \begin{equation}
    \Del = \frac{1}{h_x} \sum_{k = -d_e}^{d_e} \sum_{(n',j') \neq (n,j)} \sum_{\hax_{n',j'}} \sum_{\hax_{n,j}} c_{d_e,k} E_{\tel}\lb \{ x_{n',j'}\}, x_{n,j} + k h_x \rb \ketbra{\hax_{n',j'}}{\hax_{n',j'}} \otimes \ketbra{\hax_{n,j}}{\hax_{n,j}}
\end{equation}
where the coefficients $\{ c_{d_e,k} \}$ are as in Definition~\ref{def:discrete_derivative}.

We prove the following lemma which upper bounds the complexity of simulating $e^{-iL_{\tel}t}$.
\begin{lem}[Complexity of simulating $e^{-iL_{\tel}t}$]
    Assume the following.
    \begin{enumerate}
        \item 
        Let $\epsilon \in (0,1)$ and $t \in \mathbb{R}_{\geq 0}$,
        \item $U_{H_{\tel}}$ be a Hermitian $\lb \lambda, a_{\tel}, \frac{h_x h_p \epsilon}{36 N d_e t}\rb$-block-encoding of the electronic Hamiltonian $H_{\tel}\lb \{x_{n}\} \rb$,
        \item $\gamma$ be a lower bound on the spectral gap of the block-encoded operator $\widetilde{H}_{\tel}\lb \{x_{n}\} \rb$ over all phase space grid points,
        \item for any $d_e \in \mathbb{N}_{+}$ it holds that 
        $\max_{x^* \in [-x_{\text{max}},x_{\text{max}}]^{3N}} \left| \frac{\partial^{(2d_e+1)} E_{\tel}}{\partial {x_{n,j}^{(2d_e+1)}}}(x^*) \right| \leq \chi u^{2d_e +1}$ for some constant $\chi$ with units of energy and $u$ with units of inverse length.
        \item Let $U_I$ be the initial state preparation oracle from Definition \ref{def:state_prep} and let $\delta$ be a lower bound on the initial overlap with the true electronic ground state of $\widetilde{H}_{\tel}$. 
    \end{enumerate}
    
    In order to implement an $\epsilon$-precise Liouvillian simulation unitary $U_{L_{\tel}}$ of $e^{-iL_{\tel}t}$ with success probability at least $1 - \xi$ it is sufficient to query $U_{H_{\tel}}$ a total number of times
     \begin{equation}
        O  \lb N d_e \log{\lb \frac{N d_e}{\xi} \rb} \lb \frac{\lambda t}{h_x h_p} + \log{\lb \frac{N\lambda \ln{(d_e)} t}{h_x h_p \epsilon} \rb}\log{\lb \frac{N d_e \log{\lb \frac{N\lambda \ln{(d_e)} t}{h_x h_p \epsilon} \rb}}{\epsilon} \rb} + \frac{\lambda}{\gamma \, \delta} \log{ \lb \frac{N d_e}{ \delta \epsilon} \rb} \rb \rb,
    \end{equation}
    where
    \begin{equation}
        d_e \in O \lb \frac{\log{\left( \frac{N \chi u t}{h_p \epsilon}\right)}}{\log{\left( \frac{1}{u h_x} \right)}} \rb.
    \end{equation}
    Furthermore, we require
    \begin{equation}
        O \lb \frac{N d_e}{\delta} \log{\lb \frac{N d_e}{\xi} \rb} \log{ \lb \frac{N d_e}{ \delta \epsilon} \rb} \rb
    \end{equation}
    queries to the initial state preparation oracle $U_I$ from Definition \ref{def:state_prep}.
\label{lem:complexity_eL_el}
\end{lem}

As mentioned in the main text, one important feature of the electronic Liouvillian $L_{\tel}$ is that all summands commute with each other (see Definition \ref{def:el_liouvillian}). The evolution operator associated with $L_{\tel}$ can thus be decomposed as follows:
\begin{align}
    e^{-iL_{\tel}t} &= e^{-i \lb i\sum_{n=1}^{N} \sum_{j=1}^{3} \Del \otimes D_{p_{n,j}}^1 \rb t} \nonumber\\
    &= \prod_{n,j} e^{\Del \otimes D_{p_{n,j}}^1 t},
\label{exp_Lel2}
\end{align}
where we use $p_{n,j}$ to denote either a real or virtual momentum variable.

Let us now explain how to implement a single exponential appearing in Eq.~(\ref{exp_Lel2}). Note that $\Del \otimes D_{p_{n,j}}^1$ acts nontrivially on the nuclear momentum register. We deal with the discrete nuclear momentum derivative $D_{p_{n,j}}^1$ via a quantum Fourier transform ($\qft$) whose action on the nuclear momentum register is defined as follows:
\begin{equation}
    \qft \ket{\hap}:= \frac{1}{\sqrt{g_p}} \sum_l e^{2\pi i \hap l/g_p}\ket{l}.
\end{equation}
Here we dropped the $n,j$ indices of the integer momentum variable $\hap$ for ease of notation.
The quantum Fourier transform diagonalizes finite difference operators. Recall that $D^1$ is a first-order finite difference operator of the form $D^1 := \frac{1}{2 h_p}\sum_{\hap} \lb \ketbra{\hap-1}{\hap} - \ketbra{\hap}{\hap-1}\rb$. Thus,
\begin{equation}
\begin{split}
    D^1 &= \lb \qft \; \qft^{-1} \rb D_1 \lb \qft \; \qft^{-1} \rb = \qft \lb \qft^{-1}  D_1 \qft \rb \qft^{-1} \\
    &= \qft \lb \frac{1}{\sqrt{g_p}} \sum_{\hap} \sum_{l,k} \sum_{l',k'} e^{-2\pi i l k/g_p} e^{2\pi i l' k'/g_p} \ketbra{l}{k} \frac{\ketbra{\hap-1}{\hap} - \ketbra{\hap}{\hap-1}}{2 h_p} \ketbra{l'}{k'} \rb \qft^{-1} \\
    &= \qft \lb \frac{1}{\sqrt{g_p}} \sum_{\hap} \sum_{l, k'} \frac{e^{-2\pi i l (\hap-1)/g_p} e^{2\pi i \hap k'/g_p} - e^{-2\pi i l \hap/g_p} e^{2\pi i (\hap-1) k'/g_p}}{2 h_p} \ketbra{l}{k'} \rb \qft^{-1} \\
    &= \qft \lb \sum_{l,k'} \frac{e^{2\pi i l /g_p} - e^{-2\pi i k'/g_p}}{2 h_p} \ketbra{l}{k'} \rb \lb \sum_{\hap} e^{2\pi i \hap (k' - l)/g_p}\rb \qft^{-1} \\
    &= \qft \lb i \sum_l \frac{\sin{\lb 2\pi l/g_p \rb}}{h_p} \ketbra{l}{l} \rb \qft^{-1}.
\end{split}
\end{equation}
The above calculation shows that $\qft$ does indeed diagonalize $D^1$. Higher-order finite difference operators can also be diagonalized via $\qft$ but will have different eigenvalues. For simplicity, we only consider a first-order finite difference operator here.
A single exponential of Eq.~(\ref{exp_Lel2}) can then be expressed as follows:
\begin{equation}
\begin{split}
    e^{\Del D_{p_{n,j}}^1 t} \equiv e^{\Del \otimes D_{p_{n,j}}^1 t} &= \lb \mathbb{1} \otimes \qft \rb e^{i\Del \otimes \sum_l \frac{\sin{\lb 2\pi l/g_p \rb}}{h_p} \ketbra{l}{l} t} \lb \mathbb{1} \otimes \qft^{-1} \rb \\
    &= \lb \mathbb{1} \otimes \qft \rb \prod_l e^{i\Del \otimes  \frac{\sin{\lb 2\pi l/g_p \rb}}{h_p} \ketbra{l}{l} t} \lb \mathbb{1} \otimes \qft^{-1} \rb
\end{split}
\end{equation}
where we used the fact that $U e^A U^\dagger = e^{U A U^{\dagger}}$ for any square matrix $A$ and unitary matrix $U$ of the same dimension.
The following lemma allows us to simplify the above expression.

\begin{lem}[Projector exponential]
    Let $\{ A_l \}_{l=1}^M$ with $A_l \in \mathbb{C}^{M \times M}$ be a set of $M$ matrices and let $\{ P_l \}_{l=1}^M$ with $P_l \in \mathbb{C}^{M \times M}$ and $P_l^2 = P_l$ for all $l \in [M]$ be a set of $M$ orthogonal projectors that satisfy $\sum_l P_l = \mathbb{1}$. Then it holds that
    \begin{equation}
        \prod_l e^{A_l \otimes P_l} = \sum_l e^{A_l} \otimes P_l. 
    \end{equation}
\label{lem:projector_exp}
\end{lem}

\begin{proof}
    Using the Taylor series expansion of a matrix exponential and the fact that $P_l^2 = P_l$, $P_k P_l = 0$ if $k \neq l$ and $\sum_l P_l = \mathbb{1}$, we obtain
    \begin{equation}
    \begin{split}
        \prod_l e^{A_l \otimes P_l} &= \prod_l \sum_{k=0}^{\infty} \frac{\lb A_l \otimes P_l \rb^k}{k!} = \prod_l \lb \mathbb{1} \otimes \mathbb{1} + \lb e^{A_l} - \mathbb{1} \rb \otimes P_l \rb \\
        &= \mathbb{1} \otimes \mathbb{1} + \sum_l \lb e^{A_l} - \mathbb{1} \rb \otimes P_l = \mathbb{1} \otimes \mathbb{1} + \sum_l  e^{A_l} \otimes P_l - \mathbb{1} \otimes \sum_l P_l \\
        &=  \sum_l  e^{A_l} \otimes P_l.
    \end{split}
    \end{equation}
\end{proof}

Lemma \ref{lem:projector_exp} implies that 
\begin{equation}
     e^{\Del D_{p_{n,j}} t} = \lb \mathbb{1} \otimes \qft \rb \sum_l e^{i\Del \frac{\sin{\lb 2\pi l/g_p \rb}}{h_p} t} \otimes \ketbra{l}{l} \lb \mathbb{1} \otimes \qft^{-1} \rb.
\end{equation}
Recalling that $\Del = \frac{1}{h_x} \sum_{k = -d_e}^{d_e} \sum_{\hax_{n,j}} c_{d_e,k} E_{\tel}(x_{n,j} + k h_x) \ketbra{\hax_{n,j}}{\hax_{n,j}}$ we obtain the following equality:
\begin{equation}
     e^{\Del \otimes D_{p_{n,j}}^1 t} = \lb \mathbb{1} \otimes \qft \rb \sum_{\hax_{n,j}} \sum_l \prod_{k} e^{i c_{d_e,k} \frac{E_{\tel}(x_{n,j} + k h_x)}{h_x}  \frac{\sin{\lb 2\pi l/g_p \rb}}{h_p} t} \ketbra{\hax_{n,j}}{\hax_{n,j}} \otimes \ketbra{l}{l} \lb \mathbb{1} \otimes \qft^{-1} \rb.
\end{equation}
We implement the above expression via controlled Hamiltonian simulation.
\begin{defn}[Controlled Hamiltonian simulation (\cite{Chakraborty2019block-encoding}, Definition 51)]
Let $M = 2^J$ for some $J \in \mathbb{N}$, $\gamma \in \mathbb{R}$ and $\epsilon \ge 0$. We say that the unitary
\begin{equation}
    W := \sum_{m=-M}^{M}\ket{m}\!\bra{m}\otimes e^{i m \gamma H}
\end{equation}
implements a controlled $(M, \gamma)$-simulation of the Hamiltonian $H$, where $\ket{m}$ denotes a (signed) bitstring $\ket{b_J b_{J-1} \dots b_0}$ such that $m = -b_J 2^J + \sum_{j=0}^{J-1} b_j 2^j$.
\end{defn}

\begin{lem}[Complexity of controlled Hamiltonian simulation (\cite{Chakraborty2019block-encoding}, Lemma 52)]
    Let $M = 2^J$ for some $J \in \mathbb{N}$, $\gamma \in \mathbb{R}$ and $\epsilon \ge 0$. Suppose that U is an $(\alpha, a, \epsilon/|8(J+1)^2 M \gamma|)$-block-encoding of the Hamiltonian $H$. Then we can implement a $(1, a+2, \epsilon)$-block-encoding of a controlled $(M, \gamma)$-simulation of the Hamiltonian $H$ with
    \begin{equation}
        O \lb |\alpha M \gamma| + J\log{(J/\epsilon)} \rb
    \end{equation}
    uses of controlled-$U$ or its inverse and with $O \lb a|\alpha M \gamma| + aJ\log{(J/\epsilon)} \rb$ two-qubit gates.
\label{lem:controlled_ham}
\end{lem}

Before applying the above lemma, we first use coherent quantum arithmetic to compute an $\epsilon_{\sin}$-precise binary approximation of $\sin{\lb 2\pi l/g_p \rb}$ in an ancilla register controlled by the $\ket{l}$ register. This can be done using a truncated Taylor series expansion of the sine function and has Toffoli cost in $O \lb \log{\lb 1/\epsilon_{\sin} \rb}\rb$ since the error of the truncation vanishes exponentially quickly. The size of the ancilla register is also in $O \lb \log{\lb 1/\epsilon_{\sin} \rb}\rb$.

Controlled by this $\sin{\lb 2\pi l/g_p \rb}$ ancilla register, we then simulate $ \exp \lb -i H_{\tel}\lb \{x_n\} \rb t_{c_k, l} \rb$ using $U_{H_{\tel}}$, a Hermitian $\lb \lambda, a_{\tel}, \epsilon_{\text{b-e}} \rb$-block-encoding of $H_{\tel}$ where
\begin{equation}
    t_{c_k, l} := \frac{c_{d_e,k}}{h_x} \frac{\sin{\lb 2\pi l/g_p \rb}}{h_p} t
\end{equation}
is a rescaled time variable depending on the finite difference coefficients $\{c_{d_e, k}\}$ of $\Del$.
For convenience, let $\widetilde{H}_{\tel} := \lambda \lb \bra{0} \otimes \mathbb{1} \rb U_{H_{\tel}} \lb \ket{0} \otimes \mathbb{1} \rb$ denote the Hermitian matrix that $U_{H_{\tel}}$ block-encodes. This implies that $\norm{H_{\tel} - \widetilde{H}_{\tel}} \leq \epsilon_{\text{b-e}}$.

Note that Lemma~\ref{lem:controlled_ham} applies to integer values of $m$ which translates to integer values of $\sin{\lb 2\pi l/g_p \rb}$ in our case. Hence, we need to ``blow up'' the values of $\sin{\lb 2\pi l/g_p \rb}$ by a factor of $O \lb 1/\epsilon_{\sin} \rb$. This then entails a renormalization of the exponent by a factor of $ O \lb \epsilon_{\sin} \rb$, which can be done via a rescaling of the form $\gamma \rightarrow \epsilon' \gamma$ for an appropriate $\epsilon' \in O \lb \epsilon_{\sin} \rb$.

The general strategy is first to shift the nuclear position register of a single nuclear position variable $\hax_{n,j}$ according to the finite difference scheme of order $2 d_e$. This is done using a unitary adder. Then we (approximately) prepare the ground state $\ket{\widetilde{\psi} \{ x_n \}}$ of $\widetilde{H}_{\tel}$, controlled by the entire nuclear position register $\ket{\{\hax_{n,j}\}}$, in the electronic register. Next, controlled by the entire nuclear position register $\ket{\{\hax_{n,j}\}}$ and the Fourier-transformed momentum register associated with the 1d momentum variable $\hap_{n,j}$, $\ket{l}$, we apply $\exp \lb -i H_{\tel}\lb \{x_n\} \rb t_{c_k, l} \rb$ to the electronic register. This generates states of the form
\begin{equation}
    e^{i c_{d_e,k} \frac{E_{\tel}(x_{n,j} + k h_x)}{h_x}  \frac{\sin{\lb 2\pi l/g_p \rb}}{h_p} t} \ket{\hax_{n,j} + k} \ket{l}.
\end{equation}
Finally, we uncompute the electronic ground state. Now we simply repeat the above procedure for each stencil point of the finite difference scheme. More precisely, we shift the nuclear position register of the 1d nuclear position variable of interest, $x_{n,j}$, to the next stencil point, prepare the electronic ground state for that nuclear configuration, and then apply $\exp \lb -i H_{\tel}\lb \{x_n\} \rb t_{c_k, l} \rb$ to the electronic register.
In the last step, we shift the position register corresponding to $x_{n,j}$ back to the original state to obtain the desired phase factor
\begin{equation}
    \prod_k e^{i c_{d_e,k} \frac{E_{\tel}(x_{n,j} + k h_x)}{h_x}  \frac{\sin{\lb 2\pi l/g_p \rb}}{h_p} t} \ket{\hax_{n,j}} \ket{l}.
\end{equation}
The overall procedure for simulating $e^{-iL_{\tel}t}$ is summarized in Algorithm \ref{alg:el_ev} as well as Fig.~\ref{fig:Circuit_L_el}.

One might consider using the gradient computation algorithm developed in~\cite{Gilyen2019gradient} to compute $\Del$ in the exponent. The hope is that $O\lb N^{1/2} \rb$ rather than $O \lb N \rb$ evaluations of the electronic ground state energy are sufficient. However, a straightforward application fails in our case since we have to compute the gradient in superposition over all nuclear positions. This is problematic because the gradient computation algorithm produces different global phases for different nuclear positions, i.e.~the global phases become local phases that cannot simply be ignored. Uncomputing these local phases is nontrivial and left for future work.

Both the simulation of $\exp \lb -i H_{\tel}\lb \{x_n\} \rb t_{c_k, l} \rb$ and the electronic ground state preparation require access to $U_{H_{tel}}$, a block-encoding of $H_{\tel}\lb \{x_n\} \rb$ as given in Definition \ref{H_el}. For the ground state preparation, it is important that $U_{H_{tel}}$ is a Hermitian block-encoding, meaning that $\widetilde{H}_{\tel}$ is Hermitian. This is discussed in more detail in the proof of Lemma \ref{lem:complexity_eL_el}.
Note that the only way the nuclear positions $\{x_n\}$ enter the electronic Hamiltonian is via the phase factors of the electron-nucleus interaction terms. In~\cite{Su2021first_quant_sim}, the nuclear positions are accessed via a QROM. The phase $k_{c-b}\cdot \hax_n$ is computed in an ancilla register which is then hit with a phase gradient to produce the phase factor $\exp \lb ik_{c-b}\cdot x_n \rb$. In our case, instead of using a QROM to access $x_n$, we swap the nuclear position register $\ket{\hax_n}$ into an ancilla register and compute $k_{c-b}\cdot \hax_n$. The swap is controlled by the ancilla register preparing the state $\sum_{n=0}^{N-1} \sqrt{\frac{Z_n}{Z}}\ket{n}$ which is needed for block-encoding the electron-nucleus interaction terms. The Toffoli cost associated with the controlled $\swap$s is in $O(N)$, which matches the complexity of the original QROM model.
Apart from accessing the nuclear positions differently, we can employ exactly the same techniques presented in~\cite{Su2021first_quant_sim} to block-encode $H_{\tel}\lb \{x_n\}\rb$ which leads to the complexity expressions of Lemma \ref{lem:block-encode_Hel}.

Let us now discuss the electronic ground state preparation in more detail.
\begin{defn}[Fidelity]
    Let $\ket{x}, \ket{y} \in \mathbb{C}^{2^n \times 2^n}$ be two quantum states. The fidelity $F(x,y)$ between $\ket{x}$ and $\ket{y}$ is given by
    \begin{equation}
        F(x,y) := | \braket{x}{y} |.
    \end{equation}
\end{defn}

\begin{lem}[Ground state preparation with \textit{a priori} ground state energy bound (\cite{Lin2020ground_state}, Theorem 6, reformulated)] 
    Suppose we have a Hamiltonian $\widetilde{H} = \sum_k \widetilde{E}_k \ket{\widetilde{\psi}_k}\!\bra{\widetilde{\psi}_k} \in \mathbb{C}^{N \times N}$, where $\widetilde{E}_k \le \widetilde{E}_{k+1}$, which is given through its $(\lambda, m, 0)$-block-encoding $U_H$. Also suppose we have an initial state $\ket{\phi_0}$, prepared by a unitary $U_I$, together with a lower bound on the overlap $|\braket{\widetilde{\psi}_0}{\phi_0}| \ge \delta$. Furthermore, we require the following bound on the ground state energy and the spectral gap: $\widetilde{E}_0 \le \mu - \gamma/2 < \mu + \gamma/2 \le \widetilde{E}_1$, where $\mu$ is an upper bound on the ground state energy and $\gamma$ is a lower bound on the spectral gap of $\widetilde{H}$.
    Then the ground state $\ket{\widetilde{\psi}_0}$ can be prepared with fidelity at least $1-\epsilon_{\text{prep}}$ using
    \begin{equation}
        O \lb \frac{\lambda}{\gamma \delta} \log{\lb \frac{1}{\delta \epsilon_{\text{prep}}} \rb}  \rb
    \end{equation}
    queries to $U_H$ and
    \begin{equation}
        O \lb \frac{1}{\delta} \rb
    \end{equation}
    queries to $U_I$.
\label{lem:ground_state_prep}
\end{lem}

\begin{figure}
    \centering
    \begin{adjustbox}{max width=\textwidth}
        \Qcircuit @C=1em @R=1.5em  {
           \lstick{\ket{0}} & \qw & \multigate{1}{W} & \multigate{2}{e^{-iH_{\tel} t_{c_{-1}, l}}} & \multigate{1}{W^{-1}}  & \qw & \multigate{1}{W} & \multigate{2}{e^{-iH_{\tel} t_{c_{1}, l}}} & \multigate{1}{W^{-1}} & \qw & \qw \\
           \lstick{\ket{\hax_{1,1}}} & \gate{-1} & \ghost{W} & \ghost{e^{-iH_{\tel} t_{c_{-1}, l}}} & \ghost{W^{-1}}  & \gate{+2} & \ghost{W} & \ghost{e^{-iH_{\tel} t_{c_{1}, l}}} & \ghost{W^{-1}} & \gate{-1} & \qw \\
           \lstick{\ket{\hap_{1,1}}} & \gate{\qft} & \qw & \ghost{e^{-iH_{\tel} t_{c_{-1}, l}}} & \qw & \qw & \qw & \ghost{e^{-iH_{\tel} t_{c_{1}, l}}} & \qw & \gate{\qft^{-1}} & \qw
    }
    \end{adjustbox}
    \caption{Circuit for implementing the evolution under the electronic Liouvillian for a single nucleus in 1d. The top register corresponds to the electronic register. ``$-1$'' denotes a unitary adder of the form $\sum_{\hax} \ketbra{\hax + 1}{\hax}$ and similarly, ``$+2$'' denotes a unitary adder of the form $\sum_{\hax} \ketbra{\hax -2}{\hax}$.}
    \label{fig:Circuit_L_el}
\end{figure}
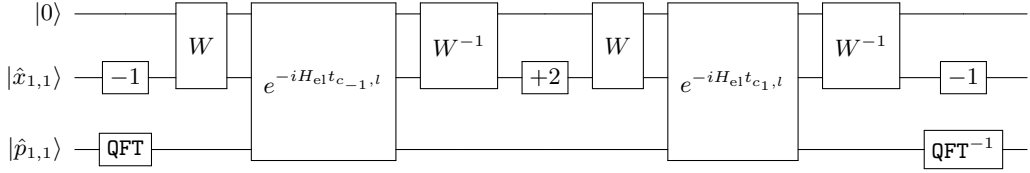

Now we are ready to prove Lemma \ref{lem:complexity_eL_el} which provides an upper bound on the complexity of simulating the evolution under the electronic Liouvillian.

\begin{proof}[Proof of Lemma~\ref{lem:complexity_eL_el}]
    The error in approximating $e^{-iL_{\tel}t}$ consists of two parts. On the one hand, there is the simulation error $\epsilon_{\text{sim}}$, which arises from approximately preparing the exact ground state $\ket{\widetilde{\psi}_0}$ of $\widetilde{H}_{\tel}$ and approximately implementing $\exp \lb -i H_{\tel}\lb \{x_n\} \rb t_{c_k, l} \rb$. On the other hand, we have the discretization error $\epsilon_{\text{disc}}$ associated with the finite difference matrix $\Del$ of the derivatives of the electronic ground state energy and the finite binary representation of $\sin{\lb 2\pi l/g_p \rb}$. Let $\widetilde{L_{\tel}}$ denote the approximate discrete electronic Liouvillian. Then the overall error associated with simulating $e^{-iL_{\tel}t}$ is upper bounded as follows:
    \begin{equation}
        \begin{split}
            \norm{\lb \bra{0} \otimes \mathbb{1} \rb U_{L_{\tel}} \lb \ket{0} \otimes \mathbb{1} \rb - e^{-iL_{\tel}t}} &\le \norm{\lb \bra{0} \otimes \mathbb{1} \rb U_{L_{\tel}} \lb \ket{0} \otimes \mathbb{1} \rb - e^{-i \widetilde{L_{\tel}}t}} + \norm{e^{-i \widetilde{L_{\tel}}t} -e^{-iL_{\tel}t}} \\
            &\le \epsilon_{\text{sim}} + \norm{\widetilde{L_{\tel}} - L_{\tel}}t\\
            &\le \epsilon_{\text{sim}} + \epsilon_{\text{disc}}t,
        \end{split}
    \end{equation}
    where we used Duhamel's formula in going from the second to the third line. The overall error is less or equal to $\epsilon$ if we ensure that $\epsilon_{\text{sim}} \leq \epsilon/2$ and $\epsilon_{\text{disc}} \leq \epsilon/(2t)$.
    Recall that $L_{\tel}$ is a sum of $3N$ commuting terms and each term involves a central finite difference formula of order $2 d_e$. Thus, a total of $6 N d_e$ exponentials need to be implemented.
    We first discuss the simulation error $\epsilon_{\text{sim},1}$ of a single exponential. By the triangle inequality, if $\epsilon_{\text{sim},1} \leq \epsilon/(12 N d_e)$ then $\epsilon_{\text{sim}} \leq \epsilon/2$. Let $W$ denote the unitary that prepares an approximate ground state of $\widetilde{H}_{\tel}$ for fixed nuclear positions according to Lemma \ref{lem:ground_state_prep}, i.e. 
    \begin{equation}
        W\ket{\{\hax_{n,j}\}}\ket{0} = \ket{\{\hax_{n,j}\}}\ket{\widetilde{\phi}_0 \lb \{x_{n,j}\} \rb}
    \end{equation}
    with
    \begin{equation}
        |\braket{\widetilde{\psi}_0 \lb \{x_{n,j}\} \rb}{\widetilde{\phi}_0 \lb \{x_{n,j}\} \rb}| \ge 1- \epsilon_{\text{prep}}.
    \end{equation}
    Note that we can view $U_{H_{\tel}}$ as an exact block-encoding of $\widetilde{H}_{\tel}$, which allows us to use Lemma \ref{lem:ground_state_prep} directly without further error propagation.
    In the following discussion we will refrain from writing out the $\lb \{x_{n,j}\} \rb$-dependence explicitly.
    Now, it holds that 
    \begin{equation}
        \ket{\widetilde{\phi}_0} = e^{i\alpha} \lb 1- \epsilon'_{\text{prep}} \rb \ket{\widetilde{\psi}_0} + \beta \ket{\widetilde{\psi}_0^\perp}
    \end{equation}
    for some angle $\alpha \in [0, 2\pi)$, $0 \le \epsilon'_{\text{prep}} \le \epsilon_{\text{prep}}$ and $|\beta|^2 = 2\epsilon'_{\text{prep}} - \lb \epsilon'_{\text{prep}} \rb^2 \le 2 \epsilon'_{\text{prep}}$. Letting 
    \begin{equation}
        \ket{\psi_0'} := e^{i\alpha}\ket{\widetilde{\psi}_0}
    \end{equation}
    we thus have that
    \begin{equation}
        \norm{\ket{\widetilde{\phi}_0} - \ket{\psi_0'}} \le \sqrt{2\epsilon_{\text{prep}}}.
    \end{equation}

    Let $U_{\text{e,el}}$ be an $\epsilon_{\text{e,el}}$-precise block-encoding of 
    \begin{equation}
        \sum_{\{\hax_n\}, l} \ketbra{\{\hax_n\}}{\{\hax_n\}} \otimes \ketbra{\sin{\lb 2\pi l/g_p \rb}}{\sin{\lb 2\pi l/g_p \rb}} \otimes e^{-i H_{\tel}t_{c_k, l}}
    \end{equation}
    and let $\widetilde{U}_{\text{e,el}} := \lb \bra{0} \otimes \mathbb{1} \rb U_{\text{e,el}} \lb \ket{0} \otimes \mathbb{1} \rb$ denote the block-encoded approximation to the above operator. Our goal is to bound the error of the phase factors obtained via phase kick back from the electronic register, i.e.~we wish to bound
    \begin{equation}
        \epsilon_{\text{sim},1} := \norm{\lb \mathbb{1} \otimes \bra{0} \rb W^{-1}  \widetilde{U}_{\text{e,el}} W \lb \mathbb{1} \otimes \ket{0} \rb  - \sum_{\{\hax_n\}} e^{-i E_{\tel} \lb \{x_{n,j}\} \rb t_{c_k, l}} \ketbra{\{ \hax_n \}}{\{ \hax_n \}}},
    \end{equation}
    for fixed $t_{c_k, l}$.
    Note that the above definition implies that the electronic register is projected out to the $\ket{0}$ state at the end of the simulation. In other words, the error $\epsilon_{\text{sim},1}$ is only measured within the Hilbert space of the nuclear position and momentum registers but not the electronic register. Importantly, the error matrix 
    \begin{equation}
        \mathcal{E}_{{\text{sim},1}} := \lb \bra{0} \otimes \mathbb{1} \rb W^{-1}  \widetilde{U}_{\text{e,el}} W \lb \ket{0} \otimes \mathbb{1} \rb  - \sum_{\{\hax_n\}} e^{-i E_{\tel} \lb \{x_{n,j}\} \rb t_{c_k, l}} \ketbra{\{ \hax_n \}}{\{ \hax_n \}}
    \end{equation}
    is diagonal in the nuclear position and momentum basis since $W$ and $\widetilde{U}_{\text{e,el}}$ act trivially on the nuclear position and momentum register. Hence, $\epsilon_{\text{sim},1}$ is simply the largest value on the diagonal of $\mathcal{E}_{{\text{sim},1}}$. This allows us to consider the phase error for each nuclear computational basis state separately. 
    Let 
    \begin{equation}
        \widetilde{U}_{\text{e,el}} \lb \{ x_n \}, t_{c_k, l} \rb := \lb \bra{\{\hax_n\}}\bra{\sin{\lb 2\pi l/g_p \rb}} \otimes \mathbb{1} \rb \widetilde{U}_{\text{e,el}} \lb \ket{\{\hax_n\}} \ket{\sin{\lb 2\pi l/g_p \rb}} \otimes \mathbb{1} \rb
    \end{equation}
    denote a single exponential of $\widetilde{U}_{\text{e,el}}$ for fixed nuclear positions $\{x_n\}$ and Fourier parameter $l$. Similarly, let $W \lb \{ x_n \} \rb$ denote the electronic ground state preparation unitary for fixed nuclear positions $\{x_n\}$. Then we have that
    \begin{equation}
    \begin{split}
        \epsilon_{\text{sim},1} &= \max_{\{ \hax_n \}, l} \left| \bra{0} W^{-1} \lb \{ x_n \} \rb \widetilde{U}_{\text{e,el}} \lb \{ x_n \}, t_{c_k, l} \rb W \lb \{ x_n \} \rb \ket{0} - e^{-i E_{\tel}t_{c_k, l}\lb \{ x_n \} \rb} \right| \\
        &\leq \max_{\{ \hax_n \}, l} \norm{W^{-1} \lb \{ x_n \} \rb \widetilde{U}_{\text{e,el}} \lb \{ x_n \}, t_{c_k, l} \rb W \lb \{ x_n \} \rb \ket{0} - e^{-i E_{\tel}t_{c_k, l}\lb \{ x_n \} \rb}\ket{0}}.
    \end{split}
    \end{equation}
    In the following, we will not write out the $\{ x_n \}$ and $l$ dependence explicitly.
    Applying the triangle inequality repeatedly and using the submultiplicativity of the induced 2-norm, one finds the following upper bound on the approximation error for a single exponential: 
    \begin{equation}
    \begin{split}
        \epsilon_{\text{sim},1} &\leq \norm{W^{-1} \widetilde{U}_{\text{e,el}} W \ket{0} - e^{-i E_{\tel}t_{c_k}} \ket{0}} \\
        &\leq \norm{W^{-1} \widetilde{U}_{\text{e,el}} W \ket{0} - W^{-1} \widetilde{U}_{\text{e,el}} \ket{\psi_0'}} + \norm{W^{-1} \widetilde{U}_{\text{e,el}} \ket{\psi_0'} - e^{-i E_{\tel}t_{c_k}} \ket{0}} \\
        &\leq \norm{\ket{\widetilde{\phi}_0} - \ket{\psi_0'}} \\
        &\quad + \norm{W^{-1} \widetilde{U}_{\text{e,el}} \ket{\psi_0'} - W^{-1} e^{-i \widetilde{E}_{\tel}t_{c_k}} \ket{\psi_0'}} + \norm{W^{-1} e^{-i \widetilde{E}_{\tel}t_{c_k}} \ket{\psi_0'} - e^{-i E_{\tel}t_{c_k}} \ket{0}} \\
        &\leq \sqrt{2 \epsilon_{\text{prep}}} + \norm{\widetilde{U}_{\text{e,el}} \ket{\psi_0'} - e^{-i \widetilde{E}_{\tel}t_{c_k}} \ket{\psi_0'}} \\
        &\quad + \norm{W^{-1} e^{-i \widetilde{E}_{\tel}t_{c_k}} \ket{\psi_0'} - W^{-1} e^{-i E_{\tel}t_{c_k}} \ket{\psi_0'}} + \norm{W^{-1} e^{-i E_{\tel}t_{c_k}} \ket{\psi_0'} - e^{-i E_{\tel}t_{c_k}} \ket{0}} \\
    \end{split}
    \end{equation}
    The second term is upper bounded by the block-encoding error $\epsilon_{\text{e,el}}$ of $U_{\text{e,el}}$. Duhamel's formula can be used to upper bound the third term:
    \begin{equation}
    \begin{split}
        \norm{W^{-1} e^{-i \widetilde{E}_{\tel}t_{c_k}} \ket{\psi_0'} - W^{-1} e^{-i E_{\tel}t_{c_k}} \ket{\psi_0'}} &\leq \left| e^{-i \widetilde{E}_{\tel}t_{c_k}} - e^{-i E_{\tel}t_{c_k}} \right| \leq \left| \widetilde{E}_{\tel}t_{c_k} -  E_{\tel}t_{c_k} \right| \\
        &\leq \left| \widetilde{E}_{\tel} -  E_{\tel} \right| \frac{t}{h_x h_p}.
    \end{split}
    \end{equation}
    Now recall that $\norm{\widetilde{H}_{\tel} - H_{\tel}} \leq \epsilon_{\text{b-e}}$.
    Eigenvalue perturbation theory then tells us that~\cite{Horn1985}
    \begin{equation}
        \left| \widetilde{E}_{\tel} -  E_{\tel} \right| \leq \epsilon_{\text{b-e}}.
    \end{equation}
    For the last term note that
    \begin{equation}
    \begin{split}
        \norm{W^{-1}\ket{\psi_0'} - \ket{0}} &= \norm{W^{-1}\ket{\psi_0'} - W^{-1} W \ket{0}} \\
        &\leq \norm{W^{-1}}\norm{\ket{\psi_0'} - W \ket{0}} = \norm{\ket{\psi_0'} - \ket{\widetilde{\phi}_0}} \\
        &\le \sqrt{2\epsilon_{\text{prep}}}.
    \end{split}
    \end{equation}
    Putting it all together we find that
    \begin{equation}
        \epsilon_{\text{sim},1} \leq 2 \sqrt{2 \epsilon_{\text{prep}}} + \epsilon_{\text{e,el}} + \frac{t}{h_x h_p} \epsilon_{\text{b-e}}.
    \end{equation}

    To achieve $\epsilon_{\text{sim},1} \leq \frac{\epsilon}{12 N d_e}$ it suffices to have
    \begin{align}
        \epsilon_{\text{prep}} &\leq \frac{1}{2} \lb \frac{\epsilon}{72 N d_e} \rb^2 \\
        \epsilon_{\text{e,el}} &\leq \frac{\epsilon}{36 N d_e} \\
        \epsilon_{\text{b-e}} &\leq \frac{h_x h_p \epsilon}{36 N d_e t}.
    \end{align}
    Next, let us discuss the discretization error $\epsilon_{\text{disc},1}$ of a single term of $L_{\tel}$. Let $\epsilon_{\Del}$ denote the error tolerance associated with approximating $\sum_{\hax_{n,j}} \frac{\partial E_{\tel}}{\partial {x_{n,j}}} \ketbra{\hax_{n,j}}{\hax_{n,j}}$ with $\Del$, i.e. 
    \begin{equation}
        \norm{\Del - \sum_{\hax_{n,j}} \frac{\partial E_{\tel}}{\partial {x_{n,j}}} \ketbra{\hax_{n,j}}{\hax_{n,j}}} \leq \epsilon_{\Del}.
    \end{equation}
    Furthermore, let $\widetilde{\sin}{\lb 2\pi l/g_p \rb}$ denote an approximation to $\sin{\lb 2\pi l/g_p \rb}$ satisfying 
    \begin{equation}
        |\widetilde{\sin}{\lb 2\pi l/g_p \rb} - \sin{\lb 2\pi l/g_p \rb}| \leq \epsilon_{\sin}
    \end{equation}
    for all $l \in [g_p]$.
    By the triangle inequality, we then have that
    \begin{equation}
    \begin{split}
        \epsilon_{\text{disc},1} &= \norm{ \Del \frac{\widetilde{\sin}{\lb 2\pi l/g_p \rb}}{h_p} - \sum_{\hax_{n,j}} \frac{\partial E_{\tel}}{\partial {x_{n,j}}} \ketbra{\hax_{n,j}}{\hax_{n,j}} \frac{\sin{\lb 2\pi l/g_p \rb}}{h_p}} \\ 
        &\leq \norm{\Del \frac{\widetilde{\sin}{\lb 2\pi l/g_p \rb}}{h_p} - \Del \frac{\sin{\lb 2\pi l/g_p \rb}}{h_p}} + \norm{\Del \frac{\sin{\lb 2\pi l/g_p \rb}}{h_p} - \sum_{\hax_{n,j}} \frac{\partial E_{\tel}}{\partial {x_{n,j}}} \ketbra{\hax_{n,j}}{\hax_{n,j}} \frac{\sin{\lb 2\pi l/g_p \rb}}{h_p}} \\
        &\leq \frac{\norm{\Del}}{h_p} \epsilon_{\sin} + \frac{1}{h_p}\epsilon_{\Del}.
    \end{split}
    \end{equation}
    We obtain $\epsilon_{\text{disc}} \leq \frac{\epsilon}{2t}$ if $\epsilon_{\text{disc},1} \leq \frac{\epsilon}{6Nt}$. This can be achieved by ensuring that 
    \begin{align}
        \epsilon_{\sin} &\leq \frac{h_p \epsilon}{12N\norm{\Del}t} \\
        \epsilon_{\Del} &\leq \frac{h_p \epsilon}{12Nt}.
    \end{align}
    From Lemma \ref{lem:bound_FD_coeff} it follows that 
    \begin{equation}
        \norm{\Del} \leq \lambda \frac{2 \ln{(d_e+1)}}{h_x}.
    \end{equation}
    The size of the ancilla register used for representing $\sin{\lb 2\pi l/g_p \rb}$ is thus in 
    \begin{equation}
        O \lb \log{\lb \frac{1}{\epsilon_{\sin}} \rb} \rb \subseteq O \lb \log{\lb \frac{N\lambda \ln{(d_e)} t}{h_x h_p \epsilon} \rb} \rb.
    \end{equation}
    The order of the finite difference approximation, $d_e$, is constrained by 
    \begin{equation}
        \epsilon_{\Del} \leq \frac{h_p \epsilon}{12Nt}.
    \label{de_constraint}
    \end{equation}
    Lemma \ref{lem:finite_diff} implies that
    \begin{equation}
        \epsilon_{\Del} \in O \left( \max_{x^* \in [-x_{\text{max}},x_{\text{max}}]^{3N}} \left| \frac{\partial^{(2d_e+1)} E_{\tel}}{\partial {x_{n,j}^{(2d_e+1)}}}(x^*) \right| \left(  \frac{e h_x}{2} \right)^{2d_e} \right).
    \end{equation}
    
    By assumption,
    \begin{equation}
        \max_{x^* \in [-x_{\text{max}},x_{\text{max}}]^{3N}} \left| \frac{\partial^{(2d_e+1)} E_{\tel}}{\partial {x_{n,j}^{(2d_e+1)}}}(x^*) \right| \leq \chi u^{2d_e +1}.
    \end{equation}

    We can satisfy the constraint in Eq.~\eqref{de_constraint} by choosing
    \begin{equation}
        d_e \in O \left( \frac{\log{\left( \frac{h_p \epsilon}{N \chi u t}\right)}}{\log{\left( u h_x \right)}} \right) = O \lb \frac{\log{\left( \frac{N \chi u t}{h_p \epsilon}\right)}}{\log{\left( \frac{1}{u h_x} \right)}} \rb.
    \end{equation}

    By Lemma \ref{lem:ground_state_prep} the $6 N d_e$ electronic ground state preparations require
    \begin{equation}
        O \lb N d_e \frac{\lambda}{\gamma \delta} \log{\lb \frac{N d_e}{\delta \epsilon} \rb} \rb
    \end{equation}
    queries to $U_{H_\tel}$ and
    \begin{equation}
        O \lb \frac{N d_e}{\delta} \log{\lb \frac{N d_e}{\delta \epsilon} \rb} \rb
    \end{equation}
    queries to the initial state preparation oracle from Definition \ref{def:state_prep}.
    Furthermore, by Lemma \ref{lem:controlled_ham}, we need
    \begin{equation}
        O \lb N d_e \lb \frac{\lambda t}{h_x h_p} + \log{\lb \frac{N\lambda \ln{(d_e)} t}{h_x h_p \epsilon} \rb}\log{\lb \frac{N d_e \log{\lb \frac{N\lambda \ln{(d_e)} t}{h_x h_p \epsilon} \rb}}{\epsilon} \rb} \rb \rb
    \end{equation}
    queries to $U_{H_\tel}$ for the $6 N d_e$ controlled simulations of $e^{-i H_{\tel}t_{c_k, l}}$. 

    Lastly, note that the simulation of each of the $6N d_e$ exponentials is associated with a certain failure probability due to the probabilistic nature of block-encodings.
    By the union bound, we can ensure an overall success probability of at least $1 - \xi$ if the failure probability of a single exponential is in $O \lb \xi/(N d_e) \rb$. This can be achieved via (fixed-point) amplitude amplification at the expense of a multiplicative factor of $\log \lb \frac{N d_e}{\xi} \rb$ to the query complexities of $U_{H_\tel}$ and $U_I$.

    Combining all results yields the complexity expressions stated in Lemma \ref{lem:complexity_eL_el}.
\end{proof}

\section{Implementation of the overall Liouvillian evolution operator}
\label{app:overall_Liouvillian}

The main goal of this appendix is to prove Theorem \ref{thm:complexity_liouvillian} which upper bounds the complexity of simulating $e^{-iLt}$. Let us first discuss some intermediate results.
As explained previously, we implement the overall Liouvillian evolution operator $e^{-iL t}$ via a $(2k)$th-order Trotter product formula combining $e^{-iL_{\tcla} t}$ and $e^{-iL_{\tel} t}$.
The following lemma provides an upper bound on the query complexity of simulating Liouvillian dynamics in the $NVE$ and $NVT$ ensemble.
\begin{lem}[Query complexity of Born-Oppenheimer Liouvillian simulation]
    Let $L = L_{\tcla} + L_{\tel}$ be the discrete Liouvillian operator either in the $NVE$ ensemble (Definition \ref{def:NVE_liouvillian}) or the $NVT$ ensemble (Definition \ref{def:NVT_liouvillian}). Let $k \in \mathbb{N}_+$.
    An $\epsilon$-precise approximation to the evolution operator $U_L = e^{-iL t}$ can be implemented with success probability $\ge 1 - \xi$ using
    \begin{equation*}
        \widetilde{O} \lb 5^k t \lb \alpha \log{\lb \frac{\mu'}{\epsilon \xi} \rb} + \mu' \lb \frac{\mu' t}{\epsilon} \rb^{1/2k} \log{\lb \frac{1}{\xi} \rb} \rb \rb
    \end{equation*}
    queries to an $\lb \alpha, -, \frac{\epsilon}{5^k t} \rb$ block-encoding of the classical Liouvillian $L_{\tcla}$ where $\alpha \in \{ \alpha_{NVE}, \alpha_{NVT} \}$ and $\mu' \in \{ \mu'_{NVE}(2k), \mu'_{NVT}(2k) \}$ is an upper bound on the spectral norm of the nested commutator of $L_{\tcla}$ and $L_{\tel}$ as given in Definition \ref{def:commutator_spectral_norm}.
    An additional
    \begin{equation*}
       \widetilde{O} \lb 5^k N d_e \frac{\lambda}{\gamma \delta} \frac{\lb \mu' t \rb^{1+1/(2k)}}{\epsilon^{1/(2k)}} \log{\lb \frac{1}{\xi} \rb} \rb
    \end{equation*}
    queries to a Hermitian $\lb \lambda, -, \frac{h_x h_p \epsilon}{5^k 36 N d_e t} \rb$-block-encoding of the electronic Hamiltonian $H_{\tel}$ are needed.
    Lastly, we require
    \begin{equation*}
        \widetilde{O} \lb \frac{5^k N d_e}{\delta} \frac{\lb \mu' t \rb^{1+1/(2k)}}{\epsilon^{1/(2k)}} \log{\lb \frac{1}{\xi} \rb} \rb
    \end{equation*}
    queries to the initial electronic state preparation oracle $U_I$ from Definition \ref{def:state_prep}.
\label{lem:query_complexity_eL}
\end{lem}

To prove Lemma \ref{lem:query_complexity_eL} we first need to discuss the complexity of quantum simulation via a higher-order Trotter product formula.

\begin{lem}[Trotter error with commutator scaling (\cite{Childs2021trotter_error}, Theorem 6, Corollary 7)]
    Let $L = \sum_{\gamma=1}^{\Gamma} L_\gamma$ be an operator consisting of $\Gamma$ Hermitian summands and $t \ge 0$. Let
    \begin{equation}
        \mathcal{S}_{\ell}(t)=\prod_{\upsilon=1}^{\Upsilon} \prod_{\gamma=1}^{\Gamma} e^{-i a_{(\upsilon, \gamma)} L_{\pi_\upsilon(\gamma)}t}
    \end{equation}
    be an $\ell$th-order product formula with $\ell \in \mathbb{N}_+$. Define
    \begin{equation}
        \Tilde{\alpha}_{c}(\ell) := \sum_{\gamma_1, \gamma_2, \dots, \gamma_{\ell+1}}^\Gamma \norm{[L_{\gamma_{\ell+1}}, \cdots [L_{\gamma_2}, L_{\gamma_1}] \cdots ]}.
    \end{equation}
    Then the additive Trotter error, defined by $\mathcal{S}(t) = e^{-iLt} + \mathcal{A}(t)$, can be asymptotically bounded as
    \begin{equation}
        \norm{\mathcal{A}(t)} \in O(\Tilde{\alpha}_{c} t^{\ell+1}).
    \end{equation}
    We have $\norm{\mathcal{S}_{\ell}^r(t/r) - e^{-iLt}} \in O(\epsilon)$ if
    \begin{equation}
        r \in O \left(\frac{\Tilde{\alpha}_{c}^{1/\ell} t^{1 + 1/\ell}}{\epsilon^{1/\ell}}\right).
    \end{equation}
\label{lem:trotter}
\end{lem}

Note that $\Tilde{\alpha}_{c}(\ell) \leq \lb 2 \sum_{\gamma=1}^{\Gamma} \norm{L_\gamma} \rb^{\ell+1}$. We can establish a tighter bound as shown below.

The following definition will simplify the subsequent discussion of bounding $\Tilde{\alpha}_{c}(\ell)$ for Liouvillian simulations in the Born-Oppenheimer approximation.

\begin{defn}[Terms of the Liouvillian]
    Let $L$ be the discrete Liouvillian operator either in the $NVE$ ensemble (Definition \ref{def:NVE_liouvillian}) or the $NVT$ ensemble (Definition \ref{def:NVT_liouvillian}). Then we define
    \begin{align}
        K_{n,j}^{(NVE)} &:= \frac{\partial H_{\tcla}^{(NVE)}}{\partial {p_{n,j}}} D_{x_{n,j}} = \sum_{\hap_{n,j}} D_{x_{n,j}} \otimes \frac{p_{n,j}}{m_n} \ketbra{{\hap_{n,j}}}{{\hap_{n,j}}} \\
        K_{n,j}^{(NVT)} &:= \frac{\partial H_{\tcla}^{(NVT)}}{\partial {p_{n,j}}} D_{x_{n,j}} = \sum_{\hap_{n,j}} \sum_{\has} D_{x_{n,j}} \otimes \frac{p_{n,j}}{m_n (s + s_{\min})^2} \ketbra{{\hap_{n,j}}}{{\hap_{n,j}}} \otimes \ketbra{\has}{\has} \\
        V_{n,n',j}^{\tcla} &:= \lb \frac{\partial H_{\tcla}^{(NVE/NVT)}}{\partial {x_{n,j}}} \rb_{n'} D_{p_{n,j}} =  \sum_{\hax_{n}} \sum_{\hax_{n'}} \frac{Z_n Z_{n'}(x_{n,j} - x_{n',j})}{\lb \norm{x_n - x_{n'}}^2 + \Delta^2 \rb^{3/2}} \ketbra{\hax_{n}}{\hax_{n}} \otimes \ketbra{\hax_{n'}}{\hax_{n'}} \otimes D_{p_{n,j}} \\
        V_{n,j}^{\tel} &:= \Del \otimes D_{p_{n,j}}^1 \\
        K^{\text{bath}} &:= \frac{\partial H_{\tcla}^{(NVT)}}{\partial_{p_{s}}} D_{s} = \sum_{\hap_{s}} D_{s} \otimes \frac{p_{s}}{Q} \ketbra{{\hap_{s}}}{\hap_{s}}  \\
        V^{\text{bath}}_{n,j} &:= \lb \frac{\partial H_{\tcla}^{(NVT)}}{\partial_{s}} D_{p_s} \rb_{n,j} = -\sum_{\hap_{n,j}} \sum_{\has} \frac{2 p_{n,j}^2}{m_{n} (s + s_{\min})^3} \ketbra{{\hap_{n,j}}}{{\hap_{n,j}}} \otimes \ketbra{{\has}}{\has} \otimes D_{p_s} \\
        V^{\text{bath}}_T &:= \lb \frac{\partial H_{\tcla}^{(NVT)}}{\partial_{s}} D_{p_s} \rb_T = \sum_{\has} \frac{N_f k_B T}{s+s_{\min}} \ketbra{{\has}}{\has} \otimes D_{p_s} .
    \end{align}
\label{def:terms}
\end{defn}
With these terms defined, we can then compute the spectral norms of each of the Liouvillian terms.  These norms are needed to compute the bounds on the Trotter errors which dominate the scaling given in Theorem~\ref{thm:complexity_liouvillian}.

\begin{lem}[Spectral norm of Liouvillian terms]
    The spectral norm of the Liouvillian terms from Definition \ref{def:terms} can be upper bounded as follows:
     \begin{align}
        \norm{K_{n,j}^{(NVE)}} &\leq \frac{p_{\text{max}}}{m_{\text{min}}} \frac{2 \lb \ln{d_x} + 1 \rb}{h_x}  \label{eq:Knve_bd}\\
        \norm{K_{n,j}^{(NVT)}} &\leq \frac{p_{\text{max}}}{m_{\text{min}} s_{\text{min}}^2} \frac{2 \lb \ln{d_x} + 1 \rb}{h_x} \label{eq:Knvt_bd}\\
        \norm{V_{n,n',j}^{\tcla}} &\leq \frac{2 Z_{\text{max}}^2 x_{\text{max}}}{\Delta^3} \frac{2 \lb \ln{d_p} + 1 \rb}{h_p} \label{eq:Vclass} \\
        \norm{V_{n,j}^{\tel}} &\leq \lambda \frac{2 (\ln{d_e+1)}}{h_x h_p} \\
        \norm{K^{\text{bath}}} &\leq \frac{p_{s, \text{max}}}{Q} \frac{2 \lb \ln{d_s} + 1 \rb}{h_s}  \\
        \norm{V^{\text{bath}}_{n,j}} &\leq \frac{2 p_{\text{max}}^2}{m_{\text{min}} s_{\text{min}}^3} \frac{2 \lb \ln{d_{p_s}} + 1 \rb}{h_{p_s}} \\
        \norm{V^{\text{bath}}_{T}} &\leq \frac{N_f k_B T}{s_{\text{min}}} \frac{2 \lb \ln{d_{p_s}} + 1 \rb}{h_{p_s}}.
    \end{align}
\label{lem:bounds_liouvillian_terms}
\end{lem}
\begin{proof}
    First, recall that for any two matrices $A$ and $B$ it holds that $\norm{A \otimes B} = \norm{A} \cdot \norm{B}$. Next, note that the above terms are all of the form $A_{\text{diag}} \otimes D$, where $A_{\text{diag}}$ is a diagonal matrix and $D$ is a central finite difference matrix.

    Specifically, consider first the quantity $\|D_{x_n,j}\|$.  We have from Definition~\ref{def:discrete_derivative} that
    \begin{equation}
        \|D_{x_n,j}\| = \frac{1}{h_x} \|\sum_{\hax} \sum_{k = -d}^d c_{d,k} \ket{\hax-k}\!\bra{\hax}\| \le \frac{\sum_{k = -d}^d |c_{d,k}|}{h_x}
    \end{equation}
    We then have from Lemma \ref{lem:bound_FD_coeff} that
    \begin{equation}
        \|D_{x_n,j}\| \le \frac{2 \left( \ln{d} + 1 \right)}{h_x}.\label{eq:Dbd}
    \end{equation}
    As $m_n \ge m_{\min}$ and $|p_{n,j}| \le p_{\max}$ we then have that
\begin{equation}
\left\|\sum_{\hap_{n,j}} D_{x_{n,j}} \otimes \frac{p_{n,j}}{m_n} \ketbra{{\hap_{n,j}}}{{\hap_{n,j}}}\right\|\le \frac{2p_{\max}(\ln d_x +1)}{m_{\min} h_x} 
\end{equation}
This validates the claim in~\eqref{eq:Knve_bd}. 
The claim of~\eqref{eq:Knvt_bd} immediately follows from the same reasoning and the fact that $s\ge s_{\min}$.

The result of~\eqref{eq:Vclass} also follows from the above bound on $\|D\|$ and the fact that 
\begin{align}
\left\|\sum_{\hax_{n}} \sum_{\hax_{n'}} \frac{Z_n Z_{n'}(x_{n,j} - x_{n',j})}{\lb \norm{x_n - x_{n'}}^2 + \Delta^2 \rb^{3/2}} \ketbra{\hax_{n}}{\hax_{n}} \otimes \ketbra{\hax_{n'}}{\hax_{n'}} \right\| \le \frac{Z_{\max}^2 \max |x_{n,j} - x_{n',j}|}{\min \lb \norm{x_n - x_{n'}}^2 + \Delta^2 \rb^{3/2}} \le \frac{2Z_{\max}^2 x_{\max}}{\Delta^3}.\label{eq:Zdiag}
\end{align}
The momentum derivative expression is exactly the same as previous, except that the grid spacing is $h_p$ rather than $h_x$ putting this observation together with~\eqref{eq:Zdiag} yields
\begin{equation}
\norm{V_{n,n',j}^{\tcla}} \leq\frac{2 Z_{\text{max}}^2 x_{\text{max}}}{\Delta^3} \|D_{p,n_j}\|\leq \frac{2 Z_{\text{max}}^2 x_{\text{max}}}{\Delta^3} \frac{2 \lb \ln{d_p} + 1 \rb}{h_p}
\end{equation}

Next note that $D^1$ is defined to be the centered difference formula which has a coefficient sum of $1$. This observation than yields
\begin{align}
\|\Del\| &= \frac{1}{h_x} \|\sum_{k = -d_e}^{d_e} \sum_{(n',j') \neq (n,j)} \sum_{\hax_{n',j'}} \sum_{\hax_{n,j}} c_{d_e,k} E_{\tel}\lb \{ x_{n',j'}\}, x_{n,j} + k h_x \rb \ketbra{\hax_{n',j'}}{\hax_{n',j'}} \otimes \ketbra{\hax_{n,j}}{\hax_{n,j}}\|\nonumber\\
&\le \frac{\max(E_{\rm el}) \sum_k |c_{d_e,k}|}{h_x} \le \frac{\lambda}{h_x} .
\end{align}
Next using the bound of~\eqref{eq:Dbd} with the substitution of $x\rightarrow p$ we find
\begin{equation}
\norm{V_{n,j}^{\tel}} \leq \lambda \frac{2 (\ln{d_e+1})}{h_x h_p}
\end{equation}

The remaining bounds then follow precisely from the above bound techniques for $D$ and noting the minimum values of $s$ and maximum values of $p$.
\end{proof}

\begin{defn}[Commutator spectral norm of the Liouvillian]
    Let $L$ be the discretized Liouvillian in the $NVE$ or $NVT$ ensemble and let $\ell \in \mathbb{N}_+$. Then we define
    \begin{equation}
        \mu'_{NVE}(\ell) := 3N \norm{K_{n,j}^{(NVE)}} + 6N\ell \norm{V_{n,n',j}^{\tcla}} + 3N \norm{V_{n,j}^{\tel}}
    \end{equation}
     and
    \begin{equation}
        \mu'_{NVT} (\ell) := 3N \norm{K_{n,j}^{(NVT)}} + 6N\ell \norm{V_{n,n',j}^{\tcla}} + 3N \norm{V_{n,j}^{\tel}} + \norm{K^{\text{bath}}} + \ell \norm{V^{\text{bath}}_{n,j}} + \norm{V^{\text{bath}}_{T}}
    \end{equation}
\label{def:commutator_spectral_norm}
\end{defn}

\begin{lem}[Upper bound on $\Tilde{\alpha}_{c}(\ell)$ for Liouvillian simulations in the Born-Oppenheimer approximation]
    Let $\Tilde{\alpha}_{c}(\ell)$ be defined as in Lemma \ref{lem:trotter}. Let $L = L_{\tcla} + L_{\tel}$ be the discrete Liouvillian operator for the $NVE$ ensemble (Definition \ref{def:NVE_liouvillian}). Then $\Tilde{\alpha}_{c}(\ell)$ associated with approximating $e^{-i L t}$ with an $\ell$-th order product formula involving $e^{-i L_{\tcla}t}$ and $e^{-i L_{\tel}t}$ is upper bounded as follows:
    \begin{equation}
        \Tilde{\alpha}_{c}^{(NVE)}(\ell) \le  2^{\ell} \lb \mu'_{NVE} \rb^{\ell+1}.
    \end{equation}

    For Liouvillian simulations in the $NVT$ ensemble (Definition \ref{def:NVT_liouvillian}), we have that    
    \begin{equation}
        \Tilde{\alpha}_{c}^{(NVE)}(\ell) \le  2^{\ell} \lb \mu'_{NVT} \rb^{\ell+1}.
    \end{equation}
    
\label{lem:bound_alpha_c}
\end{lem}

\begin{proof}
    We prove Lemma \ref{lem:bound_alpha_c} via induction on $\ell$. The first-order formula with $\ell=1$ constitutes the base case. The only nonzero commutator at this level is $[L_{\tcla}, L_{\tel}]$ since $[L_{\tcla}, L_{\tcla}] = [L_{\tel}, L_{\tel}] = 0$. Note that all $V_{n,n',j}^{\tcla}$ terms, $K^{\text{bath}}$ and $V^{\text{bath}}_T$ commute with all terms $V_{n,j}^{\tel}$ of $L_\tel$. Hence, we only need to upper bound commutators of the following types: $\left[ K_{n,j}^{(NVE/NVT)}, V_{n',j'}^{\tel} \right]$ and $\left[ V^{\text{bath}}_{n,j}, V_{n',j'}^{\tel} \right]$.
    Generally, $\left[ K_{n,j}^{(NVE/NVT)}, V_{n',j'}^{\tel} \right] \neq 0$ since $\Del$ depends on all nuclear position variables $\{ x_n \}$. In particular, $\left[ D_{x_{n,j}}, D_{n',j'}^{\tel} \right] \neq 0$ in general. We upper bound these commutators by the product of the norms of the individual operators, resulting in
    \begin{equation}
        \norm{\left[ K_{n,j}^{(NVE/NVT)}, V_{n',j'}^{\tel} \right]} \leq 2 \norm{K_{n,j}^{(NVE/NVT)}} \norm{ V_{n',j'}^{\tel}}.
    \end{equation}
    There are a total of $3N \times 3N = 9N^2$ commutators of the above form.
    
    On the other hand, $\left[ V^{\text{bath}}_{n,j}, V_{n',j'}^{\tel} \right]$ can only be nonzero if $n=n'$ and $j=j'$. In that case we obtain
    \begin{equation}
        \norm{\left[ V^{\text{bath}}_{n,j}, V_{n,j}^{\tel} \right]} \leq 2 \norm{V^{\text{bath}}_{n,j}} \norm{ V_{n,j}^{\tel}}.
    \end{equation}
    There are a total of $3N$ commutators of the above form.
    For the $NVE$ ensemble we therefore have that for $\ell=1$
    \begin{equation}
    \begin{split}
        \Tilde{\alpha}_{c}^{(NVE)}(1) &\leq 18 N^2 \norm{K_{n,j}^{(NVE)}} \norm{V_{n,j}^{\tel}} \\
        & \leq 2 \lb 3N \norm{K_{n,j}^{(NVE)}} + 6N\ell \norm{V_{n,n',j}^{\tcla}} + 3N \norm{V_{n,j}^{\tel}} \rb^{2}.
    \end{split}
    \label{eq:alpha_NVE_1}
    \end{equation}

    Similarly we find for $\ell=1$ the following for the $NVT$ ensemble:
    \begin{equation}
    \begin{split}
        \Tilde{\alpha}_{c}^{(NVT)}(1) &\leq  18 N^2 \norm{K_{n,j}^{(NVT)}} \norm{V_{n,j}^{\tel}} + 6N \norm{V^{\text{bath}}_{n,j}} \norm{ V_{n,j}^{\tel}} \\
        &\leq 2 \lb 3N \norm{K_{n,j}^{(NVT)}} + 6N\ell \norm{V_{n,n',j}^{\tcla}} + 3N \norm{V_{n,j}^{\tel}} + \norm{K^{\text{bath}}} + \ell \norm{V^{\text{bath}}_{n,j}} + \norm{V^{\text{bath}}_{T}} \rb^{2}.
    \end{split}
    \label{eq:alpha_NVT_1}
    \end{equation}

    This establishes the $\ell = 1$ base case.

    Let us now discuss the induction step for the $NVE$ ensemble. By assumption assume that there exists a value of $\ell\ge 1$ such that 
    \begin{equation}
        \Tilde{\alpha}_{c}^{(NVE)}(\ell) \le  2^{\ell} \lb 3N \norm{K_{n,j}^{(NVE)}} + 6N\ell \norm{V_{n,n',j}^{\tcla}} + 3N \norm{V_{n,j}^{\tel}} \rb^{\ell+1}.
    \label{hypothesis_NVE}
    \end{equation}
    A single summand of the nested commutator of $\Tilde{\alpha}_{c}^{(NVE)}(\ell)$ is then a string $S_{\ell+1}$ of $\ell + 1$ operators $O_a \in \left\{ K_{n,j}^{(NVE)}, V_{n,n',j}^{\tcla}, V_{n,j}^{\tel} \right\}$ where $a \in \{1, 2,\dots, \ell +1 \}$. Importantly, only certain combinations of operators yield ``non-trivial'' strings, i.e.~strings that have a nonzero contribution to the nested commutator. An example of a trivial string would be a string containing the same operator $\ell + 1$ times. An alternating sequence of two noncommuting operators, such as $K_{n,j}^{(NVE)}$ and $V_{n,j}^{\tel}$, would be an example of a non-trivial string. 
    
    For the induction step we now add an $(\ell+2)$-th operator $O_{\ell+2}$ to $S_{\ell+1}$ to construct nontrivial strings of length $\ell+2$ ($S_{\ell+2}$) as needed for an $(\ell+1)$-th order product formula. First, let us try adding a $K_{n,j}^{(NVE)}$ term to some fixed non-trivial string $S_{\ell+1}$. If $S_{\ell+1}$ contains at least one $V_{n,j}^{\tel}$ term then the resulting string $S_{\ell+2}$ will be non-trivial. Given some fixed $V_{n,j}^{\tel}$, there are $3N$ choices for $K_{n,j}^{(NVE)}$ to create a non-trivial string $S_{\ell+2}$ from $S_{\ell+1}$.

    Next, let us try adding a $V_{n,n',j}^{\tcla}$ term to some fixed non-trivial string $S_{\ell+1}$. In the worst case, $S_{\ell+1}$ contains up to $\ell+1$ different $K_{n,j}^{(NVE)}$ terms. Given some fixed $K_{n,j}^{(NVE)}$ there are $6(N-1)$ different $V_{n,n',j}^{\tcla}$ terms that would yield a nonzero commutator. Hence, for any given string $S_{\ell+1}$ there are at most $6N(\ell+1)$ possibilities to create a non-trivial string $S_{\ell+2}$ via addition of a $V_{n,n',j}^{\tcla}$ term.

    Lastly, let us try adding a $V_{n,j}^{\tel}$ term to some fixed non-trivial string $S_{\ell+1}$. If $S_{\ell+1}$ contain at least one $K_{n,j}^{(NVE)}$ term then the resulting string $S_{\ell+1}$ will be non-trivial. Hence, there are $3N$ choices for $V_{n,j}^{\tel}$ to create a non-trivial string $S_{\ell+2}$ from $S_{\ell+1}$.

    Putting everything together, we therefore obtain the following recursion:
    \begin{equation}
        \Tilde{\alpha}_{c}^{(NVE)}(\ell+1) \le 2 \lb 3N \norm{K_{n,j}^{(NVE)}} + 6N(\ell+1) \norm{V_{n,n',j}^{\tcla}} + 3N \norm{V_{n,j}^{\tel}} \rb \Tilde{\alpha}_{c}^{(NVE)}(\ell).
    \end{equation}
    Using the hypothesis (Eq.~(\ref{hypothesis_NVE})), we arrive at
    \begin{equation}
          \Tilde{\alpha}_{c}^{(NVE)}(\ell+1) \le  2^{\ell+1} \lb 3N \norm{K_{n,j}^{(NVE)}} + 6N(\ell+1) \norm{V_{n,n',j}^{\tcla}} + 3N \norm{V_{n,j}^{\tel}} \rb^{\ell+2}
    \end{equation}
    as desired. This demonstrates the inductive step, and our proof then follows trivially by induction using the fact that the base case of $\ell=1$ has already been demonstrated in~\eqref{eq:alpha_NVE_1}.

    Let us now turn to the induction step for the $NVT$ ensemble. We use the same strategy as for the $NVE$ ensemble. As an induction hypothesis, assume there exists $\ell \ge 1$ such that 
    \begin{equation}
        \Tilde{\alpha}_{c}^{(NVT)}(\ell) \le 2^{\ell} \lb 3N \norm{K_{n,j}^{(NVT)}} + 6N\ell \norm{V_{n,n',j}^{\tcla}} + 3N \norm{V_{n,j}^{\tel}} + \norm{K^{\text{bath}}} + \ell \norm{V^{\text{bath}}_{n,j}} + \norm{V^{\text{bath}}_{T}} \rb^{\ell+1}.
    \label{hypothesis_NVT}
    \end{equation}    
    A single summand of the nested commutator of $\Tilde{\alpha}_{c}^{(NVT)}(\ell)$ is now a string $S_{\ell+1}$ of $\ell + 1$ operators $O_a \in \left\{ K_{n,j}^{(NVT)}, V_{n,n',j}^{\tcla}, V_{n,j}^{\tel}, K^{\text{bath}}, V^{\text{bath}}_{n,j}, V^{\text{bath}}_T \right\}$ where $a \in \{1, 2,\dots, \ell +1 \}$.

    For the induction step we now add an $(\ell+2)$-th operator $O_{\ell+2}$ to $S_{\ell+1}$ to construct nontrivial strings of length $\ell+2$ ($S_{\ell+2}$) as needed for an $(\ell+1)$-th order product formula. First, let us try adding a $K_{n,j}^{(NVT)}$ term to some fixed non-trivial string $S_{\ell+1}$. If $S_{\ell+1}$ contains at least one $V_{n,j}^{\tel}$ term then the resulting string $S_{\ell+2}$ will be non-trivial. Given some fixed $V_{n,j}^{\tel}$, there are $3N$ choices for $K_{n,j}^{(NVT)}$ to create a non-trivial string $S_{\ell+2}$ from $S_{\ell+1}$.

    Next, let us try adding a $V_{n,n',j}^{\tcla}$ term to some fixed non-trivial string $S_{\ell+1}$. In the worst case, $S_{\ell+1}$ contains up to $\ell+1$ different $K_{n,j}^{(NVT)}$ terms. Given some fixed $K_{n,j}^{(NVT)}$ there are $6(N-1)$ different $V_{n,n',j}^{\tcla}$ terms that would yield a nonzero commutator. Hence, for any given string $S_{\ell+1}$ there are at most $6N(\ell+1)$ possibilities to create a non-trivial string $S_{\ell+2}$ via addition of a $V_{n,n',j}^{\tcla}$ term.

    Let us now try adding a $V_{n,j}^{\tel}$ term to some fixed non-trivial string $S_{\ell+1}$. If $S_{\ell+1}$ contain at least one $K_{n,j}^{(NVT)}$ term then the resulting string $S_{\ell+1}$ will be non-trivial. Hence, there are $3N$ choices for $V_{n,j}^{\tel}$ to create a non-trivial string $S_{\ell+2}$ from $S_{\ell+1}$.

    The $K^{\text{bath}}$ term does not commute with $K_{n,j}^{(NVT)}$, $V^{\text{bath}}_{n,j}$ or $V^{\text{bath}}_T$. Hence, we can create a non-trivial string $S_{\ell+2}$ by adding $K^{\text{bath}}$ to $S_{\ell+1}$ if $S_{\ell+1}$ contains $K_{n,j}^{(NVT)}$, $V^{\text{bath}}_{n,j}$ or $V^{\text{bath}}_T$.

    Next, let us try adding a $V^{\text{bath}}_{n,j}$ term to some fixed non-trivial string $S_{\ell+1}$. In the worst case, $S_{\ell+1}$ contains up to $\ell+1$ different $K_{n,j}^{(NVT)}$ terms. Hence, for any given string $S_{\ell+1}$ there are at most $\ell+1$ possibilities to create a non-trivial string $S_{\ell+2}$ via addition of a $V^{\text{bath}}_{n,j}$ term.
    
    Lastly, let us try adding the $V^{\text{bath}}_T$ term to some fixed non-trivial string $S_{\ell+1}$. We can create a non-trivial string $S_{\ell+2}$ if $S_{\ell+1}$ contains the $K^{\text{bath}}$ term.

    Putting everything together, we therefore obtain the following recursion:
    \begin{equation}
    \begin{split}
        \Tilde{\alpha}_{c}^{(NVT)}(\ell+1) &\le 2 \Big( 3N \norm{K_{n,j}^{(NVT)}} + 6N(\ell+1) \norm{V_{n,n',j}^{\tcla}} + 3N \norm{V_{n,j}^{\tel}} \\
        & \qquad  + \norm{K^{\text{bath}}} + (\ell+1) \norm{V^{\text{bath}}_{n,j}} + \norm{V^{\text{bath}}_T} \Big) \Tilde{\alpha}_{c}^{(NVT)}(\ell).
    \end{split}
    \end{equation}
    Using the hypothesis (Eq.~(\ref{hypothesis_NVT})), we arrive at
    \begin{equation}
    \begin{split}
          \Tilde{\alpha}_{c}^{(NVT)}(\ell+1) &\le 2^{(\ell+1)} \Big( 3N \norm{K_{n,j}^{(NVT)}} + 6N(\ell+1) \norm{V_{n,n',j}^{\tcla}} + 3N \norm{V_{n,j}^{\tel}} \\
          & \qquad \qquad + \norm{K^{\text{bath}}} + (\ell+1) \norm{V^{\text{bath}}_{n,j}} + \norm{V^{\text{bath}}_{T}} \Big)^{\ell+2}
    \end{split}
    \end{equation} 
    as desired. The bound then immediately follows induction given the $\ell=1$ base case has already been demonstrated in~\eqref{eq:alpha_NVT_1}
\end{proof}

The bounds in Lemma \ref{lem:bound_alpha_c} exploit the commutator structure of $L = L_{\tcla} + L_{\tel}$.
Not taking the commutator structure into account, one obtains the following bounds on $\Tilde{\alpha}_{c}^{(NVE/NVT)}(\ell)$ (Lemma 1 of~\cite{Childs2021trotter_error}):
\begin{align*}
      \Tilde{\alpha}_{c}^{(NVE)}(\ell) &\le  2^{\ell} \lb 3N \norm{K_{n,j}^{(NVE)}} + 6N^2 \norm{V_{n,n',j}^{\tcla}} + 3N \norm{V_{n,j}^{\tel}} \rb^{\ell+1} \\
      \Tilde{\alpha}_{c}^{(NVT)}(\ell) &\le 2^{\ell} \lb 3N \norm{K_{n,j}^{(NVT)}} + 6N^2 \norm{V_{n,n',j}^{\tcla}} + 3N \norm{V_{n,j}^{\tel}} + \norm{K^{\text{bath}}} + 3N \norm{V^{\text{bath}}_{n,j}} + \norm{V^{\text{bath}}_{T}} \rb^{\ell+1}.
\label{trivial_bound_alpha_com}
\end{align*}
The main improvement of Lemma \ref{lem:bound_alpha_c} over these bounds lies in the reduction of the coefficients of $\norm{V_{n,n',j}^{\tcla}}$ and $\norm{V^{\text{bath}}_{n,j}}$ from $6N^2$ to $6N\ell$ and from $3N$ to $\ell$, respectively.

While the above results apply to general product formulas, we will only be using $2k$-th order Trotter-Suzuki product formulas for our simulation. Hence, we have $\ell = 2k$ with $k\in \mathbb{N}_+$ in the following discussion.

Before proving Theorem \ref{thm:complexity_liouvillian}, it will also be useful to bound the total evolution time associated with a $2k$-th order product formula. The total evolution time is the sum of the absolute values of the evolution time of each segment for a fixed operator.

\begin{lem}[Total evolution time of a higher-order product formula]
    Let $t \ge 0$ be the desired evolution time of the simulation.
    The total evolution time of a $2k$-th order product formula is 
    \begin{equation}
        T_{2k} \le 5^{k-1} t \in O \lb 5^k t \rb.
    \end{equation}
\label{lem:tot_time}
\end{lem}
\begin{proof}
    Recall the following recursive definition of the $2k$-th order product formula $\mathcal{S}_{2k}(t)$ from Definition \ref{def:2ktrotter}:
    \begin{align}
        \mathcal{S}_{2}(t) &:= e^{L_1 \frac{t}{2}} \cdots e^{L_\Gamma \frac{t}{2}} e^{L_\Gamma \frac{t}{2}} \cdots e^{L_1 \frac{t}{2}}\\
        \mathcal{S}_{2k}(t) &:= \mathcal{S}^2_{2k-2}(u_k t) \mathcal{S}_{2k-2}((1-4u_k)t)\mathcal{S}^2_{2k-2}(u_k t),
    \end{align}
    where
    \begin{equation}
        \frac{1}{3} \le u_k := \frac{1}{\lb 4 - 4^{\frac{1}{2k-1}} \rb} \le \frac{1}{2} \hspace{0.4cm} \forall k \in \mathbb{N}, k \ge 2.
    \end{equation}
    Hence,
    \begin{equation}
        T_{2k}(t) = 4 T_{2k-2}(u_k t) + |T_{2k-2}((1 - 4 u_k )t)| \le 5 T_{2k-2}(t).
    \end{equation}
    Together with the base case, $T_2 = t$, this implies that
    \begin{equation}
        T_{2k}(t) \le 5^{k-1}t.
    \end{equation}
\end{proof}

The number of exponentials for a $2k$-th order product formula with $\Gamma$ summands is given by~\cite{Berry2006trotter}
\begin{equation}
    N_{\rm{exp}} = 2(\Gamma-1)\, 5^{k - 1} + 1.
\end{equation}
In our case, $\Gamma = 2$ ($L_{\tcla}$ and $L_{\tel}$), so $N_{\rm{exp}} = 2 \cdot 5^{k - 1} + 1 \in O(5^{k})$.

Now we are ready to prove Lemma \ref{lem:query_complexity_eL}.

\begin{proof}[Proof of Lemma~\ref{lem:query_complexity_eL}]
    As explained earlier, the discretized Liouvillian $L$ is split into a classical part, $L_{\tcla}$ (Definition \ref{def:class_liouvillian}), and an electronic part, $L_{\tel}$ (Definition \ref{def:el_liouvillian}). We then use a $2k$-th order product formula to recombine the two parts. The time evolution is divided into $r$ time steps, resulting in a total number of $O\lb r 5^k \rb$ exponentials with $r$ chosen according to Lemma \ref{lem:trotter}.  Each exponential is then implemented using qubitization~\cite{Low2019hamiltonian}. 
    By the triangle inequality, we can achieve overall simulation error $\leq \epsilon$ if the error of a single exponential is in $O \lb \epsilon/(r 5^k) \rb$. Furthermore, recall that each exponential is simulated using a QSVT-based method, meaning that each exponential comes with a certain failure probability.
    Invoking the union bound, we can ensure an overall success probability of at least $1-\xi$ if the failure probability of a single exponential is in $ O \lb \xi/(r 5^k) \rb$. This can be achieved via amplitude amplification at the expense of a multiplicative factor of $\log \lb r 5^k/ \xi \rb$ to the query complexities.
    The evolution under the classical Liouvillian $L_{\tcla}$ is simulated using qubitization and the QSVT (Lemma \ref{lem:rob_block-Ham_sim}) with total evolution time as in Lemma \ref{lem:tot_time}, resulting in  
    \begin{equation}
            O \lb \log{\lb \frac{r 5^k}{\xi} \rb} \lb \alpha 5^k t + r 5^k \log{\lb \frac{r 5^k}{\epsilon} \rb} \rb \rb.
    \end{equation}
    queries to the block-encoding of $L_{\tcla}$, where $\alpha = \alpha_{NVE}$ for simulations in the $NVE$ ensemble (Lemma \ref{lem:bounds_L_class_NVE}) and $\alpha = \alpha_{NVT}$ for simulations in the $NVT$ ensemble (Lemma \ref{lem:bounds_L_class_NVT}).

    From Lemma \ref{lem:trotter} we have that $r \in O \left(\frac{\Tilde{\alpha}_{c}^{1/2k} t^{1 + 1/2k}}{\epsilon^{1/2k}}\right)$ and from Lemma \ref{lem:bound_alpha_c} we have that
    $\Tilde{\alpha}_{c}(2k) \leq 2^{2k}(\mu')^{2k+1}$ where $\mu' \in \{ \mu'_{NVE}, \mu'_{NVT} \}$.
    The number of queries to a block-encoding of $L_{\tcla}$ is then in
    \begin{equation}
        \widetilde{O} \lb 5^k t \lb \alpha \log{\lb \frac{\mu'}{\epsilon \xi} \rb} + \mu' \lb \frac{\mu' t}{\epsilon} \rb^{1/2k} \log{\lb \frac{1}{\xi} \rb} \rb \rb.
    \end{equation}
    
    The evolution under the electronic Liouvillian $L_{\tel}$ is simulated according to Lemma \ref{lem:complexity_eL_el} with total evolution time as in Lemma \ref{lem:tot_time}, resulting in 
    \begin{equation}
    \begin{split}
        & O \Bigg( 5^k N d_e \log{\lb \frac{r 5^k N d_e}{\xi} \rb} \lb \frac{\lambda t}{h_x h_p} + r \log{\lb \frac{r 5^k N \lambda \ln{(d_e)} t}{h_x h_p \epsilon} \rb} \log{\lb \frac{r 5^k N d_e \log{\lb \frac{r 5^k N \lambda \ln{(d_e)} t}{h_x h_p \epsilon} \rb}}{\epsilon} \rb} \rb \\
        & \quad + r 5^k N d_e  \frac{\lambda}{\gamma \, \delta} \log{ \lb \frac{r 5^k N d_e}{ \delta \epsilon} \rb} \log{\lb \frac{r 5^k N d_e}{\xi} \rb} \Bigg)
    \end{split}
    \end{equation}
    queries to a block-encoding of $H_{\tel}$ where $d_e \in O \lb \frac{\log{\left( \frac{N r 5^k \chi t}{h_p \epsilon}\right)}}{\log{\left( \frac{1}{h_x} \right)}} \rb$.
    Furthermore,
    \begin{equation}
        O  \lb \frac{r 5^k N d_e}{\delta} \log{\lb \frac{r 5^k N d_e}{\xi} \rb} \rb
    \end{equation}
    queries to the state preparation oracle $U_I$ are needed.\\
    Using 
    \begin{align}
        &r \in O \lb \frac{\Tilde{\alpha}_{c}^{1/2k} t^{1 + 1/2k}}{\epsilon^{1/2k}}\rb, \\
        &\Tilde{\alpha}_{c}(2k) \leq 2^{2k}(\mu')^{2k+1}, \\
        &\frac{\lambda}{h_x h_p} \leq \mu',
    \end{align}
    we find that
    \begin{equation}
    \begin{split}
        &\widetilde{O} \lb 5^k N d_e t \lb \frac{\lambda}{h_x h_p} \log{\lb \frac{\Tilde{\alpha}_{c}}{\epsilon \xi} \rb} + \lb \frac{\Tilde{\alpha}_{c} t}{\epsilon} \rb^{1/(2k)} \lb \log{\lb \frac{\lambda}{h_x h_p} \rb} + \frac{\lambda}{\gamma \delta} \rb \log{\lb \frac{1}{\xi} \rb} \rb \rb \\
        &\subseteq \widetilde{O} \lb 5^k N d_e \frac{\lambda}{\gamma \delta} \frac{\lb \mu' t \rb^{1+1/(2k)}}{\epsilon^{1/(2k)}} \log{\lb \frac{1}{\xi} \rb} \rb
    \end{split}
    \end{equation}
    queries to the block-encoding of $H_{\tel}$ and
    \begin{equation}
        \widetilde{O} \lb \frac{5^k N d_e}{\delta} \frac{\lb \mu' t \rb^{1+1/(2k)}}{\epsilon^{1/(2k)}} \log{\lb \frac{1}{\xi} \rb} \rb
    \end{equation}
    queries to the state preparation oracle $U_I$ are sufficient.
\end{proof}

\subsection{Proof of Theorem \ref{thm:complexity_liouvillian}}

The previous results now give us the tools that we need to prove Theorem \ref{thm:complexity_liouvillian}, which provides upper bounds on the Toffoli complexity of simulating Liouvillian dynamics.
We restate it here for convenience.

\liouvillian*

\begin{proof}
    The upper bound on the number of queries to $\widetilde{U}_I$ follows from Lemma~\ref{lem:query_complexity_eL}. Note that Lemma~\ref{lem:query_complexity_eL} deals with errors between operators, i.e.~the complexity bounds hold for the worst-case input state. Problem~\ref{prob:sim} on the other hand is formulated in terms of the simulation error for a fixed input state. This means that here we only need a good initial electronic state for grid points associated with a nonzero amplitude at some point during the simulation because any simulation errors that occur on grid points which are associated with zero amplitude throughout the simulation do not contribute to the error of the final state. Hence, we can use $\widetilde{U}_I$ instead of $U_I$.

    Realizing that $\mu' \leq \mu$ and choosing
    \begin{equation}
        k = \text{round} \left[ \sqrt{\frac{1}{2} \log_5\lb \mu t /\epsilon \rb + 1}\right],
    \label{choice_k}
    \end{equation}
    we then find that the query complexity for $\widetilde{U}_I$ is in
     \begin{equation}
        \widetilde{O} \lb 5^{\sqrt{2 \log_5\lb \mu t /\epsilon \rb}} \frac{N d_e \mu t}{\widetilde{\delta}} \log{\lb \frac{1}{\xi} \rb} \rb \subseteq \widetilde{O} \lb \frac{N d_e \mu^{1+o(1)} t^{1+o(1)}}{\widetilde{\delta} \, \epsilon^{o(1)}} \log \lb \frac{1}{\xi} \rb \rb,
    \end{equation}
    which is subpolynomial in $1/\epsilon$ and almost linear in $\mu$ and $t$.
    The above heuristic choice for $k$ can be obtained by balancing the factor $5^k$ with $\lb \frac{\mu t}{\epsilon} \rb^{1/2k}$ to minimize the overall complexity with respect to $k$~\cite{Berry2006trotter}.
    
    Next, let us discuss the overall Toffoli complexity which follows from multiplying the query complexities of the block-encodings with their respective Toffoli complexity. More specifically, from Lemma \ref{lem:query_complexity_eL} we have that
    \begin{equation}
    \begin{split}
        &\widetilde{O} \lb 5^k t \lb \alpha \log{\lb \frac{\mu'}{\epsilon \xi} \rb} + \mu' \lb \frac{\mu' t}{\epsilon} \rb^{1/2k} \log{\lb \frac{1}{\xi} \rb} \rb \rb \\
        &\subseteq \widetilde{O} \lb 5^k \frac{\mu^{1+1/2k} t^{1+1/2k}}{\epsilon^{1/2k}} \log{\lb \frac{1}{\xi} \rb} \rb 
    \end{split}
    \end{equation}
    queries to an $\epsilon/(5^k t)$-precise block-encoding of the classical Liouvillian $L_{\tcla}$ are sufficient. Lemmas \ref{lem:bounds_L_class_NVE} and \ref{lem:bounds_L_class_NVT} imply that such a block-encoding requires
    \begin{equation}
    \begin{split}
        &\widetilde{O} \lb N \log \lb \frac{g \alpha 5^k t}{\epsilon} \rb + \log^{\log 3}{\lb \frac{\alpha 5^k t}{\epsilon} \rb}  + d \log (g) \rb \\
        &\subseteq \widetilde{O} \lb (N+d) \log^{\log 3} \lb \frac{\mu 5^k t}{\epsilon} \rb  \rb
    \end{split}
    \end{equation}
    Toffoli gates, where $d$ is the maximum order of the finite difference schemes.
    Choosing $k$ as in Eq.~\eqref{choice_k}, we find that the Toffoli complexity associated with simulating the classical Liouvillian is in
    \begin{equation}
        \widetilde{O} \lb (N + d) \frac{\mu^{1+o(1)} t^{1+o(1)}}{\epsilon^{o(1)}} \log{\lb \frac{1}{\xi} \rb} \rb.
    \end{equation}
    According to Lemma \ref{lem:query_complexity_eL} we also need
    \begin{equation}
    \begin{split}
        &\widetilde{O} \lb 5^k N d_e \frac{\lambda}{\gamma \delta} \frac{\lb \mu' t \rb^{1+1/(2k)}}{\epsilon^{1/(2k)}} \log{\lb \frac{1}{\xi} \rb} \rb \\
        &\subseteq \widetilde{O} \lb  5^k d_e \frac{\mu^{2+1/2k} t^{1+1/2k}}{\gamma \, \delta \, \epsilon^{1/2k}} \log{\lb \frac{1}{\xi} \rb} \rb \\
        &\subseteq \widetilde{O} \lb  d_e \frac{\mu^{2+ o(1)} t^{1+o(1)}}{\gamma \, \delta \, \epsilon^{o(1)}} \log{\lb \frac{1}{\xi} \rb} \rb
    \end{split}
    \end{equation}
    queries to an $\frac{h_x h_p \epsilon}{5^k 36 N d_e t}$-precise block-encoding of the electronic Hamiltonian $H_{\tel}$. Lemma \ref{lem:block-encode_Hel} implies that such a block-encoding has Toffoli cost in
    \begin{equation}
        O \lb N + \Tilde{N} + \log{\lb \frac{B 5^k 36 N d_e t}{h_x h_p \epsilon}\rb} \rb.
    \end{equation}
    The Toffoli cost associated with simulating the electronic Liouvillian is then in
    \begin{equation}
        \widetilde{O} \lb  N_{\text{tot}} d_e \frac{\mu^{2+ o(1)} t^{1+o(1)}}{\gamma \, \delta \, \epsilon^{o(1)}} \log{\lb \frac{1}{\xi} \rb} \rb.
    \end{equation}
    Combining all results, we find that the overall Toffoli complexity of simulating $e^{-iLt}$ is in
    \begin{equation}
        \widetilde{O} \lb \frac{N_{\text{tot}} d \mu^{2+o(1)} t^{1+o(1)}}{\gamma \, \delta \, \epsilon^{o(1)}} \log \lb \frac{1}{\xi} \rb \rb.
    \end{equation}
    Note that the above statements are independent of the initial quantum state encoding the initial phase space density. In particular, $\gamma$ is a lower bound on the spectral gap of the block-encoded operator $\widetilde{H}_{\tel}\lb \{x_{n}\} \rb$ over all phase space grid points. Similarly, $\delta$ is a lower bound on the overlap of the initial electronic state with the true electronic ground state over all phase space grid points.
    However, if we are dealing with a fixed initial state as in Problem \ref{prob:sim}, we only need to consider the spectral gap and the electronic ground state of the grid points associated with a nonzero amplitude at some point during the simulation because any simulation errors that occur on grid points which are associated with zero amplitude throughout the simulation do not contribute to the error of the final state. Let $\widetilde{\gamma} \geq \gamma$ be a lower bound on the spectral gap of $\widetilde{H}_{\tel}\lb \{x_{n}\} \rb$ over that subset of grid points.
    Likewise, let $\widetilde{\delta} \geq \delta$ be a lower bound on the overlap of the initial electronic state with the true electronic ground state over the same subset of grid points.
    Then Problem \ref{prob:sim} can be solved using only $O \lb \frac{1}{\widetilde{\gamma} \widetilde{\delta}} \rb$ rather than $O \lb \frac{1}{\gamma \delta} \rb$ Toffoli gates.
\end{proof}

\section{Details on estimating the free energy}
\label{app:free}

Recall from Definition \ref{def:free_energy} that the free energy is given by 
\begin{equation}
    \mathcal{F} = \mathcal{U} - T \mathcal{S}_G
\end{equation}
where $\mathcal{U}$ is the internal energy of the system and $\mathcal{S}_G$ is the Gibbs entropy of the system.
The concept of Gibbs entropy can be extended to the quantum world in the form of the von Neumann entropy:
\begin{defn}[Von Neumann entropy]
    Let $\rho \in \mathbb{C}^{\eta \times \eta}$ be a density matrix. Then
    \begin{equation}
        S_N := -\text{Tr} \lb \rho \ln  \rho \rb
    \end{equation}
    is the von Neumann entropy associated with $\rho$.
\end{defn}
Note that $S_N = 0$ for a pure state. The idea is to estimate the Gibbs entropy and the internal energy of our system separately and add the results to estimate the free energy. This means that our algorithm requires at least two separate simulations.
Before we show how to obtain the Gibbs entropy from a quantum algorithm for estimating the von Neumann entropy, let us explain the usage of the Nosé thermostat within the Liouvillian framework in a little more detail. The main difference to plain Liouvillian dynamics in the $NVE$ ensemble is the new variable for the heat bath, $s$, and its associated momentum variable, $p_s$. Furthermore, the momenta $\{ p'_n \}$ appearing in the $NVT$ Hamiltonian $H_{NVT}$ from Definition \ref{H_NVT} are virtual momenta of the extended system. They are related to the real momenta $\{ p_n \}$ of the physical system via the relation $p_n = \frac{p'_n}{s}$. In the discretized setting, we introduce a cutoff $s_{\min}$ to avoid infinities in the simulation.

In the following, we will drop the particle index of the position and momentum variables for ease of notation.
As mentioned in the main text, for continuous variables, it can be shown that the microcanonical partition function $\mathcal{Z}$ of the extended system gives rise to the canonical partition function when restricted to the real system~\cite{Nose1984partition_function, Huenenberger2005thermostat}. More specifically, it can be proven, via a change of variables, that
\begin{equation}
\begin{split}
    \mathcal{Z} &\propto \int d\{x\} \int d\{p'\} \int ds \int dp_s \, \delta \lb H_{NVT} \lb \{x\}, \{p'\}, s, p_s \rb - E_{\text{ext}}\rb \\
    &\propto  \int d\{x\} \int d\{p\} e^{- H_{NVE} \lb \{x\}, \{ p \} \rb/(k_BT)},
\end{split}
\end{equation}
where $E_{\text{ext}}$ is the conserved energy of the extended system and $H_{NVE}$ is the $NVE$ Hamiltonian from Eq.~\eqref{H_NVE}.

\begin{algorithm}[t!]
    \caption{Free energy estimation}
    \label{alg:free_energy}
    \KwIn{Quantum state $\ket{\psi_0} = \sum_{\hax, \hap', \overline{S}} c_{x, p', S}(0)\ket{\{\hax\}}\ket{\{\hap'\}} \ket{\overline{S}}$ encoding initial KvN wavefunction of the system together with the heat bath. \newline
    Further input parameters: $t, \epsilon, \xi, k, N, \tn, \{m_n\}_{n=1}^N, \{Z_n\}_{n=1}^N, x_{\text{max}}, p'_{\text{max}}, h_x, h_{p'}, d_x, d_p, d_e, \Delta, B, h_{\tel}, \delta, \widetilde{\gamma}, \chi, N_f, T, Q, s_{\text{min}}, h_s, h_{p_s}, d_s, d_{p_s}$.}
    \KwOut{With success probability $\geq 1 - \xi$ an $\epsilon$-precise estimate of the free energy associated with the phase space density after time $t$.}
    \begin{enumerate}[leftmargin=*]
        \item  Apply the $NVT$ Liouvillian simulation algorithm as summarized in Algorithm~\ref{alg:liouvillian} to $\ket{\psi_0}$ to obtain
        \begin{equation*}
            \ket{\psi_t} = U_{L_{NVT}} \ket{\psi_0} = \sum_{\hax, \hap', \overline{S}} {c}_{x, p', S}(t)\ket{\{\hax\}}\ket{\{\hap'\}} \ket{\overline{S}}.
        \end{equation*}
        \item Duplicate the $\ket{s}$ register of the bath using $U_{\text{dup}}$ to retain the information of $s$. This yields $\ket{\Psi_t}$ which is a purification of the density matrix $\rho_\text{sys} (t)$ for which we want to estimate the von Neumann entropy\;
        \item Eliminate the off-diagonal elements of $\rho_\text{sys} (t)$ via controlled phase gradients between an ancillary register $\ket{j}$ and $\ket{\Psi_t}$\;
        \item Tracing out the ancillary register $\ket{j}$ and the bath register $\ket{S}$ but not the duplicated $\ket{s}$ register, we obtain a reduced phase space density over the nuclear position and momentum registers $\ket{\{\hax_n\}} \ket{\{\hap_n\}}$\;
        \item Use the algorithm associated with Theorem 13 of \cite{Gilyen2019dist_prop_test} to estimate the Gibbs entropy of $\rho_{sys} (t)$ within error $\epsilon/(2k_BT)$\;
        \item Apply the Hadamard test to each of the three internal energy components $H_{\text{kin}}$, $H_{\text{pot}}$ and $H_{E_{\tel}}$ to estimate their expectation values within error $\epsilon/6$\;
        \item Classically add the estimates of the Gibbs entropy and the internal energy components to obtain an $\epsilon$-precise estimate of the free energy $\mathcal{F}$\;
    \end{enumerate}
\end{algorithm}

In terms of our quantum algorithm, we now have three types of quantum registers representing classical variables: the nuclear positions register $\ket{\{\hax\}}$, the nuclear (virtual) momentum register $\ket{\{\hap'\}}$ and the bath register 
\begin{equation}
    \ket{\overline{S}} := \ket{\has}\ket{\hap_s}.
\end{equation}
Let 
\begin{equation}
    \ket{\psi_0} := \sum_{\hax, \hap', \overline{S}} c_{x, p', S}(0)\ket{\{\hax\}}\ket{\{\hap'\}} \ket{\overline{S}}
\end{equation}
be a quantum state encoding the initial KvN wave function of the system + bath where the $\{c_{x, p', S}(0)\}$ are complex amplitudes. 
We time evolve $\ket{\psi_0}$ according to the $NVT$ Liouvillian from Definition \ref{def:NVT_liouvillian}, resulting in
\begin{equation}
   \ket{\psi_t} := U_{L_{NVT}} \ket{\psi_0} = \sum_{\hax, \hap', \overline{S}} {c}_{x, p', S}(t)\ket{\{\hax\}}\ket{\{\hap'\}} \ket{\overline{S}},
\end{equation}
where
\begin{equation}
    U_{L_{NVT}} := e^{-iL_{NVT}t}
\end{equation}
is the unitary that implements the Liouvillian time evolution of the system + bath. The discrete analog of integrating out the bath variables as done in Eq.~(\ref{partition_func}) would be to trace out the bath register $\ket{\overline{S}}$. However, at this stage, we cannot simply trace out $\ket{\overline{S}}$ since we would loose all information of $\ket{\has}$, which is needed to compute the real momenta $\{ p \}$. In other words, we first need to perform a discrete analog of the change of variables $p' \rightarrow p'/s = p$.
We do so by duplicating the $\ket{\has}$ register via a unitary $U_{\text{dup}}$ to obtain
\begin{equation}
     \ket{\Psi_t} := U_{\text{dup}} \ket{\psi_t}\ket{0} =
     \sum_{\hax, \hap', \overline{S}} {c}_{x, p', S}(t)\ket{\{\hax\}}\ket{\{\hap'\}} \ket{\overline{S}} \ket{\has} = \sum_{\hax, \hap', \has, \hap_s} {c}_{x, p', s, p_s}(t)\ket{\{\hax\}}\ket{\{\hap'\}} \ket{\has} \ket{\hap_s} \ket{\has}.
\end{equation}
Note that $U_{\text{dup}}$ can be implemented using $O \lb \log g_s \rb$ $\cnot$ gates. More specifically, we apply a single $\cnot$ to each qubit of the $\ket{\has}$ register where each $\cnot$ has a different target qubit in the duplication ancilla register.

The above quantum state $\ket{\Psi_t}$ can be regarded as a purification of the following density matrix, which describes the dynamics of the nuclei under the influence of the heat bath:
\begin{equation}
\begin{split}
    \rho_{\text{sys}}(t) := \text{Tr}_S \lb \ket{\Psi_t}\!\bra{\Psi_t} \rb &= \sum_{\substack{\hax, \hap', \overline{S}\\ \hax' ,\hap''}} c_{x, p', S}(t)c_{x', p'', S}^*(t) \ket{\{\hax\}} \ket{\{\hap'\}, s}\!\bra{\{\hax'\}} \bra{\{\hap''\}, \has} \\
    &= \sum_{\substack{\hax, \hap', \has, \hap_s\\ \hax' , \hap''}} c_{x, p', s, p_s}(t)c_{x', p'', s, p_s}^*(t) \ket{\{\hax\}}\ket{\{\hap'\}, \has}\!\bra{\{\hax' \}}\bra{\{ \hap''\}, \has}.
\end{split}
\end{equation}
Note that the combined register $\ket{\{\hap'\}, \has}$ can be regarded as the real momentum register. It is effectively just a different representation of $\ket{\{\hap'/\has\}}$. The dimension of $\rho_{\text{sys}}(t)$ is 
\begin{equation}
    \eta = g_x^{3N} g_{p'}^{3N} g_s
\end{equation}
and the probability of finding the nuclei in a particular configuration $\ket{\{\hax^*, \hap'^*, \has^*\}}$ is given by
\begin{equation}
    \bra{\{\hax^*, \hap'^*, \has^*\}} \rho_{\text{sys}}(t)\ket{\{\hax^*, \hap'^*, \has^*\}} = \sum_{\overline{S}} |c_{x^*, p'^*, S}(t)|^2 =: b_{x^*, p'^*, s^*}(t).
\label{diag_elements}
\end{equation}

Using the above ideas, we can reduce the problem of estimating the Gibbs entropy associated with $\rho_{\text{sys}}(t)$ to the problem of estimating the von Neumann entropy of a modified density matrix. 
The reason for requiring a modified density matrix is that in contrast to the von Neumann entropy, the Gibbs entropy associated with $\rho_{\text{sys}}(t)$ depends only on the diagonal elements of $\rho_{\text{sys}}(t)$ since these represent the classical probabilities of the different microstates, see Eq.~(\ref{diag_elements}).
We can eliminate the off-diagonal elements by applying controlled phase gradients to the purification $\ket{\Psi_t}$ as we will now explain in more detail.

Let $\ket{j}$ denote a computational basis state of a $\log \lb \eta_v \rb$-qubit ancilla register where
\begin{equation}
    \eta_v := g_x^{3N}g_{p'}^{3N}.
\end{equation}
Furthermore, let $\ket{n}$ denote a $\eta_v$-dimensional computational basis state obtained by considering all the nuclear position and virtual momentum variables as a single register, i.e.,
\begin{equation}
    \ket{n} \equiv \ket{\{\hax\}}\ket{\{\hap'\}}.
\end{equation}
This change in perspective simplifies the implementation of the controlled phase gradients~\cite{Gidney2018phase_gradient}.
Preparing a uniform superposition over the $\ket{j}$ register and applying controlled phase gradients then yields
\begin{equation}
\begin{split}
    \frac{1}{\sqrt{\eta_v}} \sum_j \ket{\Psi_t} \ket{j} &= \frac{1}{\sqrt{\eta_v}} \sum_{n,\overline{S},j} c_{n, S}(t) \ket{n}\ket{\overline{S}}\ket{\has}\ket{j} \\ 
    &\xrightarrow{U_{\text{pg}}} \frac{1}{\sqrt{\eta_v}} \sum_{n,\overline{S},j} c_{n, S}(t) e^{-2 \pi i n \frac{j}{\eta_v}} \ket{n}\ket{\overline{S}}\ket{\has}\ket{j} =: \ket{\Psi'_t} =: U_{L_{NVT}}' \ket{\psi_0}\ket{0} \ket{0},
\end{split}
\end{equation}
where $U_{\text{pg}}$ denotes the controlled phase gradient unitary and
\begin{equation}
     U_{L_{NVT}}' := U_{\text{pg}} \cdot \lb U_{\text{dup}} \otimes \mathbb{1} \rb \cdot \lb U_{L_{NVT}} \otimes \mathbb{1} \otimes \mathbb{1} \rb \in \mathbb{C}^{\eta_{\text{pur}} \times \eta_{\text{pur}}}
\label{purification}
\end{equation}
with 
\begin{equation}
    \eta_{\text{pur}}:= g_x^{6N} g_{p'}^{6N} g_s^2 g_{p_s}.
\end{equation}

Let us now check that this gives the correct density matrix after tracing out the $\ket{\overline{S}}$ and $\ket{j}$ registers.
First, note that the purification state $\ket{\Psi'_t}$ can be written in density matrix notation as
\begin{equation}
    \ket{\Psi'_t} \! \bra{\Psi'_t} = \frac{1}{\eta_v} \sum_{\substack{n,\overline{S},j \\ n', \overline{S}', j'}} c_{n, S}(t) c_{n', S'}^*(t) e^{-2 \pi i \frac{nj - n'j'}{\eta_v}} \ket{n}\ket{\overline{S}}\ket{\has}\ket{j} \bra{n'}\bra{\overline{S}'}\bra{\has'}\bra{j'}.
\end{equation}
Tracing out the $\ket{j}$ register results in
\begin{equation}
\begin{split}
    \text{Tr}_j \lb \ket{\Psi'_t}\!\bra{\Psi'_t} \rb &=
     \sum_{\substack{n,\overline{S} \\ n', \overline{S}'}} c_{n, S}(t) c_{n', S'}^*(t) \lb \frac{1}{\eta_v} \sum_j e^{-2 \pi i j\frac{n - n'}{\eta_v}} \rb \ket{n}\ket{\overline{S}}\ket{\has} \bra{n'}\bra{\overline{S}'}\bra{\has'} \\
     &= \sum_{\substack{n,\overline{S} \\ n', \overline{S}'}} c_{n, S}(t) c_{n', S'}^*(t) \delta_{n,n'} \ket{n}\ket{\overline{S}}\ket{\has} \bra{n'}\bra{\overline{S}'}\bra{\has'} \\
     &= \sum_{n, \overline{S}, \overline{S}'} c_{n, S}(t) c_{n, S'}^*(t) \ket{n}\ket{\overline{S}}\ket{\has} \bra{n}\bra{\overline{S}'}\bra{\has'}
\end{split}
\end{equation}
and tracing out the bath register $\ket{\overline{S}}$ yields
\begin{equation}
    \rho_{\text{sys}}'(t) := \sum_{n, \overline{S}} |c_{n, S}(t)|^2 \ket{n}\ket{\has}\bra{n}\bra{\has} = \sum_{\hax,\hap',\has,\hap_s} |c_{x,p',s,p_s}(t)|^2 \ket{\{\hax\}}\ket{\{\hap'\}, \has}\!\bra{\{\hax\}}\bra{\{\hap'\}, \has},
\label{diag_rho_sys}
\end{equation}
which consists only of the diagonal elements of $\rho_{\text{sys}}(t)$ as desired. Note that an alternative method of eliminating the off-diagonal elements of $\rho_{\text{sys}}(t)$ consists of duplicating the entire nuclear register $\ket{\{\hax\}}\ket{\{\hap'\}}$ and then tracing out the duplicated register similar to what was done with the $\ket{\has}$ register earlier on. In terms of complexity, this duplication approach is essentially equivalent to the phase gradient approach.
Algorithm~\ref{alg:free_energy} summarizes the key steps of our free energy estimation protocol.

\subsection{Estimating the Gibbs entropy}

Let us now explain how to estimate the Gibbs entropy of our system. First, note that, in general, we cannot implement the Liouvillian evolution operator $U_{L_{NVT}}$ exactly. However, Theorem \ref{thm:complexity_liouvillian} shows that we can efficiently construct an approximation $\widetilde{U}_{L_{NVT}}$ such that 
\begin{equation}
    \norm{\widetilde{U}_{L_{NVT}} - U_{L_{NVT}}} \leq \epsilon
\end{equation}
for any $\epsilon \in (0,1)$. Let us assume that $U_{\text{pg}}$ and $U_{\text{dup}}$ can be implemented with negligible error. Then the resulting approximation $\widetilde{U}_{L_{NVT}}'$ of $U_{L_{NVT}}'$ satisfies 
\begin{equation}
    \norm{\widetilde{U}_{L_{NVT}}' - U_{L_{NVT}}'} \leq \epsilon
\end{equation}
We denote the corresponding approximate density matrix of the system $\widetilde{\rho}_{\text{sys}}'(t)$. The following inequality will be useful for upper bounding the difference in the von Neumann entropy associated with $\rho_{\text{sys}}'(t)$ and $\widetilde{\rho}_{\text{sys}}'(t)$ in terms of their trace distance.

\begin{defn}[Trace distance]
    Let $\rho, \sigma \in \mathbb{C}^{\eta \times \eta}$ be density matrices. Then their trace distance is given by
    \begin{equation}
        \mathcal{T}(\rho, \sigma) := \frac{1}{2} \norm{\rho - \sigma}_1 = \frac{1}{2} {\rm{Tr}} \lb \sqrt{\lb \rho - \sigma \rb^{\dagger} \lb \rho - \sigma \rb} \rb = \frac{1}{2} \sum_{j=1}^{\eta} \left| \lambda_j \right|,
    \end{equation}
    where $\lambda_j \in \mathbb{R}$ is the $j$-th eigenvalue of $\rho - \sigma$.
\label{def:trace_dist}
\end{defn}

In the following, we will use $\norm{\cdot}_1$ to refer to the trace norm (i.e.~the Schatten 1-norm). As before, $\norm{\cdot}$ denotes the (induced) 2-norm.

\begin{lem}[Fannes inequality~\cite{Fannes1973entropy, Audenaert2007entropy}]
    Let $\rho$ and $\sigma$ be $\eta$-dimensional density matrices. If $\mathcal{T}(\rho, \sigma) \leq 1/(2e)$ then
    \begin{equation}
        \left| S_N \lb \rho \rb - S_N \lb \sigma \rb \right| \leq 2\mathcal{T} \log_2\lb \eta \rb - 2\mathcal{T} \log_2\lb 2\mathcal{T} \rb.
    \end{equation}
\label{lem:fannes}
\end{lem}

We now show how to estimate the Gibbs entropy $\mathcal{S}_G$ of our system.
\begin{lem}[Estimation of the Gibbs entropy]
    Let $\epsilon \in (0,1)$ and let $\rho_{\text{sys}}'(t)$ be the $\eta$-dimensional diagonal density matrix of the system as defined in Eq.~(\ref{diag_rho_sys}) where $\eta\ge 6$.
    Let $\eta_{\text{pur}}$ be the dimension of $U_{L_{NVT}}'$ as defined in Eq.~(\ref{purification}).
    Furthermore, let $\widetilde{U}_{L_{NVT}}'$ be an $\lb \frac{\epsilon}{4\eta_{\text{pur}} \log \lb \eta/\nu \rb} \rb$-precise approximation to $U_{L_{NVT}}'$ where $\nu \in (0,1)$ is a lower bound on $2\mathcal{T}\lb \widetilde{\rho}_{\text{sys}}'(t), \rho_{\text{sys}}'(t) \rb$ and 
    \begin{equation}
        \epsilon \leq \frac{2 \log \lb \eta/\nu \rb}{e}.
    \end{equation}
    There exists a quantum algorithm that outputs an estimate of the Gibbs entropy associated with $\rho_{\text{sys}}(t)$ within error $\epsilon$ with success probability $\ge 1 - \xi$ using
    \begin{equation}
        \widetilde{O} \lb \frac{\eta}{\epsilon^{1.5}} \log{\lb \frac{1}{\xi} \rb}\rb
    \end{equation}
    queries to $\widetilde{U}_{L_{NVT}}'$.
\label{lem:gibbs_entropy}
\end{lem}

\begin{proof}
    By construction, estimating the Gibbs entropy of $\rho_{\text{sys}}(t)$ is equivalent to estimating the von Neumann entropy of $\rho_{\text{sys}}'(t)$ since $\rho_{\text{sys}}'(t)$ consists only of the diagonal elements of $\rho_{\text{sys}}(t)$. However, we only have access to $\widetilde{U}_{L_{NVT}}'$, an $\epsilon_U$-precise approximation of $U_{L_{NVT}}'$ in $\ell^2$-distance. We will show that $\epsilon_U$ needs to be upper bounded by 
    \begin{equation}
        \epsilon_U \leq \frac{\epsilon}{4\eta_{\text{pur}} \log \lb \eta/\nu \rb}
    \end{equation}
    This gives rise to an approximation $\widetilde{\rho}_{\text{sys}}'(t)$ of $\rho_{\text{sys}}'(t)$.
    The idea then is to use Theorem 13 of~\cite{Gilyen2019dist_prop_test} to obtain an $\epsilon_{\text{est}}$-precise estimate $\widetilde{S}\lb \widetilde{\rho}_{\text{sys}}'(t) \rb$ of the von Neumann entropy of $\widetilde{\rho}_{\text{sys}}'(t)$ where $\epsilon_{\text{est}} \in (0,1)$. 
    The algorithm of~\cite{Gilyen2019dist_prop_test} requires access to a purification of the density matrix, which in our case is simply $\ket{\Psi'_t}$. The work of~\cite{Gilyen2019qsvt} shows that a polynomial of degree $\widetilde{O} \lb \sqrt{\frac{\eta}{\epsilon_\text{est}}}\rb$ to approximate $\log \lb \frac{1}{\widetilde{b}_{x, p', s}(t)}\rb$ within error $\epsilon_{\rm est}$ where $\{\widetilde{b}_{x, p', s}(t)\}$ are the (diagonal) elements of $\widetilde{\rho}_{\text{sys}}'(t)$.
    This implies that quantum amplitude estimation~\cite{brassard2002quantum} can be used to learn $\widetilde{S}\lb \widetilde{\rho}_{\text{sys}}'(t) \rb$ with constant success probability within error $\epsilon_{\text{est}}$ using $\widetilde{O} \lb \frac{\eta}{\epsilon_{\text{est}}^{1.5}}\rb$ queries to $\widetilde{U}_{L_{NVT}}'$. 
    The Chernoff bound implies that we can achieve a success probability $\geq 1- \xi$ with $\log{(1/\xi)}$ repetitions of the algorithm. Next, let us discuss the required block-encoding precision of $U_{L_{NVT}}'$.
    By the triangle inequality we have that
    \begin{equation}
    \begin{split}
        \left| \widetilde{S}\lb \widetilde{\rho}_{\text{sys}}'(t) \rb - S\lb \rho_{\text{sys}}'(t) \rb \right| &\leq \left| \widetilde{S}\lb \widetilde{\rho}_{\text{sys}}'(t) \rb - S\lb \widetilde{\rho}_{\text{sys}}'(t) \rb \right| + \left| S\lb \widetilde{\rho}_{\text{sys}}'(t) \rb - S\lb \rho_{\text{sys}}'(t) \rb \right| \\
        &\leq \epsilon_{\text{est}} + \epsilon_{\text{Fan}},
    \end{split}
    \end{equation}
    where $\epsilon_{\text{Fan}}$ is determined by Fannes inequality (Lemma \ref{lem:fannes}). To achieve overall error $\leq \epsilon$ it suffices to ensure that 
    $\epsilon_{\text{est}} \leq \epsilon/2$ and $\epsilon_{\text{Fan}} \leq \epsilon/2$.
    Let us now bound $\epsilon_{\text{Fan}}$ in terms of $\epsilon_U$ and $\nu$. For simplicity, let 
    \begin{equation}
        \ket{\Psi_0} := \ket{\psi_0}\ket{0}\ket{0}.
    \end{equation}
    Then we have that
    \begin{equation}
    \begin{split}
        \norm{\widetilde{U}_{L_{NVT}}' \ketbra{\Psi_0}{\Psi_0} \widetilde{U}_{L_{NVT}}'^{\dagger} - U_{L_{NVT}}' \ketbra{\Psi_0}{\Psi_0} U_{L_{NVT}}'^{\dagger}} &\leq \norm{\widetilde{U}_{L_{NVT}}' \ketbra{\Psi_0}{\Psi_0} \widetilde{U}_{L_{NVT}}'^{\dagger} - \widetilde{U}_{L_{NVT}}' \ketbra{\Psi_0}{\Psi_0} U_{L_{NVT}}'^{\dagger}} \\
        & \quad + \norm{\widetilde{U}_{L_{NVT}}' \ketbra{\Psi_0}{\Psi_0} U_{L_{NVT}}'^{\dagger} - U_{L_{NVT}}' \ketbra{\Psi_0}{\Psi_0} U_{L_{NVT}}'^{\dagger}} \\
        &\leq \norm{\widetilde{U}_{L_{NVT}}'^{\dagger} - U_{L_{NVT}}'^{\dagger}} + \norm{\widetilde{U}_{L_{NVT}}' - U_{L_{NVT}}'} \\
        &\leq 2 \epsilon_{U}.
    \end{split}
    \end{equation}
    It follows from Definition \ref{def:trace_dist} that
    \begin{equation}
        \frac{1}{2}\norm{\widetilde{U}_{L_{NVT}}' \ketbra{\Psi_0}{\Psi_0} \widetilde{U}_{L_{NVT}}'^{\dagger} - U_{L_{NVT}}' \ketbra{\Psi_0}{\Psi_0} U_{L_{NVT}}'^{\dagger}}_1 \leq \eta_{\text{pur}} \epsilon_U,
    \end{equation}
    where $\eta_{\text{pur}}$ is the dimension of $U_{L_{NVT}}'$ (or equivalently, of $\ket{\Psi_0}$). Since the trace distance is contractive under the partial trace, we obtain the following bound:
    \begin{equation}
         \mathcal{T}\lb \widetilde{\rho}_{\text{sys}}'(t), \rho_{\text{sys}}'(t) \rb \leq \eta_{\text{pur}} \epsilon_U.
    \end{equation}
    By Lemma \ref{lem:fannes} we then have that
    \begin{equation}
        \epsilon_{\text{Fan}} = \left| S\lb \widetilde{\rho}_{\text{sys}}'(t) \rb - S\lb \rho_{\text{sys}}'(t) \rb \right| \leq 2\eta_{\text{pur}} \epsilon_U \lb \log(\eta) - \log \lb \nu \rb \rb
    \end{equation}
    as long as $\mathcal{T}\lb \widetilde{\rho}_{\text{sys}}'(t), \rho_{\text{sys}}'(t) \rb \in \left[\frac{\nu}{2}, \frac{1}{2e}\right]$. If 
    \begin{equation}
        \epsilon_U \leq \frac{\epsilon}{4\eta_{\text{pur}} \log \lb \eta/\nu \rb} \leq \frac{1}{2\eta_{\text{pur}} e}
    \end{equation}
    then $\epsilon_{\text{Fan}} \leq \epsilon/2$ as desired. Note that this requires $\epsilon \leq \frac{2 \log \lb \eta/\nu \rb}{e}$. In our case we always have $\eta \geq 6$ since the phase space is at least 6-dimensional. This implies that $\frac{2 \log \lb \eta/\nu \rb}{e} \geq 1$ for any $\nu \in (0,1)$. Demanding $\epsilon \in (0,1)$ is thus a sufficiently restrictive criterion.
\end{proof}

A challenge facing this algorithm arises from its scaling with the dimension of the space. In general, it scales exponentially with the number of particles and hence we cannot compute the entropy directly. An alternative approach is to coarse-grain the position and momentum variables of the nuclei, e.g., by tracing out the $l$ least significant qubits associated with each position or momentum variable. This effectively reduces the dimension of the density matrices $\rho_{\text{sys}}(t)$ and $\rho_{\text{sys}}'(t)$ from 
\begin{equation}
    \eta = 2^{3N \lb \log g_x + \log g_{p'} \rb + \log(g_s)}
\end{equation}
to 
\begin{equation}
    \eta' = 2^{3N \lb \log{g_x} + \log{g_{p'}} - 2l \rb + \log(g_s)}
\end{equation}
and accordingly only $\widetilde{O} \lb \frac{\eta'}{\epsilon^{1.5}} \log {\lb \frac{1}{\xi} \rb} \rb$ queries to $\widetilde{U}_{L_{NVT}}'$ are required. The exact entropy can then be estimated by extrapolating the entropy in the limit where the coarse-graining tends to zero.

\subsection{Estimating the internal energy}
\label{app:internal_energy}

Next, let us discuss how to estimate the internal energy $\mathcal{U}$ of our system. 
First, note that a classical system can be described by a density matrix $\rho$ and a Hamiltonian $H$ both of which are diagonal in the computational basis. The internal energy of a classical system can thus be computed as follows:
\begin{equation}
    \mathcal{U} = \text{Tr} \lb  \rho H \rb.
\end{equation}
In our case, we can identify $\rho \equiv \rho_{\text{sys}}'(t)$ and $H \equiv H_{\text{nuc}}$. Recall from Section \ref{sec:prelim} that 
\begin{equation}
    H_{\text{nuc}} = H_{\text{kin}} + H_{\text{pot}} + H_{E_{\tel}}
\end{equation}
where
\begin{align*}
    H_{\text{kin}} &= \sum_{n,j} \sum_{\hap'_{n,j}} \sum_{\has} \frac{{p'}_{n,j}^2}{m_n (s + s_{\min})^2} \ketbra{{\hap'_{n,j}}}{{\hap'_{n,j}}} \otimes \ketbra{\has}{\has} \\
    H_{\text{pot}} &= \sum_{n' \neq n} \sum_{\hax_{n}} \sum_{\hax_{n'}} \frac{Z_n Z_{n'}}{\lb \norm{x_n - x_{n'}}^2 + \Delta^2 \rb^{1/2}}\ketbra{{\hax_n}}{{\hax_n}} \otimes \ketbra{{\hax_{n'}}}{{\hax_{n'}}} \\
    H_{E_{\tel}} &= \sum_{\{ \hax_n \}} E_{\tel} \lb \{ x_n \} \rb \ketbra{\{ \hax_n \}}{\{ \hax_n \}}.
\end{align*}
The idea then is to block-encode each of the three terms of $H_{\text{nuc}}$, use the Hadamard test to estimate the expectation value of each term individually and then add the results classically.

Note that $H_{\text{nuc}}$ is diagonal in the nuclear position and momentum basis. Since the block-encoding of $H_{\text{nuc}}$ will also be diagonal in the nuclear position and momentum basis, we technically do not need to worry about getting rid of the off-diagonal elements of $\rho_{\text{sys}}(t)$ since 
\begin{equation}
    \text{Tr} \lb \rho_{\text{sys}}(t) H_{\text{nuc}} \rb = \text{Tr} \lb \rho_{\text{sys}}'(t) H_{\text{nuc}} \rb.
\end{equation}
However, we will use $\rho_{\text{sys}}'(t)$ to be consistent with the previous discussion on estimating the Gibbs entropy.

The following three lemmas provide upper bounds on the cost of block-encoding the three terms of $H_{\text{nuc}}$.

\begin{lem}[Block-encoding of $H_{\text{kin}}$]
    There exists an $(\alpha_{\text{kin}}, a_{\text{kin}}, \epsilon)$-block-encoding of $H_{\text{kin}}$ with normalization constant
    \begin{equation}
        \alpha_{\text{kin}} \in O \lb N \frac{{p'}_{\max}^2}{m_{\min} s_{\min}^2} \rb
    \end{equation}
    and a number of ancilla qubits
    \begin{equation}
        a_{\text{kin}} \in O \lb \log \lb \frac{\alpha_{\text{kin}}}{\epsilon} \rb \rb
    \end{equation}
    that can be implemented using
    \begin{equation}
        O \lb N \log \lb \frac{ g_{p'} \alpha_{\text{kin}}}{\epsilon} \rb + \log^{\log 3} \lb \frac{\alpha_{\text{kin}}}{\epsilon} \rb \rb
    \end{equation}
    Toffoli gates.
\label{lem:H_kin}
\end{lem}

\begin{proof}
    The proof of Lemma \ref{lem:H_kin} proceeds along the same lines as the proof of Lemma \ref{lem:bounds_L_class_NVE}. In particular, we use the alternating sign trick~\cite{Berry2014alter_sign_trick, Su2021first_quant_sim} to block-encode a single summand of $H_{\text{kin}}$ and then use Lemma \ref{lem:sum_block_encoding} to combine the block-encodings of all $3N$ terms. More specifically, we use the alternating sign trick to construct $U_{p_{n,j}}$ which provides an $\lb \alpha_p, a_p , \epsilon_p \rb$-block-encoding of
    \begin{equation}
        \sum_{\hap'_{n,j}} \sum_{\has} \frac{{p'}_{n,j}^2}{(s + s_{\min})^2} \ketbra{{\hap'_{n,j}}}{{\hap'_{n,j}}} \otimes \ketbra{\has}{\has}.
    \end{equation}
    Then we use the following to prepare the distribution of coefficients for the masses of each particle in the system under the assumption of three-dimensional dynamics:
    \begin{equation}
        \prep_m \ket{0} := \sum_{n=1}^{N} \sqrt{\frac{1/m_n}{\alpha_m}} \ket{n} \otimes \frac{1}{\sqrt{3}} \sum_{j=1}^3 \ket{j}, \label{eq:prepm}
    \end{equation}
    where 
    \begin{equation}
        \alpha_m = \sum_{n,j} \frac{1}{m_n} \leq \frac{3N}{m_{\min}}.
    \end{equation}
    The above definition implies that $a_m = \lceil \log{N} \rceil + \lceil \log{3} \rceil$.

    Let us now explain how to use the alternating sign trick to implement $U_{p_{n,j}}$.
    The idea is to prepare an ancilla register consisting of $a_p$ qubits in the state 
    \begin{equation}
        \prep_p \ket{0} := \sum_{l=0}^{2^{a_p}-1} \frac{1}{\sqrt{2^{a_p}}} \ket{l},
    \end{equation} 
    and then use the inequality testing circuit from Fig.~\ref{fig:inequality_circuit} to
    test the following inequality:
    \begin{equation}
        l \lb \has + \has_{\min} \rb^2 \leq {\hap'_{n,j}}{^2},
    \label{H_kin_ineq}
    \end{equation}
    where $\has_{\min} \in \mathbb{N}$ such that $s_{\min} = \has_{\min} h_s$.
    As long as $l$ satisfies the above inequality, the coefficient of $\ket{l}$ is set to +1 as is done in existing work involving LCU or qubitization for general purpose simulations~\cite{Berry2014alter_sign_trick,Su2021first_quant_sim}. For larger $l$ the coefficient of $\ket{l}$ is set to alternate between $\pm 1$.
    To test the inequality we first use $O(1)$ quantum Karatsuba multiplications~\cite{Gidney2019} to compute the left- and right-hand side of Eq.~\eqref{H_kin_ineq}. This can be done using $O \lb a_p^{\log 3}\rb$ Toffoli gates, whereas the inequality test itself requires only $O \lb a_p \rb$ Toffolis, see Lemma \ref{lem:inequal}.  We then have that $U_{p_n,j} = \prep_p^\dagger \cdot \sel_p \cdot \prep_p$ where $\sel_p$ includes the quantum Karatsuba multiplications, the inequality testing and a controlled-$Z$ gate to obtain the desired alternating sequence of $\pm 1$. Fig.~\ref{fig:alternating_sign_trick} shows a circuit diagram of the alternating sign trick for the slightly simpler case of block-encoding $\sum_{\hap_{n,j}=0}^{g_p-1} p_{n,j} \ketbra{\hap_{n,j}}{\hap_{n,j}}$.
    The number of ancilla qubits $a_p$ determines the precision $\epsilon_p$ of $U_{p_{n,j}}$. In particular,
    \begin{equation}
        \norm{(\bra{0} \otimes \mathbb{1})U_{p_{n,j}}(\ket{0}\otimes \mathbb{1}) - \frac{1}{\alpha_p} \sum_{\hap'_{n,j}} \sum_{\has} \frac{{p'}_{n,j}^2}{(s + s_{\min})^2} \ketbra{{\hap'_{n,j}}}{{\hap'_{n,j}}} \otimes \ketbra{\has}{\has}} \leq \frac{1}{2^{a_p}}.
    \end{equation}
    Note that $\alpha_p \in O \lb \frac{{p'}_{\max}^2}{s_{\min}^2} \rb$.
    We can ensure that $(\bra{0}  \otimes I)U_{p_{n,j}}(\ket{0}\otimes I)$ is an $\epsilon_p$-precise approximation to $\sum_{\hap'_{n,j}} \sum_{\has} \frac{{p'}_{n,j}^2}{(s + s_{\min})^2} \ketbra{{\hap'_{n,j}}}{{\hap'_{n,j}}} \otimes \ketbra{\has}{\has}$ by choosing $a_p \in \Theta \lb \log \lb \frac{\alpha_p}{\epsilon_p} \rb \rb$.
    
    Instead of constructing $3N$ different block-encoding for each of the $3N$ terms, we use an additional ancilla register which we call a ``$\swap$ register''~\cite{Su2021first_quant_sim}. The $\mathtt{SELECT}$ operation can then be modified to swap the appropriate (virtual) momentum variable into the $\swap$ register controlled by the $\prep_m$ register. This allows us to apply the block-encoding $U_{p_{n,j}}$ only once (to the $\swap$ register holding the appropriate momentum variable) rather than $3N$ times (to each individual momentum variable). However, we do require a total of $O \lb N \log{\lb g_{p'} \rb} \rb$ $\swap$ operations, implying $O \lb N \log{\lb g_{p'} \rb} \rb$ Toffolis.
    
    Applying Lemma \ref{lem:sum_block_encoding} to $\prep_m$ and $\left\{ U_{p_{n,j}}\right\}$ yields an $\lb \alpha_m \alpha_p, a_m + a_p, \alpha_p \epsilon_m + \alpha_m \epsilon_p \rb$-block-encoding of $H_{\text{kin}}$. This implies that 
    \begin{equation}
        \alpha_{\text{kin}} = \alpha_m \alpha_p \in O \lb \frac{{N p'}_{\max}^2}{m_{\min} s_{\min}^2} \rb.
    \end{equation}
    To achieve overall block-encoding error $\leq \epsilon$ it suffices to ensure that $\epsilon_m \leq \frac{\epsilon}{2 \alpha_p}$ and $\epsilon_p \leq \frac{\epsilon}{2 \alpha_m}$. Thus,
    \begin{equation}
        a_{\text{kin}} = a_m + a_p \in O \lb \log \lb \frac{\alpha_{\text{kin}}}{\epsilon} \rb \rb.
    \end{equation}
    It follows from Lemma \ref{lem:prep_error} that we need to prepare the state $\prep_m \ket{0}$ within error $\frac{\epsilon_{m}}{\alpha_{m} \sqrt{N}}$.
    Such a general quantum state preparation has Toffoli cost in  
    \begin{equation}
        O \lb N \log \lb \frac{\alpha_{m} \sqrt{N}}{\epsilon_m} \rb \rb \subseteq O \lb N \log{\lb \frac{\alpha_{\text{kin}}}{\epsilon} \rb} \rb,
    \label{eq:stateprep1}
    \end{equation}
    where we have used the assumption that we choose the uncertainty to saturate $\epsilon_m = \epsilon/2\alpha_p$. We require another 
    \begin{equation}
        O \lb a_p^{\log 3} \rb \subseteq O \lb a_{\text{kin}}^{\log 3} \rb \subseteq O \lb  \log^{\log 3} \lb \frac{\alpha_{\text{kin}}}{\epsilon} \rb \rb
    \end{equation}
    Toffolis for the quantum Karatsuba multiplications~\cite{Gidney2019} used in the comparison test given in~\eqref{H_kin_ineq}.  Addition can be performed in linear time and thus the cost of performing the entire comparison test is given by the cost of multiplication.  This cost is additive to the cost of the state preparation given in~\eqref{eq:stateprep1}.  Combining all results yields the desired complexity expressions.
\end{proof}

\begin{lem}[Block-encoding of $H_{\text{pot}}$]
    There exists an $(\alpha_{\text{pot}}, a_{\text{pot}}, \epsilon)$-block-encoding of $H_{\text{pot}}$ with normalization constant
    \begin{equation}
        \alpha_{\text{pot}} \in O \lb N^2 \frac{Z_{\max}^2}{\Delta} \rb
    \end{equation}
    and a number of ancilla qubits
    \begin{equation}
        a_{\text{pot}} \in O \lb \log \lb \frac{\alpha_{\text{pot}}}{\epsilon} \rb \rb.
    \end{equation}
    This block-encoding can be implemented using
    \begin{equation}
        O \lb N \log \lb \frac{g_x \alpha_{\text{pot}}}{\epsilon} \rb + \log^{\log 3} \lb \frac{\alpha_{\text{pot}}}{\epsilon} \rb \rb
    \end{equation}
    Toffoli gates where $g_x$ is the number of discrete positions considered in the classical part of the Liouvillian.
\label{lem:H_pot}
\end{lem}

\begin{proof}
    We use the same strategy as in the proof of Lemma \ref{lem:H_kin}. In particular, we use the alternating sign trick~\cite{Berry2014alter_sign_trick} to block-encode a single summand of $H_{\text{pot}}$ and then use Lemma \ref{lem:sum_block_encoding} to combine the block-encodings of all $O \lb N^2 \rb$ terms. More specifically, we use the alternating sign trick to construct 
    $U_{V_{n,n',j}}$, an $\lb \alpha_V, a_V, \epsilon_V \rb$-block-encoding of
    \begin{equation}
        \sum_{\hax_{n}} \sum_{\hax_{n'}} \frac{1}{\lb \norm{x_n - x_{n'}}^2 + \Delta^2 \rb^{1/2}}\ketbra{{\hax_n}}{{\hax_n}} \otimes \ketbra{{\hax_{n'}}}{{\hax_{n'}}}.
    \end{equation}
    Then we use the following sub-prepare to attach the atomic numbers $\{Z_n\}$:
    \begin{equation}
        \prep_Z \ket{0} := \frac{1}{\sqrt{\alpha_Z}} \sum_{n=1}^{N} \sqrt{Z_N} \ket{n} \otimes \sum_{n'=1}^{N} \sqrt{Z_{n'}} \otimes \sum_{j=1}^3 \ket{j},
    \end{equation}
    where $\alpha_Z = \sum_{n,n',j} Z_n Z_{n'} \leq 3 N^2 Z_{\max}^2$. The above definition implies that 
    \begin{equation}
        a_Z = 2\lceil\log{N}\rceil + \lceil \log{3} \rceil.
    \end{equation} 
    Importantly the resultant state is a product state, meaning that $\sum_{n=1}^{N} \sqrt{Z_n}\ket{n}$, $\sum_{n'=1}^{N} \sqrt{Z_{n'}} \ket{n'}$ and $\sum_{j=1}^3 \ket{j}$ can be prepared individually.
  
    Let us now explain the construction of $U_{V_{n,n',j}}$. 
    Using $\prep_V \ket{0} := \sum_{l=0}^{2^{a_V}-1} \frac{1}{\sqrt{2^{a_V}}} \ket{l}$, we test the following inequality:
    \begin{equation}
        l^2 \lb \norm{\hax_n - \hax_{n'}}^2 + \overline{\Delta}^2 \rb \le 1,
    \label{H_coulomb_ineq}
    \end{equation}
    where $\overline{\Delta} \in \mathbb{N}$ such that $\Delta = \overline{\Delta} h_x$.
    As long as $l$ satisfies the above inequality, the coefficient of $\ket{l}$ is set to +1. For larger $l$ the coefficient of $\ket{l}$ is set to alternate between $\pm 1$.
    To determine the correct sign we also need to test $x_{n,j} \le x_{n',j}$ which has Toffoli complexity in $O \lb \log{(g_x)} \rb$.
    The advantage of testing Eq.~({\ref{coulomb_ineq}}) rather than directly
    \begin{equation}
        l \le \frac{1}{\lb \norm{\hax_n - \hax_{n'}}^2 + \overline{\Delta}^2 \rb^{1/2}}
    \end{equation}
    is that we do not have to calculate fractions containing square roots.
    However, the inequality test in Eq.~({\ref{H_coulomb_ineq}}) does require us first to compute the left- and right-hand sides of the inequality using $O(1)$ quantum Karatsuba multiplications. This can be done using $O\lb \lb a_V \rb^{\log 3} \rb$ Toffoli gates~\cite{Gidney2019}, whereas the inequality test itself requires only $O\lb a_V \rb$ Toffolis~\cite{Gidney2019}.
    
    The number of ancilla qubits $a_V$ determines the precision $\epsilon_V$ of $U_{V_{n,n',j}}$. In particular,
    \begin{equation}
        \norm{(\bra{0} \otimes \mathbb{1})U_{V_{n,n',j}}(\ket{0}\otimes \mathbb{1}) - \frac{1}{\alpha_V} \sum_{\hax_{n}} \sum_{\hax_{n'}} \frac{1}{\lb \norm{x_n - x_{n'}}^2 + \Delta^2 \rb^{1/2}}\ketbra{{\hax_n}}{{\hax_n}} \otimes \ketbra{{\hax_{n'}}}{{\hax_{n'}}}} \leq \frac{1}{2^{a_V}}.
    \label{eq:vBlockEnc}\end{equation}
    Note that $\alpha_V \in O \lb \frac{1}{\Delta} \rb$.
    We can ensure that $U_{V_{n,n',j}}$ is an $\epsilon_V$-precise approximation by choosing 
    \begin{align}
        a_V \in \Theta \lb \log \lb \frac{\alpha_V}{\epsilon_V} \rb \rb.\label{eq:av_bound}
    \end{align}
    
    Instead of constructing $O \lb N^2 \rb$ different block-encoding for each of the $O \lb N^2 \rb$ terms of $H_{\text{pot}}$, we use 6 $\swap$ registers for the 6 nuclear position variables appearing in $\frac{1}{\lb \norm{x_n - x_{n'}}^2 + \Delta^2 \rb^{1/2}}$.  Here the factor of $6$ occurs because we assume that we are interested in dynamics in three spatial dimensions.
    Controlled by the $\prep_Z$ register we swap the appropriate position variables into the $\swap$ registers. This allows us to apply the block-encoding $U_{V_{n,n',j}}$ only once (to the $\swap$ registers holding the appropriate position variables) rather than $O \lb N^2 \rb$ times (for each individual term). However, we do require a total of $O \lb N \log{\lb g_{x} \rb} \rb$ $\swap$ operations, resulting in $O \lb N \log{\lb g_{x} \rb} \rb$ Toffolis.
    
    Applying Lemma \ref{lem:sum_block_encoding} to $\prep_Z$ and $\left\{ U_{V_{n,n',j}} \right\}$ yields an $\lb \alpha_Z \alpha_V, a_Z + a_V, \alpha_V \epsilon_Z + \alpha_Z \epsilon_V \rb$-block-encoding of $H_{\text{pot}}$. This implies that 
    \begin{equation}
        \alpha_{\text{pot}} = \alpha_Z \alpha_V \in O \lb \frac{{N^2 Z_{\max}^2}}{\Delta} \rb.
    \end{equation}
    To achieve overall block-encoding error $\leq \epsilon$ it suffices to ensure that $\epsilon_Z \leq \frac{\epsilon}{2 \alpha_V}$ and $\epsilon_V \leq \frac{\epsilon}{2 \alpha_Z}$. Thus the total number of qubits required for the block encoding of the potential operator is 
    \begin{equation}
        a_{\text{pot}} = a_Z + a_V \in O \lb \log \lb \frac{\alpha_{\text{pot}}}{\epsilon} \rb \rb,
    \end{equation}
    where the latter asymptotic bound follows from substituting into~\eqref{eq:av_bound}.
    It follows from Lemma \ref{lem:prep_error} that we need to prepare the state $\prep_Z \ket{0}$ within error $\frac{\epsilon_Z}{N \alpha_Z}$.
    Preparing such a product state has Toffoli cost in 
    \begin{equation}
        O \lb N \log\lb \frac{\alpha_Z N}{\epsilon_{Z}} \rb \rb \subseteq O \lb N \log \lb \frac{\alpha_{\text{pot}}}{\epsilon} \rb \rb.
    \end{equation}
    We require another 
    \begin{equation}
        O \lb a_Z^{\log 3} \rb \subseteq O \lb a_{\text{pot}}^{\log 3} \rb \subseteq O \lb  \log^{\log 3} \lb \frac{\alpha_{\text{pot}}}{\epsilon} \rb  \rb
    \end{equation}
    Toffolis for the quantum Karatsuba multiplications. Combining all results yields the desired complexity expressions.
\end{proof}

\begin{lem}[Block-encoding of $H_{E_{\tel}}$]
    There exists an $\lb\alpha_{E_{\tel}}, a_{E_{\tel}}, \epsilon \rb$-block-encoding of $H_{E_{\tel}}$ with normalization constant
    \begin{equation}
        \alpha_{E_{\tel}} \in O \lb \lambda \rb
    \end{equation}
    and a number of ancilla qubits
    \begin{equation}
        a_{E_{\tel}} \in O \lb \log{\lb \frac{\lambda}{\epsilon} \rb} \rb.
    \end{equation}
    This block-encoding can be implemented using
    \begin{equation}
        \widetilde{O} \lb \lambda \lb \frac{N + \tn + \log \lb B \rb}{\epsilon} + \frac{N + \tn + \log \lb B \rb + \log^2 \lb 1/\epsilon \rb}{\gamma \delta} \rb \rb 
    \end{equation}
    Toffoli gates and
    \begin{equation}
        O \lb \frac{1}{\delta}  \rb
    \end{equation}
    queries to the initial electronic state preparation oracle $U_I$ from Definition \ref{def:state_prep}.
\label{lem:H_E_el}
\end{lem}

\begin{proof}
    The main strategy is to prepare the electronic ground states in superposition over all nuclear configurations and then use qubitization together with quantum phase estimation (QPE) to obtain estimates of the ground state energies of $H_{\tel}$ in superposition over all nuclear configurations. Then we use the alternating sign trick to construct a block-encoding of $H_{E_{\tel}}$. Let us now discuss the different subroutines and their errors in more detail.

    First, we require access to a block-encoding of the electronic Hamiltonian $H_{\tel}$. Ref.~\cite{Su2021first_quant_sim} provides explicit $\prep_{\tel}$ and $\sel_{\tel}$ subroutines for obtaining $U_{H_{\tel}}$, a $\lb \lambda, a_{\tel}, \epsilon_{\tel} \rb$-block-encoding of $H_{\tel}$. Let $\widetilde{H}_{\tel}$ denote the block-encoded operator satisfying
    \begin{equation}
        \norm{\widetilde{H}_{\tel} - H_{\tel}} \leq \epsilon_{\tel}.
    \end{equation}
    Since $H_{\tel}$ and $\widetilde{H}_{\tel}$ are both Hermitian, we can use eigenvalue perturbation theory~\cite{Horn1985} to conclude that $\left| \widetilde{E}_{\tel} - E_{\tel} \right| \leq \epsilon_{\tel}$ for all $\{ x_n \}$ where $\widetilde{E}_{\tel}$ is the ground state energy of the block-encoded operator $\widetilde{H}_{\tel}$. It then holds that 
    \begin{equation}
        \norm{\widetilde{H}_{E_\tel} - H_{E_{\tel}}} \leq \epsilon_{\tel}
    \end{equation}
    where 
    \begin{equation}
        \widetilde{H}_{E_\tel} := \sum_{\{ \hax_n \}} \widetilde{E}_{\tel} \lb \{ x_n \} \rb \ketbra{\{ \hax_n \}}{\{ \hax_n \}}.
    \end{equation}

    Next, we explain the electronic ground state preparation. 
    Let $W$ denote the unitary that prepares an approximate ground state of $\widetilde{H}_{\tel}$ for fixed nuclear positions according to Lemma \ref{lem:ground_state_prep}, i.e.
    \begin{equation}
        W\ket{\{\hax_n\}}\ket{0} = \ket{\{\hax_n\}}\ket{\widetilde{\phi}_0 \lb \{x_{n,j}\} \rb}
    \end{equation}
    with
    \begin{equation}
        |\braket{\widetilde{\psi}_0 \lb \{x_{n,j}\} \rb}{\widetilde{\phi}_0 \lb \{x_{n,j}\} \rb}| \ge 1- \epsilon_{\text{prep}}.
    \end{equation}
    Note that we can view $U_{H_{\tel}}$ as an exact block-encoding of $\widetilde{H}_{\tel}$, which allows us to use Lemma \ref{lem:ground_state_prep} directly without further error propagation. This means that we can prepare $\ket{\widetilde{\phi}_0 \lb \{x_{n,j}\} \rb}$ using
    \begin{equation}
        O \lb \frac{\lambda}{\gamma \delta} \log{\lb \frac{1}{\delta \epsilon_{\text{prep}}} \rb} \rb
    \end{equation}
    queries to $U_{H_{\tel}}$ and
    \begin{equation}
        O \lb \frac{1}{\delta} \rb
    \end{equation}
    queries to $U_I$.
    In the following discussion we will mostly refrain from writing out the $\lb \{x_{n,j}\} \rb$-dependence explicitly unless needed for clarity.
    Now, it holds that 
    \begin{equation}
        \ket{\widetilde{\phi}_0} = e^{i\alpha} \lb 1- \epsilon'_{\text{prep}} \rb \ket{\widetilde{\psi}_0} + \beta \ket{\widetilde{\psi}_0^\perp}
    \end{equation}
    for some angle $\alpha \in [0, 2\pi)$, $0 \le \epsilon'_{\text{prep}} \le \epsilon_{\text{prep}}$ and $|\beta|^2 = 2\epsilon'_{\text{prep}} - \lb \epsilon'_{\text{prep}} \rb^2 \le 2 \epsilon'_{\text{prep}}$. Letting 
    \begin{equation}
        \ket{\psi_0'} := e^{i\alpha}\ket{\widetilde{\psi}_0}
    \end{equation}
    we thus have that
    \begin{equation}
        \norm{W\ket{\{\hax_n \}}\ket{0} - \ket{\{\hax_n \}}\ket{\psi_0'}} = 
        \norm{\ket{\{\hax_n\}}\ket{\widetilde{\phi}_0} - \ket{\{\hax_n\}}\ket{\psi_0'}} \le \sqrt{2\epsilon_{\text{prep}}}.
    \end{equation}
    
    Next, we apply QPE with the following qubitization operator to the electronic register holding the electronic ground states:
    \begin{equation}
        Q := \lb 2 \ketbra{0}{0}  - \mathbb{1} \rb \cdot \prep_{\tel}^\dagger \cdot \sel_{\tel} \cdot \prep_{\tel}.
    \end{equation}
    The work of~\cite{Babbush2019} shows that $Q$ has eigenvalues $e^{\pm i \cos^{-1}{\lb \widetilde{E}_{k}/\lambda \rb}}$. QPE allows us to obtain the state $\ket{\psi_{\widetilde{E}'_{\tel}}}$, which encodes an $\epsilon_{QPE}$-precise estimate $\widetilde{E}'_{\tel}$ of the ground state energy $\widetilde{E}_{\tel}$ of $\widetilde{H}_{\tel}$ for fixed nuclear positions, with success probability $\geq 1 - \xi_{QPE}$ using
    \begin{equation}
        O \lb \frac{\lambda}{\epsilon_{QPE}} \log \lb \frac{1}{\xi_{QPE}} \rb \rb
    \end{equation}
    queries to $Q$. In other words, 
    \begin{equation}
        \ket{\psi_{\widetilde{E}'_{\tel}}} = \sqrt{1 - \xi_{QPE}'}\ket{\widetilde{E}'_{\tel}} + \sqrt{\xi_{QPE}'} \ket{\widetilde{E}_{\tel}^\perp}
    \end{equation}
    for some $0 \leq \xi_{QPE}' \leq \xi_{QPE}$.
    We refer to the corresponding unitary that prepares $\ket{\psi_{\widetilde{E}'_{\tel}}}$ as $U_Q$, i.e. 
    \begin{equation}
        U_Q \ket{\{\hax_n\}}\ket{\psi'_0}\ket{0} = \sqrt{1 - \xi_{QPE}'}\ket{\{\hax_n\}}\ket{\psi'_0}\ket{\widetilde{E}'_{\tel}} +  \sqrt{\xi_{QPE}'} \ket{\{\hax_n\}}\ket{\psi'_0}\ket{\widetilde{E}_{\tel}^\perp}.
    \end{equation}

    Lastly, we apply the alternating sign trick, which is explained in detail in the proof of Lemma \ref{lem:bounds_L_class_NVE}, to the $\ket{\psi_{\widetilde{E}'_{\tel}}}$ register to obtain $U_{\text{alt}}$, an $\epsilon_{\text{alt}}$-precise block-encoding of 
    \begin{equation}
        H_{\widetilde{E}'_{\tel}} := \sum_{\{ \hax_n \}} \widetilde{E}'_{\tel} \lb \{ x_n \} \rb \ketbra{\{ \hax_n \}}{\{ \hax_n \}}.
    \end{equation}
    The overall block-encoding error $\epsilon_{E_{\tel}}$ of $H_{E_{\tel}}$ is then given by
    \begin{equation}
    \begin{split}
        \epsilon_{E_{\tel}} := \norm{\alpha_{E_{\tel}} \lb \bra{0} \otimes  \mathbb{1} \otimes \bra{0} \otimes \bra{0} \rb W^{-1} U_Q^{-1} U_{\text{alt}} U_Q W \lb \ket{0} \otimes \mathbb{1} \otimes \ket{0} \otimes \ket{0} \rb - H_{E_\tel}},
    \end{split}
    \end{equation}
    where the first register in the expression $\lb \ket{0} \otimes \mathbb{1} \otimes \ket{0} \otimes \ket{0} \rb$ consists of the block-encoding ancilla qubits for the alternating sign trick, the second register is the nuclear position and momentum register, the third register is the electronic register and the last register is the phase estimation register.
    Note that the above definition implies that the electronic register as well as the phase estimation register are uncomputed and projected out to the $\ket{0}$ state at the end of the simulation. In other words, the error $ \epsilon_{E_{\tel}}$ is only measured within the Hilbert space of the nuclear position and momentum registers. Importantly, the error matrix 
    \begin{equation}
        \alpha_{E_{\tel}} \lb \bra{0} \otimes  \mathbb{1} \otimes \bra{0} \otimes \bra{0} \rb W^{-1} U_Q^{-1} U_{\text{alt}} U_Q W \lb \ket{0} \otimes \mathbb{1} \otimes \ket{0} \otimes \ket{0} \rb - H_{E_\tel}
    \end{equation}
    is diagonal in the nuclear position and momentum basis since $W$, $\widetilde{U}_Q$, $U_{\text{alt}}$ and $H_{E_\tel}$ are all diagonal in the nuclear position and momentum basis. Hence, $\epsilon_{E_{\tel}}$ is simply the largest value on the diagonal of $\mathcal{E}_{E_{\tel}}$. This allows us to consider the block-encoding error for each nuclear computational basis state separately. It also implies that $\alpha_{E_{\tel}} \in O \lb \lambda \rb$.
    Suppressing the nuclear momentum register we then obtain the following:
    \begin{equation}
    \begin{split}
        \epsilon_{E_{\tel}} &= \max_{\{\hax_n\}} \norm{\alpha_{E_{\tel}} \lb \bra{0} \otimes  \mathbb{1} \otimes \bra{0} \otimes \bra{0} \rb W^{-1} U_Q^{-1} U_{\text{alt}} U_Q W \lb \ket{0} \otimes \mathbb{1} \otimes \ket{0} \otimes \ket{0} \rb \ket{\{\hax_n\}} - H_{E_\tel} \ket{\{\hax_n\}}} \\
        &\leq \max_{\{\hax_n\}} \norm{\alpha_{E_{\tel}} \lb \bra{0} \otimes  \mathbb{1} \otimes \mathbb{1} \otimes \mathbb{1} \rb W^{-1} U_Q^{-1} U_{\text{alt}} U_Q W \lb \ket{0} \otimes  \mathbb{1} \otimes \mathbb{1} \otimes \mathbb{1} \rb \ket{\{\hax_n\}} \ket{0}\ket{0} - H_{E_\tel} \otimes \mathbb{1} \otimes \mathbb{1} \ket{\{\hax_n\}}\ket{0}\ket{0}} \\
        &\leq \max_{\{\hax_n\}} \norm{\alpha_{E_{\tel}} W^{-1} U_Q^{-1} \lb \bra{0} \otimes  \mathbb{1} \otimes \mathbb{1} \otimes \mathbb{1} \rb U_{\text{alt}} \lb \ket{0} \otimes  \mathbb{1} \otimes \mathbb{1} \otimes \mathbb{1} \rb U_Q W  \ket{\{\hax_n\}} \ket{0}\ket{0} - H_{E_\tel} \otimes \mathbb{1} \otimes \mathbb{1} \ket{\{\hax_n\}}\ket{0}\ket{0}}, 
    \end{split}
    \end{equation}
    where we used the fact that $\ketbra{0}{0} \otimes \mathbb{1} \otimes \mathbb{1} \otimes \mathbb{1}$ commutes with $W$ and $U_Q$ since $W$ and $U_Q$ act trivially on the first register.
    In the following discussion, we will drop the ``$\otimes \, \mathbb{1}$'' for ease of notation.
    The general strategy now is to apply the triangle inequality repeatedly to ``peel off'' the errors stemming from different subroutines layer by layer.
    First, we will isolate the error associated with $W$:
    \begin{equation}
    \begin{split}
        \epsilon_{E_{\tel}} &\leq \max_{\{\hax_n\}} \norm{\alpha_{E_{\tel}} W^{-1} U_Q^{-1} \lb \bra{0} U_{\text{alt}} \ket{0} \rb U_Q W \ket{\{\hax_n\}} \ket{0}\ket{0} - \alpha_{E_{\tel}} W^{-1} U_Q^{-1} \lb \bra{0} U_{\text{alt}} \ket{0} \rb U_Q \ket{\{\hax_n\}} \ket{\psi'_0}\ket{0}} \\
        &\quad + \max_{\{\hax_n\}} \norm{\alpha_{E_{\tel}} W^{-1} U_Q^{-1} \lb \bra{0} U_{\text{alt}} \ket{0} \rb U_Q \ket{\{\hax_n\}} \ket{\psi'_0}\ket{0} - H_{E_\tel} \ket{\{\hax_n\}}\ket{0}\ket{0}}.
    \label{epsilon_E_ineq}
    \end{split}
    \end{equation}
    By definition, the spectral norm is subordinate to the $\ell^2$-norm. Furthermore, we have that $\norm{U} = 1$ for any unitary $U$. With this in mind we can bound the first term on the RHS of the above inequality as follows:
    \begin{equation}
    \begin{split}
        &\max_{\{\hax_n\}} \norm{\alpha_{E_{\tel}} W^{-1} U_Q^{-1} \lb \bra{0} U_{\text{alt}} \ket{0} \rb U_Q W \ket{\{\hax_n\}} \ket{0}\ket{0} - \alpha_{E_{\tel}} W^{-1} U_Q^{-1} \lb \bra{0} U_{\text{alt}} \ket{0} \rb U_Q \ket{\{\hax_n\}} \ket{\psi'_0}\ket{0}} \\
        &\quad \leq \max_{\{\hax_n\}} \alpha_{E_{\tel}} \norm{W^{-1}} \norm{U_Q^{-1}} \norm{\lb \bra{0} U_{\text{alt}} \ket{0} \rb} \norm{U_Q} \norm{W \ket{\{\hax_n\}} \ket{0}\ket{0} - \ket{\{\hax_n\}} \ket{\psi'_0}\ket{0}} \\
        &\quad \leq \alpha_{E_{\tel}} \sqrt{2\epsilon_{\text{prep}}}.
    \end{split}
    \end{equation}
    To bound the second term on the RHS of \eqref{epsilon_E_ineq}, we use the triangle inequality again:
    \begin{equation}
    \begin{split}
        &\max_{\{\hax_n\}} \norm{\alpha_{E_{\tel}} W^{-1} U_Q^{-1} \lb \bra{0} U_{\text{alt}} \ket{0} \rb U_Q \ket{\{\hax_n\}} \ket{\psi'_0}\ket{0} - H_{E_\tel} \ket{\{\hax_n\}}\ket{0}\ket{0}} \\
        &\quad \leq \max_{\{\hax_n\}} \norm{\alpha_{E_{\tel}} W^{-1} U_Q^{-1} \lb \bra{0} U_{\text{alt}} \ket{0} \rb U_Q \ket{\{\hax_n\}} \ket{\psi'_0}\ket{0} - \alpha_{E_{\tel}} W^{-1} U_Q^{-1} \lb \bra{0} U_{\text{alt}} \ket{0} \rb \ket{\{\hax_n\}} \ket{\psi'_0}\ket{\widetilde{E}'_{\tel}}} \\
        &\qquad  + \max_{\{\hax_n\}} \norm{\alpha_{E_{\tel}} W^{-1} U_Q^{-1} \lb \bra{0} U_{\text{alt}} \ket{0} \rb \ket{\{\hax_n\}} \ket{\psi'_0}\ket{\widetilde{E}'_{\tel}} - H_{E_\tel} \ket{\{\hax_n\}}\ket{0}\ket{0}}
    \label{QPE_fail_ineq}
    \end{split}
    \end{equation}
    The first term on the RHS of the above inequality can be bounded as follows:
    \begin{equation}
    \begin{split}
        &\max_{\{\hax_n\}} \norm{\alpha_{E_{\tel}} W^{-1} U_Q^{-1} \lb \bra{0} U_{\text{alt}} \ket{0} \rb U_Q \ket{\{\hax_n\}} \ket{\psi'_0}\ket{0} - \alpha_{E_{\tel}} W^{-1} U_Q^{-1} \lb \bra{0} U_{\text{alt}} \ket{0} \rb \ket{\{\hax_n\}} \ket{\psi'_0}\ket{\widetilde{E}'_{\tel}}} \\
        &\leq \max_{\{\hax_n\}} \alpha_{E_{\tel}} \norm{W^{-1}} \norm{U_Q^{-1}} \norm{\lb \bra{0} U_{\text{alt}} \ket{0} \rb} \norm{U_Q \ket{\{\hax_n\}} \ket{\psi'_0}\ket{0} - \ket{\{\hax_n\}} \ket{\psi'_0}\ket{\widetilde{E}'_{\tel}}} \\
        &\leq \alpha_{E_{\tel}} \sqrt{\xi_{QPE}}.
    \end{split}
    \end{equation}
    Note that the failure probability $\xi_{QPE}$ of the phase estimation step is now part of the block-encoding error of $H_{E_{\tel}}$ in addition to the actual phase estimation error $\epsilon_{QPE}$. 
    Before explaining how $\epsilon_{QPE}$ contributes to the block-encoding error $\epsilon_{E_{\tel}}$ we first isolate the error $\epsilon_{\text{alt}}$ associated with the alternating sign trick. This can be done by applying the triangle inequality to the second term on the RHS of \eqref{QPE_fail_ineq}:
    \begin{equation}
    \begin{split}
        &\max_{\{\hax_n\}} \norm{\alpha_{E_{\tel}} W^{-1} U_Q^{-1} \lb \bra{0} U_{\text{alt}} \ket{0} \rb \ket{\{\hax_n\}} \ket{\psi'_0}\ket{\widetilde{E}'_{\tel}} - H_{E_\tel} \ket{\{\hax_n\}}\ket{0}\ket{0}} \\
        &\quad \leq \max_{\{\hax_n\}} \norm{\alpha_{E_{\tel}} W^{-1} U_Q^{-1} \lb \bra{0} U_{\text{alt}} \ket{0} \rb \ket{\{\hax_n\}} \ket{\psi'_0}\ket{\widetilde{E}'_{\tel}} -  W^{-1} U_Q^{-1} H_{\widetilde{E}'_{\tel}} \ket{\{\hax_n\}} \ket{\psi'_0}\ket{\widetilde{E}'_{\tel}}} \\
        &\qquad + \max_{\{\hax_n\}} \norm{W^{-1} U_Q^{-1} H_{\widetilde{E}'_{\tel}} \ket{\{\hax_n\}} \ket{\psi'_0}\ket{\widetilde{E}'_{\tel}} - H_{E_\tel} \ket{\{\hax_n\}}\ket{0}\ket{0}}.
    \label{alt_error_ineq}
    \end{split}
    \end{equation}
    The first term on the RHS of the above inequality can then be bounded as follows:
    \begin{equation}
    \begin{split}
        &\max_{\{\hax_n\}} \norm{\alpha_{E_{\tel}} W^{-1} U_Q^{-1} \lb \bra{0} U_{\text{alt}} \ket{0} \rb \ket{\{\hax_n\}} \ket{\psi'_0}\ket{\widetilde{E}'_{\tel}} -  W^{-1} U_Q^{-1} H_{\widetilde{E}'_{\tel}} \ket{\{\hax_n\}} \ket{\psi'_0}\ket{\widetilde{E}'_{\tel}}} \\
        &\quad \leq \max_{\{\hax_n\}} \norm{W^{-1}} \norm{U_Q^{-1}} \norm{\alpha_{E_{\tel}} \lb \bra{0} U_{\text{alt}} \ket{0} \rb \ket{\{\hax_n\}} \ket{\psi'_0}\ket{\widetilde{E}'_{\tel}} - H_{\widetilde{E}'_{\tel}} \ket{\{\hax_n\}} \ket{\psi'_0}\ket{\widetilde{E}'_{\tel}}} \\
        &\quad \leq \epsilon_{\text{alt}}.
    \end{split}
    \end{equation}
    Applying the triangle inequality to the second term on the RHS of \eqref{alt_error_ineq} now allows us to isolate the phase estimation error $\epsilon_{QPE}$:
    \begin{equation}
    \begin{split}
        &\max_{\{\hax_n\}} \norm{W^{-1} U_Q^{-1} H_{\widetilde{E}'_{\tel}} \ket{\{\hax_n\}} \ket{\psi'_0}\ket{\widetilde{E}'_{\tel}} - H_{E_\tel} \ket{\{\hax_n\}}\ket{0}\ket{0}} \\
        &\quad \leq \max_{\{\hax_n\}} \norm{W^{-1} U_Q^{-1} H_{\widetilde{E}'_{\tel}} \ket{\{\hax_n\}} \ket{\psi'_0}\ket{\widetilde{E}'_{\tel}} - W^{-1} U_Q^{-1} H_{\widetilde{E}_{\tel}} \ket{\{\hax_n\}} \ket{\psi'_0}\ket{\widetilde{E}'_{\tel}}} \\
        &\qquad + \max_{\{\hax_n\}} \norm{W^{-1} U_Q^{-1} H_{\widetilde{E}_{\tel}} \ket{\{\hax_n\}} \ket{\psi'_0}\ket{\widetilde{E}'_{\tel}} - H_{E_\tel} \ket{\{\hax_n\}}\ket{0}\ket{0}}.
    \label{QPE_error_ineq}
    \end{split}
    \end{equation}
    The first term on the RHS of the above inequality can be bounded as follows:
    \begin{equation}
    \begin{split}
        &\max_{\{\hax_n\}} \norm{W^{-1} U_Q^{-1} H_{\widetilde{E}'_{\tel}} \ket{\{\hax_n\}} \ket{\psi'_0}\ket{\widetilde{E}'_{\tel}} - W^{-1} U_Q^{-1} H_{\widetilde{E}_{\tel}} \ket{\{\hax_n\}} \ket{\psi'_0}\ket{\widetilde{E}'_{\tel}}} \\
        &\quad \leq \max_{\{\hax_n\}} \norm{W^{-1}} \norm{U_Q^{-1}} \norm{H_{\widetilde{E}'_{\tel}} \ket{\{\hax_n\}} \ket{\psi'_0}\ket{\widetilde{E}'_{\tel}} - H_{\widetilde{E}_{\tel}} \ket{\{\hax_n\}} \ket{\psi'_0}\ket{\widetilde{E}'_{\tel}}} \\
        &\quad \leq \epsilon_{QPE.}
    \end{split}
    \end{equation}
    Next, we isolate the block-encoding error $\epsilon_{\tel}$ of the electronic Hamiltonian by applying the triangle inequality to the second term on the RHS of \eqref{QPE_error_ineq}:
    \begin{equation}
    \begin{split}
        &\max_{\{\hax_n\}} \norm{W^{-1} U_Q^{-1} H_{\widetilde{E}_{\tel}} \ket{\{\hax_n\}} \ket{\psi'_0}\ket{\widetilde{E}'_{\tel}} - H_{E_\tel} \ket{\{\hax_n\}}\ket{0}\ket{0}} \\
        &\quad \leq \max_{\{\hax_n\}} \norm{W^{-1} U_Q^{-1} H_{\widetilde{E}_{\tel}} \ket{\{\hax_n\}} \ket{\psi'_0}\ket{\widetilde{E}'_{\tel}} - W^{-1} U_Q^{-1} H_{E_\tel} \ket{\{\hax_n\}}\ket{0}\ket{0}} \\
        &\qquad + \max_{\{\hax_n\}} \norm{W^{-1} U_Q^{-1} H_{E_{\tel}} \ket{\{\hax_n\}} \ket{\psi'_0}\ket{\widetilde{E}'_{\tel}} - H_{E_\tel} \ket{\{\hax_n\}}\ket{0}\ket{0}}.
    \label{el_error_ineq}
    \end{split}
    \end{equation}
    The first term on the RHS of the above inequality can then be bounded as follows:
    \begin{equation}
    \begin{split}
        &\max_{\{\hax_n\}} \norm{W^{-1} U_Q^{-1} H_{\widetilde{E}_{\tel}} \ket{\{\hax_n\}} \ket{\psi'_0}\ket{\widetilde{E}'_{\tel}} - W^{-1} U_Q^{-1} H_{E_\tel} \ket{\{\hax_n\}}\ket{0}\ket{0}} \\
        &\quad \leq \max_{\{\hax_n\}} \norm{W^{-1}} \norm{U_Q^{-1}} \norm{H_{\widetilde{E}_{\tel}} \ket{\{\hax_n\}} \ket{\psi'_0}\ket{\widetilde{E}'_{\tel}} - H_{E_{\tel}} \ket{\{\hax_n\}} \ket{\psi'_0}\ket{\widetilde{E}'_{\tel}}} \\
        &\quad \leq \epsilon_{\tel}.
    \end{split}
    \end{equation}
    To bound the second term on the RHS of \eqref{el_error_ineq} recall that $H_{E_{\tel}}$ is diagonal in the $\ket{\{\hax_n\}}$ basis, i.e. $H_{\widetilde{E}_{\tel}} \ket{\{\hax_n\}} = E_{\tel} \lb \lb \{x_n\} \rb \rb \ket{\{\hax_n\}}$. Thus, we have that
    \begin{equation}
    \begin{split}
        &\max_{\{\hax_n\}} \norm{W^{-1} U_Q^{-1} H_{E_{\tel}} \ket{\{\hax_n\}} \ket{\psi'_0}\ket{\widetilde{E}'_{\tel}} - H_{E_\tel} \ket{\{\hax_n\}}\ket{0}\ket{0}} \\
        &\quad \leq \max_{\{\hax_n\}} \left|E_{\tel} \lb \{x_n\} \rb \right|\norm{W^{-1} U_Q^{-1}\ket{\{\hax_n\}} \ket{\psi'_0}\ket{\widetilde{E}'_{\tel}} - \ket{\{\hax_n\}}\ket{0}\ket{0}} \\
        &\quad \leq \alpha_{E_{\tel}} \norm{W^{-1} U_Q^{-1}\ket{\{\hax_n\}} \ket{\psi'_0}\ket{\widetilde{E}'_{\tel}} - W^{-1} \ket{\{\hax_n\}} \ket{\psi'_0}\ket{0}} \\
        &\qquad + \alpha_{E_{\tel}} \norm{W^{-1} \ket{\{\hax_n\}} \ket{\psi'_0}\ket{0} - \ket{\{\hax_n\}}\ket{0}\ket{0}}. 
    \label{inv_error_ineq}
    \end{split}
    \end{equation}
    The first term on the RHS of the above inequality can then be bounded as follows:
    \begin{equation}
    \begin{split}
        &\alpha_{E_{\tel}} \norm{W^{-1} U_Q^{-1}\ket{\{\hax_n\}} \ket{\psi'_0}\ket{\widetilde{E}'_{\tel}} - W^{-1} \ket{\{\hax_n\}} \ket{\psi'_0}\ket{0}} \\
        &\quad \leq \alpha_{E_{\tel}} \norm{W^{-1}} \norm{U_Q^{-1}\ket{\{\hax_n\}} \ket{\psi'_0}\ket{\widetilde{E}'_{\tel}} - \ket{\{\hax_n\}} \ket{\psi'_0}\ket{0}} \\
        &\quad \leq \alpha_{E_{\tel}} \sqrt{\xi_{QPE}}.
    \end{split}
    \end{equation}
    Lastly, the second term on the RHS of \eqref{inv_error_ineq} can be bounded in terms of $\epsilon_{\text{prep}}$: 
    \begin{equation}
        \alpha_{E_{\tel}} \norm{W^{-1} \ket{\{\hax_n\}} \ket{\psi'_0}\ket{0} - \ket{\{\hax_n\}}\ket{0}\ket{0}} \leq \alpha_{E_{\tel}} \sqrt{2\epsilon_{\text{prep}}}.
    \end{equation}

    Putting everything together we find that
    \begin{equation}
        \epsilon_{E_{\tel}} \leq \alpha_{E_{\tel}} \lb 2 \sqrt{2\epsilon_{\text{prep}}} + 2 \sqrt{\xi_{QPE}} \rb + \epsilon_{\text{alt}} + \epsilon_{QPE} + \epsilon_{\tel}.
    \end{equation}
    We can ensure $\epsilon_{E_{\tel}} \leq \epsilon$ by having
    \begin{align}
        \epsilon_{\text{prep}} &\leq \frac{1}{2} \lb \frac{\epsilon}{10 \alpha_{E_{\tel}}} \rb^2 \\
        \xi_{QPE} &\leq  \lb \frac{\epsilon}{10 \alpha_{E_{\tel}}} \rb^2 \\
        \epsilon_{\text{alt}} &\leq \frac{\epsilon}{5} \\
        \epsilon_{QPE} &\leq \frac{\epsilon}{{5}} \\
        \epsilon_{\tel} &\leq \frac{\epsilon}{5}.
    \end{align}
    The condition $\epsilon_{\text{prep}} \leq \frac{1}{2} \lb \frac{\epsilon}{10 \alpha_{E_{\tel}}} \rb^2$ can be satisfied by using
    \begin{equation}
        O \lb \frac{\lambda}{\gamma \delta} \log{\lb \frac{\lambda}{\delta \epsilon} \rb} \rb
    \end{equation}
    queries to $U_{H_{\tel}}$ and
    \begin{equation}
        O \lb \frac{1}{\delta} \rb
    \end{equation}
    queries to $U_I$.
    The conditions $\xi_{QPE} \leq  \lb \frac{\epsilon}{10 \alpha_{E_{\tel}}} \rb^2$ and $\epsilon_{QPE} \leq \frac{\epsilon}{5}$ can be satisfied by using
    \begin{equation}
        O \lb \frac{\lambda}{\epsilon} \log{\lb \frac{\lambda}{\epsilon} \rb} \rb
    \end{equation}
    queries to $Q$ (or equivalently, $U_{H_{\tel}}$).
    Next, the condition $\epsilon_{\text{alt}} \leq \frac{\epsilon}{5}$ can be satisfied by using
    \begin{equation}
        a_{E_{\tel}} \in O \lb \log{\lb \frac{\max_{\{ x_n\}} \widetilde{E}_{\tel}'\lb \{ x_n \} \rb }{\epsilon} \rb} \rb \in O \lb \log{\lb \frac{\lambda}{\epsilon} \rb} \rb
    \end{equation}
    ancilla qubits. The associated Toffoli cost is in $O \lb a_{E_{\tel}} \rb$ due to the inequality testing required for the alternating sign trick. Lastly, by Lemma \ref{lem:block-encode_Hel} we need 
    \begin{equation}
        O \lb N + \Tilde{N} + \log{\lb \frac{B}{\epsilon}\rb} \rb
    \end{equation}
    Toffoli gates to ensure that the block-encoding error $\epsilon_{\tel}$ of $H_{\tel}$ is at most $\epsilon/5$. 
    
    The overall Toffoli complexity of block-encoding $H_{E_{\tel}}$ is dominated by the number Toffolis needed for all queries to the walk operator $Q$ and the number of Toffolis needed for the electronic ground state preparation. In either case, we multiply the respective query complexity with the Toffoli cost of block-encoding $H_{\tel}$ to obtain the desired complexity expression.
\end{proof}

Let us now prove the following lemma on the query complexity of estimating the internal energy.

\begin{lem}[Query complexity of estimating the internal energy]
    Let $U_{H_{\text{kin}}}$ be an $\lb \alpha_{\text{kin}}, a_{\text{kin}}, \epsilon/9 \rb$-block-encoding of $H_{\text{kin}}$, let $U_{H_{\text{pot}}}$ be an $\lb \alpha_{\text{pot}}, a_{\text{pot}}, \epsilon/9 \rb$-block-encoding of $H_{\text{pot}}$ and let $U_{E_{\tel}}$ be an $\lb \alpha_{E_{\tel}}, a_{E_{\tel}}, \epsilon/9 \rb$-block-encoding of $H_{E_{\tel}}$. Furthermore, let $\widetilde{U}_{L_{NVT}}' \in \mathbb{C}^{\eta_{\text{pur}} \times \eta_{\text{pur}}}$ be an $\epsilon/\lb 18 \eta_{\text{pur}} \alpha_{\text{nuc}} \rb$-precise approximation to $U_{L_{NVT}}'$ as defined in Eq.~(\ref{purification}) where $\alpha_{\text{nuc}} := \alpha_{\text{kin}} + \alpha_{\text{pot}} + \alpha_{E_{\tel}}$.
    There exists a quantum algorithm for estimating the internal energy $\mathcal{U}$ associated with $\rho_{\text{sys}}'(t)$ within error $\epsilon$ with probability at least $1 - \xi$ using 
    \begin{equation}
        O \lb \frac{\alpha_{\text{nuc}}}{\epsilon}  \log{\lb \frac{1}{\xi} \rb} \rb
    \end{equation}
    queries to  $\widetilde{U}_{L_{NVT}}'$, $U_{H_{\text{kin}}}$, $U_{H_{\text{pot}}}$ and $U_{E_{\tel}}$.
\label{lem:internal_energy}
\end{lem}

\begin{proof}
    First, note that the internal energy $\mathcal{U}$ of the system can be computed as follows:
    \begin{equation}
        \mathcal{U} = \text{Tr} \lb \rho_{\text{sys}}'(t) H_{\text{nuc}} \rb = \text{Tr} \lb \rho_{\text{sys}}'(t) H_{\text{kin}} \rb + \text{Tr} \lb \rho_{\text{sys}}'(t) H_{\text{pot}} \rb + \text{Tr} \lb \rho_{\text{sys}}'(t) H_{E_{\tel}} \rb.
    \end{equation}
    The idea then is to estimate each term individually within error $\epsilon/3$ using the Hadamard test as shown in Fig.~\ref{fig:Hadamard_test}. Then we add the results classically which yields an $\epsilon$-precise estimate of $\mathcal{U}$.
    
    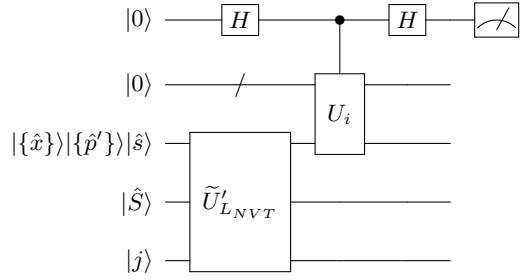
\begin{figure}
        \centering
        \mbox{
            \Qcircuit @C=1em @R=1.5em {
               \lstick{\ket{0}} & \gate{H} & \ctrl{1}  & \gate{H} & \qw & \meter \\
               \lstick{\ket{0}} & {/} \qw & \multigate{1}{U_{i}}  & \qw & \qw \\
               \lstick{\ket{\{\hax\}}\ket{\{ \hap' \}}\ket{\has}} & \multigate{2}{\widetilde{U}_{L_{NVT}}'} & \ghost{U_{i}}  & \qw & \qw \\
               \lstick{\ket{\overline{S}}} & \ghost{\widetilde{U}_{L_{NVT}}'} & \qw  & \qw & \qw \\
               \lstick{\ket{j}} & \ghost{\widetilde{U}_{L_{NVT}}'} & \qw  & \qw & \qw 
        } 
        }
        \caption{Circuit for implementing the Hadamard test to estimate the internal energy $\mathcal{U}$ of the nuclei where $\widetilde{U}_{L_{NVT}}'$ is an approximation to the evolution operator from Eq.~\eqref{purification} and $U_i \in \{U_{H_{\text{kin}}}, U_{H_{\text{pot}}}, U_{E_{\tel}} \}$.
        The second register from the top is the ancilla register needed for block-encoding $H_{\text{kin}}$, $H_{\text{pot}}$ or $H_{E_{\tel}}$. The fact that $U_{i}$ does not act on the bath register $\ket{\overline{S}}$ or the phase gradient ancilla register $\ket{j}$ can be understood as taking the partial trace over those registers when computing the probability of measuring the top qubit as 0.}
        \label{fig:Hadamard_test}
    \end{figure}

    Let 
    \begin{equation}
        \widetilde{H}_{i} := \alpha_{i} \lb \bra{0} \otimes \mathbb{1} \rb U_{i} \lb \ket{0} \otimes \mathbb{1} \rb
    \end{equation}
    be the Hamiltonian term block-encoded in $U_{i}$ where $U_i \in \{U_{H_{\text{kin}}}, U_{H_{\text{pot}}}, U_{E_{\tel}} \}$ and $\alpha_i \in \{ \alpha_{\text{kin}},  \alpha_{\text{pot}},  \alpha_{E_{\tel}} \}$.
    The probability of measuring the top qubit in Fig.~\ref{fig:Hadamard_test} as 0 is given by 
    \begin{equation}
        P''_i(0) := \frac{1}{2} \lb 1 + \frac{\text{Tr} \lb  \widetilde{\rho}_{\text{sys}}'(t) \widetilde{H}_{i} \rb}{\alpha_{i}} \rb,
    \end{equation}
    where $\widetilde{\rho}_{\text{sys}}'(t)$ is the reduced density matrix of our system obtained from $\widetilde{U}_{L_{NVT}}'$.
    Let $\hat{P}''_i(0)$ denote our estimate of $P''_i(0)$ based on the outcome of the Hadamard test.
    Furthermore, let
    \begin{equation}
        P'_i(0) := \frac{1}{2} \lb 1 + \frac{\text{Tr} \lb  \rho_{\text{sys}}'(t) \widetilde{H}_{i} \rb}{\alpha_{i}} \rb,
    \end{equation}
    denote the success probability of the Hadamard test when used with an error-free block-encoding of $U_{L_{NVT}}'$.
    Lastly, let
    \begin{equation}
        P_i(0) := \frac{1}{2} \lb 1 + \frac{\text{Tr} \lb  \rho_{\text{sys}}'(t) H_{i} \rb}{\alpha_{i}} \rb
    \end{equation}
    denote the success probability of the Hadamard test when applied to error-free block-encodings of $U_{L_{NVT}}'$ and $H_{i}$ where $H_i \in \{ H_{\text{kin}}, H_{\text{pot}}, H_{E_{\tel}}\}$.
    Estimating $P_i(0)$ within error $\frac{\epsilon}{6 \alpha_{i}}$ allows us to obtain an $\epsilon/3$-precise estimate of $\text{Tr} \lb \rho_{\text{sys}}'(t) H_{i} \rb$.
    
    By the triangle inequality, it holds that
    \begin{equation}
        |\hat{P}''_i(0) - P_i(0)| \leq |\hat{P}''_i(0) - P''_i(0)| + |P''_i(0) - P'_i(0)| + |P'_i(0) - P_i(0)|.
    \label{estimate_inequality}
    \end{equation}
    
    We now show that each term on the RHS of Eq.~(\ref{estimate_inequality}) is upper bounded by $\frac{\epsilon}{18 \alpha_{i}}$ which implies that 
    \begin{equation}
        |\hat{P}''_i(0) - P_i(0)| \leq \frac{\epsilon}{6 \alpha_{i}}
    \end{equation}
    as desired.
    
    The error associated with the last term stems from the block-encoding error of $H_{i}$.
    Let $\rho'_j$ denote the $j$-th eigenvalue of $\rho_{\text{sys}}'(t)$ and let $H_{i,j}^{-}$ denote the $j$-th eigenvalue of $\widetilde{H}_{i} - H_{i}$.
    Using von Neumann's trace inequality, we obtain the following bound:
    \begin{equation}
    \begin{split}
        |P'_i(0) - P_i(0)| &= \frac{1}{2 \alpha_{i}} \left|\text{Tr} \lb  \rho_{\text{sys}}'(t) \widetilde{H}_{i} \rb - \text{Tr} \lb  \rho_{\text{sys}}'(t) H_{i} \rb \right| = \frac{1}{2 \alpha_{i}} \left| \text{Tr} \lb  \rho_{\text{sys}}'(t) \lb \widetilde{H}_{i} - H_{i} \rb \rb \right| \\
        &\leq \frac{1}{2 \alpha_{i}} \sum_j \left| \rho'_j \right| \left| H_{i,j}^{-} \right|  \leq \frac{\epsilon}{18 \alpha_{i}}  \sum_j \left| \rho'_j \right| \\
        &\leq \frac{\epsilon}{18 \alpha_{i}},
    \end{split}
    \end{equation}
    where we used the fact that $\left| H_{i,j}^{-} \right| \leq \norm{\widetilde{H}_{i} - H_{i}} \leq \epsilon/9$ by choice of the block-encoding precision.

    The error associated with the second term on the RHS of Eq.~(\ref{estimate_inequality}) stems from the simulation error of $U_{L_{NVT}}'$. By assumption,
    \begin{equation}
        \norm{\widetilde{U}_{L_{NVT}}' - U_{L_{NVT}}'} \leq \frac{\epsilon}{18\eta_{\text{pur}} \alpha_{\text{nuc}}}.
    \end{equation}
    This implies that
    \begin{equation}
    \begin{split}
        \norm{\widetilde{U}_{L_{NVT}}' \ketbra{\Psi_0}{\Psi_0} \widetilde{U}_{L_{NVT}}'^{\dagger} - U_{L_{NVT}}' \ketbra{\Psi_0}{\Psi_0} U_{L_{NVT}}'^{\dagger}} &\leq \norm{\widetilde{U}_{L_{NVT}}' \ketbra{\Psi_0}{\Psi_0} \widetilde{U}_{L_{NVT}}'^{\dagger} - \widetilde{U}_{L_{NVT}}' \ketbra{\Psi_0}{\Psi_0} U_{L_{NVT}}'^{\dagger}} \\
        & \quad + \norm{\widetilde{U}_{L_{NVT}}' \ketbra{\Psi_0}{\Psi_0} U_{L_{NVT}}'^{\dagger} - U_{L_{NVT}}' \ketbra{\Psi_0}{\Psi_0} U_{L_{NVT}}'^{\dagger}} \\
        &\leq \norm{\widetilde{U}_{L_{NVT}}'^{\dagger} - U_{L_{NVT}}'^{\dagger}} + \norm{\widetilde{U}_{L_{NVT}}' - U_{L_{NVT}}'} \\
        &\leq \frac{2\epsilon}{18\eta_{\text{pur}} \alpha_{\text{nuc}}} = \frac{\epsilon}{9\eta_{\text{pur}} \alpha_{\text{nuc}}},
    \end{split}
    \end{equation}
    where, as before, $\ket{\Psi_0}$ is the initial state of the purification of $\rho_{\text{sys}}'(t)$.
    It then follows from Definition \ref{def:trace_dist} that 
    \begin{equation}
        \frac{1}{2}\norm{\widetilde{U}_{L_{NVT}}' \ketbra{\Psi_0}{\Psi_0} \widetilde{U}_{L_{NVT}}'^{\dagger} - U_{L_{NVT}}' \ketbra{\Psi_0}{\Psi_0} U_{L_{NVT}}'^{\dagger}}_1 \leq \frac{\epsilon}{18\alpha_{\text{nuc}}}.
    \end{equation}
    Since the trace distance is contractive under the partial trace, we obtain the following bound:
    \begin{equation}
          \mathcal{T}\lb \widetilde{\rho}_{\text{sys}}'(t), \rho_{\text{sys}}'(t) \rb \leq \frac{\epsilon}{18\alpha_{\text{nuc}}}.
    \end{equation}

    Let $\rho^{-}_j$ denote the $j$-th eigenvalue of $\widetilde{\rho}_{\text{sys}}'(t) - \rho_{\text{sys}}'(t)$ and let $\widetilde{H}_{i,j}$ denote the $j$-th eigenvalue of $\widetilde{H}_{i}$.
    Using von Neumann's trace inequality, we then have that
    \begin{equation}
    \begin{split}
        |P''_i(0) - P'_i(0)| &= \frac{1}{2 \alpha_{i}} \left|\text{Tr} \lb  \widetilde{\rho}_{\text{sys}}'(t) \widetilde{H}_{i} \rb - \text{Tr} \lb  \rho_{\text{sys}}'(t) \widetilde{H}_{i} \rb \right| = \frac{1}{2 \alpha_{i}} \left| \text{Tr} \lb \lb \widetilde{\rho}_{\text{sys}}'(t) -  \rho_{\text{sys}}'(t)\rb \widetilde{H}_{i} \rb\right| \\
        &\leq \frac{1}{2 \alpha_{i}} \sum_j \left| \rho^{-}_j \right| \left| \widetilde{H}_{i,j} \right| \leq \frac{1}{\alpha_{i}} \mathcal{T}\lb \widetilde{\rho}_{\text{sys}}'(t),  \rho_{\text{sys}}'(t) \rb \alpha_{i} \\
        &\leq \frac{\epsilon}{18 \alpha_{\text{nuc}}} \leq \frac{\epsilon}{18 \alpha_{i}}.
    \end{split}
    \end{equation}

    Lastly, we need to ensure that $|\hat{P}''_i(0) - P''_i(0)| \leq \frac{\epsilon}{18 \alpha_{i}}$.
    The idea is to use amplitude estimation to obtain an $\frac{\epsilon}{18 \alpha_{i}}$-precise estimate $\hat{P}''_i(0)$ of $P''_i(0)$ with constant success probability. This requires 
    \begin{equation}
        O \lb \frac{\alpha_{i}}{\epsilon} \rb \subseteq O \lb \frac{\alpha_{\nuc}}{\epsilon} \rb
    \end{equation}
    applications of the Hadamard test which is equivalent to $O \lb \frac{\alpha_{\nuc}}{\epsilon} \rb$ (controlled) applications of $\widetilde{U}_{L_{NVT}}'$ and $U_{i}$. 

    By the union bound, we can obtain an $\epsilon$-precise estimate of $\mathcal{U}$ with success probability $\geq 1 - \xi$ by ensuring that the failure probability associated with the estimation of each term $\text{Tr} \lb \rho_{\text{sys}}'(t) H_{i} \rb$ is $\leq \xi/3$. This can be achieved via (fixed-point) amplitude amplification at the expense of an additional multiplicative factor of $\log \lb 1/\xi \rb$ to the query complexities of $\widetilde{U}_{L_{NVT}}'$ and $U_i$.
\end{proof}

\subsection{Proof of Theorem \ref{thm:free_energy}}

For convenience, let us restate Theorem \ref{thm:free_energy} here.
\free*

\begin{proof}
    We use Lemma \ref{lem:gibbs_entropy} to estimate the Gibbs entropy associated with $\rho_{\text{sys}}(t)$ within error $\frac{\epsilon}{2 k_B T}$ with failure probability at most $\xi/2$. This requires
    \begin{equation}
        O \lb \frac{\eta \;  (k_B T)^{1.5}}{\epsilon^{1.5}} \log \lb 1/\xi \rb \rb
    \end{equation}
    queries to an $\lb \frac{\epsilon}{8 \eta_{\text{pur}} k_b T \log \lb \eta/\nu \rb} \rb$-precise approximation $\widetilde{U}_{L_{NVT}}'$ to $U_{L_{NVT}}'$ where $\nu \in (0,1)$ is again a lower bound on $2\mathcal{T}\lb \widetilde{\rho}_{\text{sys}}'(t), \rho_{\text{sys}}'(t) \rb$. For simplicity, we assume that $\nu \in \Theta \lb \epsilon/\eta \rb$ which should be fairly easy to achieve. The promise $\log \lb \eta^2/\epsilon \rb \leq \eta$ ensures that this choice of $\nu$ does not exceed the upper bound on $2\mathcal{T}\lb \widetilde{\rho}_{\text{sys}}'(t), \rho_{\text{sys}}'(t) \rb$.
    
    It then follows from Theorem \ref{thm:complexity_liouvillian} that estimating the entropy requires a total of
    \begin{equation}
        \widetilde{O} \lb \frac{\eta^{1+o(1)} N_{\text{tot}} d (k_B T)^{1.5 + o(1)} \mu_{NVT}^{2+o(1)} t^{1+o(1)}}{\widetilde{\gamma} \, \widetilde{\delta} \, \epsilon^{1.5 + o(1)}} \log \lb \frac{1}{\xi} \rb \rb
    \end{equation}
    Toffoli gates and
    \begin{equation}
        \widetilde{O} \lb \frac{\eta^{1+o(1)} N d (k_B T)^{1.5 + o(1)} \mu_{NVT}^{1+o(1)} t^{1+o(1)}}{\widetilde{\delta} \, \epsilon^{1.5 + o(1)}} \log \lb \frac{1}{\xi} \rb \rb
    \end{equation}
    queries to the initial electronic state preparation oracle $\widetilde{U}_I$.
    
    Next, we use Lemma \ref{lem:internal_energy} to estimate the internal energy $\mathcal{U}$ of the nuclei within error $\frac{\epsilon}{2}$ with failure probability at most $\xi/2$. This requires
    \begin{equation}
        O \lb \frac{\alpha_{\text{nuc}}}{\epsilon}  \log{\lb \frac{1}{\xi} \rb} \rb
    \end{equation}
    queries to an $\lb \frac{\epsilon}{36 \eta_{\text{pur}} \alpha_{\text{nuc}}} \rb$-precise approximation of $U_{L_{NVT}}'$. 
    By Theorem \ref{thm:complexity_liouvillian}, the associated Toffoli complexity is then in 
    \begin{equation}
        \widetilde{O} \lb \frac{\eta^{o(1)} N_{\text{tot}} d \, \alpha_{\text{nuc}} \, \mu_{NVT}^{2+o(1)} t^{1+o(1)}}{\widetilde{\gamma} \, \widetilde{\delta} \, \epsilon^{1+o(1)}} \log \lb \frac{1}{\xi} \rb \rb.
    \end{equation}
    Furthermore, we need
    \begin{equation}
        \widetilde{O} \lb \frac{\eta^{o(1)} N d \, \alpha_{\text{nuc}} \, \mu_{NVT}^{1+o(1)} t^{1+o(1)}}{\widetilde{\delta} \, \epsilon^{1+o(1)}} \log \lb \frac{1}{\xi} \rb \rb
    \end{equation}
    queries to the initial electronic state preparation oracle $\widetilde{U}_I$.
    According to Lemma \ref{lem:internal_energy} we also require
    \begin{equation}
        O \lb \frac{\alpha_{\text{nuc}}}{\epsilon} \log \lb \frac{1}{\xi} \rb \rb
    \end{equation}
    queries to $\epsilon/18$-precise block-encodings of $H_{\text{kin}}$, $H_{\text{pot}}$ and $H_{E_{\tel}}$. Lemmas \ref{lem:H_kin}, \ref{lem:H_pot} and \ref{lem:H_E_el} imply that the combined Toffoli complexity of all these queries is in
    \begin{equation}
    \begin{split}
        &\widetilde{O} \lb \frac{\alpha_{\text{nuc}}}{\epsilon} \log \lb \frac{1}{\xi} \rb \lb N \log \lb \frac{g \alpha_{\nuc}}{\epsilon} \rb + \log^{\log 3} \lb \frac{\alpha_{\nuc}}{\epsilon} \rb + N_{\text{tot}} \lambda \lb \frac{1}{\epsilon} + \frac{1}{\gamma \delta} \rb  \rb \rb \\
        &\subseteq \widetilde{O} \lb N_{\text{tot}} \alpha_{\text{nuc}} \lambda \log \lb \frac{1}{\xi} \rb \lb \frac{1}{\epsilon^2} + \frac{1}{\gamma \delta \epsilon} \rb \rb 
    \end{split}
    \end{equation}
    Similar to the spectral gap argument used in the proof of Theorem \ref{thm:complexity_liouvillian}, we only need to provide a lower bound $\widetilde{\gamma}$ on the spectral gap of the electronic Hamiltonian over those phase space grid points that are associated with a nonzero amplitude at some point during the simulation. 
    Likewise, we only need a lower bound $\widetilde{\delta}$ on the overlap of the initial electronic state with the true electronic ground state over phase space grid points that are visited at some point during the simulation.
    The reason for this is that any simulation errors that occur on grid points which are associated with zero amplitude throughout the simulation do not contribute to the error of the final estimate. This means that Problem \ref{prob:free_energy} can be solved using only $O \lb \frac{1}{\widetilde{\gamma} \widetilde{\delta}} \rb$ rather than $O \lb \frac{1}{\gamma \delta} \rb$ Toffoli gates.
    
    Combining the previous results, we find that the overall Toffoli complexity associated with estimating $\mathcal{F}$ is in 
   \begin{equation}
        \widetilde{O} \lb \lb \frac{\eta^{o(1)} N_{\text{tot}} d \, \alpha_{\text{nuc}} \, \mu_{NVT}^{2+o(1)} t^{1+o(1)}}{\widetilde{\gamma} \, \widetilde{\delta} \, \epsilon^{1+o(1)}}  + \frac{\eta^{1+o(1)} N_{\text{tot}} d (k_B T)^{1.5 + o(1)}  \mu_{NVT}^{2+o(1)} t^{1+o(1)}}{\widetilde{\gamma} \, \widetilde{\delta} \, \epsilon^{1.5 + o(1)}} + \frac{N_{\text{tot}} \alpha_{\text{nuc}} \lambda}{\epsilon^2} \rb  \log \lb \frac{1}{\xi} \rb \rb,
    \end{equation}
    which can be simplified as follows:
    \begin{equation}
        \widetilde{O} \lb \lb \frac{\eta^{o(1)} N_{\text{tot}} d \mu_{NVT}^{2+o(1)} t^{1+o(1)}}{\widetilde{\gamma} \, \widetilde{\delta} \, \epsilon^{1+o(1)}} \lb \alpha_{\nuc} + \frac{\eta (k_b T)^{1.5 + o(1)}}{\sqrt{\epsilon}} \rb + \frac{N_{\text{tot}} \alpha_{\text{nuc}} \lambda}{\epsilon^2} \rb \log \lb \frac{1}{\xi} \rb  \rb.
    \end{equation}
    Furthermore, the overall number of queries to $\widetilde{U}_I$ is in
    \begin{equation}
        \widetilde{O} \lb \frac{\eta^{o(1)} N d \, \alpha_{\text{nuc}} \, \mu_{NVT}^{1+o(1)} t^{1+o(1)}}{\widetilde{\delta} \, \epsilon^{1+o(1)}} \log \lb \frac{1}{\xi} \rb + \frac{\eta^{1+o(1)} N d (k_B T)^{1.5 + o(1)} \mu_{NVT}^{1+o(1)} t^{1+o(1)}}{\widetilde{\delta} \, \epsilon^{1.5 + o(1)}} \log \lb \frac{1}{\xi} \rb \rb,
    \end{equation}
    which can be rewritten as follows:
    \begin{equation}
        \widetilde{O} \lb \frac{ \eta^{o(1)} N d \mu_{NVT}^{1+o(1)} t^{1+o(1)}}{\widetilde{\delta} \, \epsilon^{1+o(1)}} \log \lb \frac{1}{\xi} \rb \lb \alpha_{\nuc} + \frac{\eta (k_b T)^{1.5 + o(1)}}{\sqrt{\epsilon}} \rb \rb.
    \end{equation}
\end{proof}

\section{Computational cost scaling of forward Euler integration}
\label{app:EulerCost}

Computational costs for the calculation of molecular forces on a fault-tolerant quantum computer are given in Ref.~\cite{Obrien2022}; however, the propagation of errors from an individual gradient estimate to the updated particle positions is not bounded.  In the following, we extend that result by giving an upper bound on the runtime of a molecular dynamics simulation algorithm where the quantum computer is used to compute the forces on the particles which are then used to update the nuclear positions on a classical computer with forward Euler's method.

The goal of Euler's forward method is to update the value of a variable $y(T)$ (e.g., the position of a nucleus) at a time $T+h$ using the derivative $y'(T)$ and the step $h$:
\begin{equation}
    y(T+h) = y(T)+ h y'(T).
\end{equation}
 
We define $y(j|y(j-1))$ as the value of the position after $j$ steps in the case of a perfect update (i.e., with the exact derivative) given the same perfectly updated position at the previous iteration. In contrast, we have $\tilde y (j|\tilde y(j-1))$ and $\tilde y(1) = \tilde y(1 | y(0))$ in the case of updates with an approximate derivative (i.e., with some error $\delta$ in the derivative calculation). Unless otherwise stated we assume that
\begin{align*}
    y_{j-1} &= y(j-1|y_{j-2}) \\
    y_{j-2} &= y(j-2|y_{j-3}) \\
    &\dots
\end{align*}

We now want to determine an upper bound on the error performed in the update after $N$ steps of updates with approximate derivatives, defined as the difference between this and the same updates computed with perfect derivatives: $|y(N|y_{N-1}) - \tilde y (N|\tilde y_{N-1})|$. 
Additionally, we consider the solutions to Newton's equation of motion for the updated variable $y$ and the approximate variable $\tilde{y}$: $y= e^{At} y_0$ and $\tilde y= e^{\tilde{A}t}  y_0$, given the initial condition $y_0$.
We call $y(j)=y_j$ the value of our variable after $j$ time steps. 
Furthermore, we are assuming that the error on the derivatives is such that $\norm{A-\Tilde{A}}\leq K_\textrm{Lips} $, where $K_\textrm{Lips}$ is the Lipschitz constant which is directly related to the norm of the differential operator.

For a single integration step of size $h$, if starting from the same previous value, we accumulate an error $|y(j|y_{j-1})-\tilde y(j|y_{j-1})|$ which is upper bounded by the error on the derivative estimation multiplied by the step:

\begin{equation}
    \begin{split}
    |y(j|y_{j-1})-\tilde y(j|y_{j-1})| &\leq \delta h, \\
        \tilde y (j|y_{j-1}) &\leq e^{K_\textrm{Lips}h} y_{j-1},\\
        \tilde y (j|\tilde y_{j-1}) &\leq e^{K_\textrm{Lips}h} \tilde y_{j-1}.
    \end{split}
\end{equation}
So we have:
\begin{equation}
   \begin{split}
    |y(j|y_{j-1})-\tilde y (j|\tilde y_{j-1})| &\leq\ |y(j|y_{j-1})- \tilde y (j|y_{j-1})|+|\tilde y (j| y_{j-1}) - \tilde y (j|\tilde y_{j-1})|\\
    &\leq \delta h + |\tilde y(j|y_{j-1})-\tilde y (j| \tilde y_{j-1})| \leq \delta h +e^{K_\textrm{Lips} h} |y_{j-1} - \tilde y_{j-1}| \\
    &\leq \delta h \left(1 +\sum_{n=1}^j e^{nK_\textrm{Lips}h} \right).
   \end{split} 
\end{equation}

After $N$ steps with $N=T/h$ (where $T$ is the total time) of Euler's forward method we get:
\begin{equation}
\begin{split}
    |y(N|y_{N-1})-\tilde y (N|\tilde y_{N-1})| &\leq \delta h \left(1 +\sum_{n=1}^N e^{nK_\textrm{Lips}h} \right) \leq \delta h \left( 1+ \frac{e^{K_\textrm{Lips}h (N+1)}-1}{e^{K_\textrm{Lips}h}-1}\right) \\ 
   &\leq \delta h \left( 1 + \frac{e^{K_\textrm{Lips}h} e^{K_\textrm{Lips}T}-1}{e^{K_\textrm{Lips}h}-1}\right) = \delta h \left( 1+\frac{e^{K_\textrm{Lips}T}-e^{-K_\textrm{Lips}h}}{1-e^{-K_\textrm{Lips}h}}\right) \\
   &\leq h \delta (1+ 2 e^{K_\textrm{Lips}T})\leq 3 h \delta e^{K_\textrm{Lips}T},
    \end{split}
\end{equation}
where we have chosen $K_\textrm{Lips}h \leq \ln(2)$ and $N=\frac{T}{h}$.

We want to make sure that the error is at most $\epsilon_\textrm{MD}$:
\begin{equation}
    h \delta (1+ 2 e^{K_\textrm{Lips}T})\leq 3 h \delta e^{K_\textrm{Lips}T} \leq \epsilon_\textrm{MD}.
\end{equation}
We choose the step size to be $h=\ln(2)/K_\textrm{Lips}$, so we have that our error on the single gradient estimation needs to be:
\begin{equation}
    \delta \leq \frac{\epsilon_\textrm{MD} K_\textrm{Lips}}{3 e^{K_\textrm{Lips}T} \ln(2)}.
\end{equation}

We want to impose this error $\delta$ as the target error in the single gradient estimation necessary to achieve an overall simulation error $\epsilon_\textrm{MD}$.
From Table IV of Ref.~\cite{Obrien2022} in the case of first quantized plane-waves, the time complexity is $T_\textrm{grad}=N_a^\frac{7}{2} \delta^{-1}$ (with the number of atoms $N_a$ considered proportional to the number of orbitals). So the  time complexity of a single gradient estimation to achieve the target error is:
\begin{equation}
    \textrm{ToffCount}_\textrm{grad} \in O \left( \frac{N_a^\frac{7}{2}e^{K_\textrm{Lips}T}}{K_\textrm{Lips} \epsilon_\textrm{MD}}\right).
\end{equation}
Since we need to perform $N=T/h=K_\textrm{Lips}T/\ln(2)$ steps, the total time is given by:
\begin{equation}
    \textrm{ToffCount}_\textrm{MD} \in O \left( N\frac{N_a^\frac{7}{2}e^{K_\textrm{Lips}T}}{K_\textrm{Lips} \epsilon_\textrm{MD}} \right)=
    O \left( T \frac{N_{a}^\frac{7}{2} e^{K_\textrm{Lips}T} }{ \epsilon_\textrm{MD}} \right).
\end{equation}
This means that the time for simulating a system with Euler's integration method scales exponentially with the simulation time $T$ while still polynomially with the other parameters. 

It is worth noting that while the underlying trajectories are potentially unstable, the overall probability density formed by an ensemble of such trajectories generically is not.  In particular, if we instead focused on the error in phase space density then this scaling would become polynomial if the shadowing theorem conditions hold~\cite{Edward2023}.  This suggests that the relative cost between this approach and our own may be somewhat deceptive; however, it is fair regardless to say that without assuming that we are interested in estimating a single particle trajectory that the number of Toffoli gates needed for an accurate simulation may scale exponentially with the evolution time in the worst case scenario.

\end{document}